\documentclass{article}

\usepackage{amsmath,amssymb,amsthm,fullpage,mathrsfs,pgf,tikz,caption,subcaption,mathtools,mathabx}
\usepackage[ruled,vlined,noresetcount]{algorithm2e}
\usepackage{amsmath,amssymb,amsthm,mathtools}
\usepackage{thmtools}
\usepackage[utf8]{inputenc} 
\usepackage[T1]{fontenc}    
\usepackage{hyperref}       
\usepackage{url}            
\usepackage{booktabs}       
\usepackage{amsfonts}       
\usepackage{nicefrac}       
\usepackage{microtype}      
\usepackage{times}
\usepackage{bbm}
\usepackage{enumitem}
\usepackage{xcolor}
\usepackage{cleveref}
\usepackage{bm}
\usepackage{float}
\usepackage[margin=1in]{geometry}

\newtheorem{theorem}{Theorem}[section]
\newtheorem{corollary}[theorem]{Corollary}

\newtheorem{lemma}[theorem]{Lemma}
\newtheorem{definition}[theorem]{Definition}

\newtheorem{remark}[theorem]{Remark}

\newtheorem{claim}[theorem]{Claim}

\newcommand{\FF}{\mathbb{F}}
\newcommand{\NN}{\mathbb{N}}
\newcommand{\RR}{\mathbb{R}}
\newcommand{\ZZ}{\mathbb{Z}}

\newcommand{\cC}{\mathcal{C}}
\newcommand{\cE}{\mathcal{E}}
\newcommand{\cG}{\mathcal{G}}
\newcommand{\cM}{\mathcal{M}}
\newcommand{\cP}{\mathcal{P}}
\newcommand{\cV}{\mathcal{V}}

\newcommand{\set}[1]{\left\{ #1 \right\}}

\DeclarePairedDelimiter{\abs}{\lvert}{\rvert}
\DeclarePairedDelimiter{\norm}{\lVert}{\rVert}

\newcommand{\interior}[1]{%
  {\kern0pt#1}^{\mathrm{o}}%
}

\DeclareMathOperator{\Ker}{ker}
\DeclareMathOperator{\Ima}{im}
\DeclareMathOperator{\supp}{supp}
\DeclareMathOperator{\rk}{rk}
\DeclareMathOperator{\Cay}{Cay}

\def\tensorcode{\varSigma}
\def\tannercode{\mathcal{T}}

\def\F2{\mathbb{F}_2} 

\newcommand{\nonl}{\renewcommand{\nl}{\let\nl\oldnl}}
\newcommand\remove[1]{{}}

\title{Good Quantum LDPC Codes with Linear Time Decoders}
\author{Irit Dinur\thanks{Department of Applied Math and Computer Science, The Weizmann Institute of Science. Email: \texttt{irit.dinur@weizmann.ac.il}.} \and Min-Hsiu Hsieh\thanks{Hong Hai Research Institute, Taipei. Email: \texttt{min-hsiu.hsieh@foxconn.com}.} \and Ting-Chun Lin\thanks{Department of Physics, University of California San Diego, CA, and Hong Hai Research Institute, Taipei. Email: \texttt{til022@ucsd.edu}.} \and Thomas Vidick\thanks{Department of Computing and Mathematical Sciences, California Institute of Technology, CA. Email: \texttt{vidick@caltech.edu}.}}

\begin{document}
\nocite{dinur2021locally}

\maketitle

\begin{abstract}
  We construct a new explicit family of good quantum low-density parity-check codes which additionally have linear time decoders.
  Our codes are based on a three-term chain
  $(\F2^{m\times m})^V \quad \xrightarrow{\delta^0}\quad (\F2^{m})^{E} \quad\xrightarrow{\delta^1} \quad\F2^F$
  where $V$ ($X$-checks) are the vertices, $E$ (qubits) are the edges, and $F$ ($Z$-checks) are the squares of a left-right Cayley complex, and where the maps are defined based on a pair of constant-size random codes $C_A,C_B:\F2^m\to\F2^\Delta$ where $\Delta$ is the regularity of the underlying Cayley graphs.
  
  One of the main ingredients in the analysis is a proof of an essentially-optimal robustness property for the tensor product of two random codes.
\end{abstract}

\section{Introduction}

Quantum error correction is an essential ingredient to achieve fault-tolerant quantum computation.
An important class of quantum codes relevant to fault-tolerance are quantum low-density parity-check (qLDPC) codes \cite{gottesman2013fault}.
These are codes whose checks act only on a constant number of qubits, and further each qubit is acted on only by a constant number of checks. 
This low connectivity is desirable because it reduces the chance for errors to spread  when checks are being measured for error correction.

Several families of qLDPC codes have been studied starting from Kitaev's toric code \cite{kitaev2003fault}, with increasing rate and distance \cite{tillich2013quantum, freedman2002z2, evra2020decodable, kaufman2020new, hastings2021fiber, panteleev2021quantum, breuckmann2021balanced}. Recently Panteleev and Kalachev \cite{panteleev2021asymptotically} gave the first construction of good qLDPC codes, i.e.\ qLDPC codes with constant rate and constant relative distance. A subsequent variation on their construction was given in~\cite{leverrier2022quantum}.\footnote{We recently learned that a similar result has been obtained independently by several other groups. A later revision of the paper will address similarities and differences with these concurrent works.} 

A natural question left open by the recent constructions of good qLDPC codes is the existence of efficient decoders for them. In this work, we give a new construction of qLDPC, which borrows many of the ingredients from~\cite{panteleev2021asymptotically} as well as ideas from the recent classical locally testable codes by Dinur, Evra, Livne, Lubotzky, and Mozes~\cite{dinur2021locally}, and show that our codes have linear time decoders. 

\begin{theorem} 
  For every $r \in (0, 1/2)$, there exist constants $\delta > 0$, $w \in \NN$
  and an explicit infinite family of quantum LDPC codes with maximum weight $w$, rate $r$, and relative distance $\delta$. Furthermore these codes are equipped with a linear time decoder that decodes up to linear distance.
\end{theorem}
Our codes are based on a three-term chain
\begin{equation}\label{eq:chain}
    (\F2^{m\times m})^V \quad \xrightarrow{\delta^0}\quad (\F2^{m})^{E} \quad\xrightarrow{\delta^1} \quad\F2^F \;.
  \end{equation} 
The chain is ordered ``geometrically'' by dimension, so that $V$ are the vertices, $E$ are the edges, and $F$ are the faces (squares) of a left-right Cayley complex. Informally, this complex has vertices labeled by elements $g$ of a finite group $G$, edges labeled by two sets of generators $A,B$ as $(g,ag)$ and $(g,gb)$ for $g\in G$, $a\in A$ and $b\in B$, and squares $(g,ag,gb,agb)$ labeled by pairs $(a,b)\in A\times B$. 
  The maps $\delta^0,\delta^1$ in~\eqref{eq:chain} are defined via a pair of base codes $C_A,C_B:\F2^m\to\F2^\Delta$ where $\Delta=|A|=|B|$. 
  \footnote{This current simplification has $0$ rate. To get positive rate, the base codes have different dimensions in the actual construction, $C_A\colon \F2^{m_a} \to \F2^{\Delta}$ and $C_B\colon \F2^{m_b} \to \F2^{\Delta}$ with $m_a \ne m_b$.}
  An advantage of the geometric ordering is that it may facilitate extending the chain to having more than three terms by going to higher dimensional geometric complexes. On the other hand, this kind of chain is asymmetric and therefore separate arguments are required for the analysis of the chain and co-chain.

Let us give an informal description of the chain map. Given a $0$-chain $c^0\in (\F2^{m\times m})^V$, such that $c(v)$ is an $m\times m $ bit matrix for each $v \in V$, let us compute $\delta^0(c^0)$ assuming that $c^0$ is supported on a single vertex $v$ (and this is extended linearly).  We first apply the encoding $C_A$ to each row of $c^0(v)$ separately to get a rectangular $m \times \Delta$ matrix, whose columns are now distributed among the $A$-edges neighboring $v$. Next, we apply the code $C_B$ to each column of $c(v)$ separately to get a rectangular $\Delta\times m$ matrix whose rows we distribute among the $B$-edges neighboring $v$. This is naturally interpreted as an element $c^1 = \delta^0(c^0) \in (\F2^m)^E$. 

Now given an arbitrary $1$-chain $c^1\in (\F2^{m})^E$, such that $c^1(e)$ is an $m$-bit vector for each edge $e\in E$, let us compute $\delta^1(c^1)$ assuming $c^1$ is supported on a single edge $e$  (and this is extended linearly). If $e$ is an $A$-edge then we compute the $C_B$ encoding of $c^1(e)$, getting a vector of $|B|=\Delta$ bits, which we distribute one per square containing $e$.  If $e$ is a $B$-edge then we compute the $C_A$ encoding of $c^1(e)$ and proceed similarly, adding the bits distributed to the same face modulo $2$.
The actual construction uses a $4$-fold cover of a left-right complex, so we have four types of vertices and edges, see details in Section \ref{sec:construction}. 

The algorithm of linear time decoding is based on local bit-flips or small-set flips.
The analysis of the distance of the code as well as of the decoding algorithm has two main components: expansion and robustness. The expansion arguments resemble previous works $\cite{panteleev2021asymptotically,dinur2021locally,leverrier2022towards}$. Technically, the key is analyzing the expansion of chains with a certain ``local minimality" condition. 
The second ingredient is a  {\em robustness} property for the pair of base codes $(C_A,C_B)$ (and their duals
$(C_A^\perp, C_B^\perp)$). 
Our second contribution in this work is a proof that two random codes are optimally robust.

A pair $(C_A, C_B)$ of codes, $C_A, C_B\subseteq \F2^n$, is said to be $d_2$-robust if for every pair of $n\times n$ matrices $M_A,M_B$  such that the rows of $M_A$ are in $C_A$ and the columns of $M_B$ are in $C_B$, if the matrix $M = M_A+M_B$ has low weight, then it can be decomposed into a sum of only a few rows in $C_A$ and a few columns in $C_B$, such that the number of rows and columns required is at most the weight of $M$ divided by the robustness parameter $d_2$. (See \Cref{sec:robust-tensor-codes} for formal definitions.)
Whereas previous works \cite{panteleev2021asymptotically} showed that random codes have robustness that is $d_2 = n^{\frac 1 2 + \epsilon}$, we show robustness with $d_2=\Theta(n)$ which is clearly best possible (up to multiplicative constants) since the weight of $M$ is quadratic in $n$ and the number of rows/columns is linear in $n$.
 
\begin{theorem} [Random Tensor Codes are Robust (Informal \Cref{thm:random-tensor-codes-are-robust})]
  For every $\rho_a, \rho_b \in (0, 1)$, there exist constants $\delta_1, \delta_2$
  such that 
  for $C_A, C_B$ sampled from the uniform distribution of linear codes of length $n$ and dimensions $\rho_a n, \rho_b n$,
  for large $n$,
  with high probability, 
  $C_A, C_B$ have distance $\delta_1 n$ and
  $(C_A, C_B)$ is $\delta_2 n$ robust.
\end{theorem}
Since the theorem is about random linear codes, it follows directly that  robustness holds simultaneously for both $(C_A,C_B)$ as well as $(C_A^\perp,C_B^\perp)$, with high probability.

The proof follows a counting argument similar to the proof of the Gilbert--Varshamov bound.
One defines certain words as `bad' and then shows that with high probability none of these `bad' words is a codeword through a union bound.
The additional twist is that one organizes the words by rank. See \Cref{sec:optimal-robust-tensor-codes} for details.

\remove{
\paragraph{Construction of the Chain Complex:}

Finally, we construct the chain complex using the Tanner codes of the left-right Cayley complex and the robust tensor codes.
Let $\cG_2 = \cG_2(G, A, B)$ be the 4-fold left-right Cayley complex with $|A|=|B|=\Delta$ and $C_A, C_B$ be the local codes of length $\Delta$.
The idea is to construct four Tanner codes and then combine them into a chain complex.
The four Tanner codes are
\begin{equation*}
  \tannercode(\cG(E^-, F), C_A)\colon \F2^F \to (\F2^{m_a})^{E^-},
\end{equation*}
\begin{equation*}
  \tannercode(\cG(E^|, F), C_B)\colon \F2^F \to (\F2^{m_b})^{E^|},
\end{equation*}
\begin{equation*}
  \tannercode(\cG(V, E^-), C_B)\colon (\F2^{m_a})^{E^-} \to (\F2^{m_a \times m_b})^V,
\end{equation*}
\begin{equation*}
  \tannercode(\cG(V, E^|), C_A)\colon (\F2^{m_b})^{E^|} \to (\F2^{m_a \times m_b})^V.
\end{equation*}
And the resulting chain complex is
\begin{equation*}
  X(\cG_2, C_A, C_B)\colon \F2^F \xrightarrow{\partial_2} (\F2^{m_a})^{E^-} \oplus (\F2^{m_b})^{E^|} \xrightarrow{\partial_1} (\F2^{m_a \times m_b})^V,
\end{equation*}
where 
\begin{equation*}
  \partial_2(c_2) = (\tannercode(\cG(E^-, F), H_A)(c_2), \tannercode(\cG(E^|, F), H_B)(c_2))
\end{equation*}
and
\begin{equation*}
  \partial_1(c^-_1, c^|_1) = \tannercode(\cG(V, E^-), H_B)(c^-_1) + \tannercode(\cG(V, E^|), H_A)(c^|_1).
\end{equation*}

\subsection{Proof Overview of Good Quantum LDPC Codes and Linear Time Decoder}

Given the construction of the 3-term chain complex, we now discuss how to show its expansion properties.
Because our construction is not symmetric, meaning that the chain complex and the co-chain complex have different structures, we have to show the two expansions separately.
Informally, we refer to the expansion of the chain complex as expansion and the expansion of the co-chain complex as co-expansion.
Under this point of view, the result of Dinur et al. \cite{dinur2021locally} can be understood as showing the expansion.
What is missing in their argument is the co-expansion.
In fact, the co-expansion is hidden in their proof.
What we will do is to split their argument into two parts, where one of them shows co-expansion, and the other shows expansion from co-expansion.
Note that in contrast to our constructions, the chain complexes in \cite{panteleev2021asymptotically} and \cite{leverrier2022quantum} are symmetric, so they only need one argument to show expansion in both directions.

To show expansion and co-expansion we rely on two theorems.
The first theorem shows co-expansion and the second theorem reduce the problem of expansion to co-expansion.
\begin{theorem} [Co-Expansion (Informal \Cref{thm:co-expansion})]
  Given $\Delta$-regular $\lambda$-spectral expander graphs $\Cay(G, A)$, $\Cay(G, B)$
  and linear codes $C_A^\perp, C_B^\perp$ of length $\Delta$ with distance $d_1$ and $(C_A^\perp, C_B^\perp)$ with robustness $d_2$.
  If $d_1 d_2 - \lambda d_2 - 8 \lambda \Delta > 0$,
  then the co-chain complex $X(\cG_2, C_A, C_B)$ has constant co-systolic distance.
\end{theorem}

\begin{theorem} [Co-Expansion $\to$ Expansion (Informal \Cref{thm:co-expansion-implies-expansion})]
  If $X(\cG_2, C_A^\perp, C_B^\perp)$ has constant co-systolic distance,
  then $X(\cG_2, C_A, C_B)$ has constant systolic distance.
\end{theorem}

A similar strategy is used to show decoder and co-decoder.

\begin{theorem} [Co-decoder (Informal \Cref{thm:co-decoder})]
  If $d_1 d_2/4 - \lambda d_2/2 - 8 \lambda \Delta > 0$,
  then $X(\cG_2, C_A, C_B)$ has a linear time co-decoder up to linear distance.
\end{theorem}

\begin{theorem} [Co-Decoder $\to$ Decoder (Informal \Cref{thm:co-expansion-implies-expansion2})]
  If $X(\cG_2, C_A^\perp, C_B^\perp)$ has a linear time co-decoder up to linear distance,
  then $X(\cG_2, C_A, C_B)$ has a linear time decoder up to linear distance.
\end{theorem}

Overall, the method of showing co-expansion and co-decoder is similar; and the method of showing expansion from co-expansion and decoder from co-decoder is similar.
We will describe these two cases separately
  and focus only on co-expansion and expansion.

First, we show co-systolic distance. 
Recall the statement for co-systolic distance,
given $c^1 \in Z^1 - B^1 \subset (\F2^{m_a})^{E^-} \times (\F2^{m_b})^{E^|}$, the goal is to show $\norm{c^1} = \Theta(|X(1)|)$.
Recall $\norm{c^1} = |\cE|$ where $\cE = \set{e \in E: c^1(e) \ne 0}$ is the non-zero edges.
The proof strategy is to count the number of ``neighbors'' between $\cE$.
We will define ``neighbors'' precisely in the main text, but for now, we only need to know that it is a local relation between the edges.
The expansion of the graph gives the upper bound on the number of ``neighbors'' between $\cE$
and the distance and the robustness of the local code gives the lower bound on the number of ``neighbors'' between $\cE$.
By comparing the two bounds, we obtain a lower bound $|\cE| = \Omega(|X(1)|)$.

In more detail, the upper bound is obtained from the expander mixing lemma.
Recall for a $\lambda$-spectral expander graph, the expander mixing lemma bounds the number of neighbors between the vertex sets $S, T\subset V$, $|E(S, T)| \le \frac{\Delta}{|V|} |S||T| + \lambda \sqrt{|S||T|}$ where $E(S, T)$ are the edges that have one end point in $S$ and one in $T$.
The lower bound is obtained by noticing if $c^1(e_{*0})$ is non-zero, then there are at least $d_1$ non-zero edges among 
its ``neighbors'' $E_{*1}(e_{*0}) \cup E_{0*}(e_{*0}) \cup E_{1*}(e_{*0})$
where for $e_{*0}=((g, ag), *0)$, we define
$E_{*1}(e_{*0}) = \{((gb, agb), *1): b \in B\}$, 
$E_{0*}(e_{*0}) = \{((g, gb), 0*): b \in B\}$, 
$E_{1*}(e_{*0}) = \{((ag, agb), 1*): b \in B\}$.
These are the edges that share a face with $e_{*0}$.
The reason that there are at least $d_1$ non-zero edges is because of the distance property of the local codes.
This already gives a lower bound on the number of ``neighbors'', but it is not strong enough to give useful result when combining with the upper bound.
To obtain a better lower bound, we additionally use robustness and local minimality.
This gives an additional relation which says 
if there are $x$ non-zero edges among $E_{*0}(v_{00}) \cup E_{0*}(v_{00})$,
then there are at least $d_2 x$ non-zero edges among
its ``neighbors'' $E_{*1}(v_{00}) \cup E_{1*}(v_{00})$
where for $v_{00}=(g, 00)$, we define
$E_{*0}(v_{00}) = \{((g, ag), *0): a \in A\}$,
$E_{*1}(v_{00}) = \{((gb, agb), *1): a \in A, b \in B\}$,
$E_{0*}(v_{00}) = \{((g, gb), 0*): b \in B\}$,
$E_{1*}(v_{00}) = \{((ag, agb), 1*): a \in A, b \in B\}$.
These are the edges that share a face with $v_{00}$.
Essentially, this provides another layer of lower bound on the number of ``neighbors''.
Finally, the complete lower bound is obtained by combining the two bound from distance and robustness.
This can be summarized into the following pictures.

\begin{figure}[H]
  \centering
  \begin{tikzpicture}
    \begin{scope} [scale=0.9]
      \draw (0, 0) -- (0, -3) -- (3.3, -3.6) -- (3.3, -0.6) -- cycle;
      \draw (0, 0) -- (0, -3) -- (3, -3) -- (3, 0) -- cycle;
      \draw (0, 0) -- (0, -3) -- (2.7, -2.4) -- (2.7, 0.6) -- cycle;

      \draw (0, 0) node [anchor=south east] {$v_{00}$};
      \draw (0, -3) node [anchor=north east] {$v_{10}$};
      \path (0, -1.5) node [anchor=east] {$e_{*0}$};
      \draw (3.4, -1.5) node [anchor=west] {$E_{*1}(e_{*0})$};
      \draw (1.5, 0.5) node [anchor=south] {$E_{0*}(e_{*0})$};
      \draw (1.5, -3.5) node [anchor=north] {$E_{1*}(e_{*0})$};
      \draw (1.5, -1.5) node {$F(e_{*0})$};
    \end{scope}

    \begin{scope} [scale=1.1, shift={(7,0)}]
      \draw (0, 0) -- (-0.2, -2.6) -- (2.4, -2.4) -- (2.6, 0.2) -- cycle;
      \draw (0, 0) -- (0, -3) -- (2.6, -2.8) -- (2.6, 0.2) -- cycle;
      \draw (0, 0) -- (-0.2, -2.6) -- (2.8, -2.6) -- (3, 0) -- cycle;
      \draw (0, 0) -- (0, -3) -- (3, -3) -- (3, 0) -- cycle;

      \draw (0, 0.0) node [anchor=south east] {$v_{00}$};
      \draw (0, -3) node [anchor=north east] {$V_{10}(v_{00})$};
      \draw (3, 0) node [anchor=south west] {$V_{01}(v_{00})$};
      \draw (3, -3) node [anchor=north west] {$V_{11}(v_{00})$};
      \path (-0.2, -1.5) node [anchor=east] {$E_{*0}(v_{00})$};
      \draw (1.5, 0.2) node [anchor=south] {$E_{0*}(v_{00})$};
      \draw (3, -1.5) node [anchor=west] {$E_{*1}(v_{00})$};
      \draw (1.5, -3) node [anchor=north] {$E_{1*}(v_{00})$};
      \draw (1.5, -1.5) node {$F(v_{00})$};
    \end{scope}
  \end{tikzpicture}


  \caption{Neighbors of an edge and a vertex.}
  \label{fig:neighbor-intro}
\end{figure}

\begin{figure}[H]
  \centering
  
  \begin{tikzpicture}
    \def\colorfrom{magenta}
    \def\colorto{cyan}
    \def\colorarrow{black}
    \def\colororig{red}
    \begin{scope} [scale=1]
      \def\lone{0.4}
      \def\ltwo{0.08}

      \draw[\colorfrom] (0, -\lone) -- (0, \lone-3);

      \draw[\colorto] (\lone, -\ltwo) -- (3-\lone, -\ltwo);
      \draw[\colorto] (\lone, 0) -- (3-\lone, 0);
      \draw[\colorto] (\lone, \ltwo) -- (3-\lone, \ltwo);

      \draw[\colorto] (3-\ltwo, -\lone) -- (3-\ltwo, -3+\lone);
      \draw[\colorto] (3, -\lone) -- (3, -3+\lone);
      \draw[\colorto] (3+\ltwo, -\lone) -- (3+\ltwo, -3+\lone);

      \draw[\colorto] (\lone, -3-\ltwo) -- (3-\lone, -3-\ltwo);
      \draw[\colorto] (\lone, -3) -- (3-\lone, -3);
      \draw[\colorto] (\lone, -3+\ltwo) -- (3-\lone, -3+\ltwo);

      \draw[\colorarrow] [->] (0,-1.5) -- (3,-1.5);
      \draw[\colorarrow] [->] (0,-1.5) to [bend right=40] (1.5,0);
      \draw[\colorarrow] [->] (0,-1.5) to [bend left=40] (1.5,-3);

      \fill[\colororig] (0,-1.5) circle[radius=0.1] node {};
    \end{scope}

    \begin{scope} [scale=1, shift={(7,0)}]
      \def\lone{0.4}
      \def\ltwo{0.04}

      \draw[\colorfrom] (-\ltwo, -\lone) -- (-\ltwo, \lone-3);
      \draw[\colorfrom] (\ltwo, -\lone) -- (\ltwo, \lone-3);

      \draw[\colorfrom] (\lone, -\ltwo) -- (3-\lone, -\ltwo);
      \draw[\colorfrom] (\lone, \ltwo) -- (3-\lone, \ltwo);

      \draw[\colorto] (3-3*\ltwo, -\lone) -- (3-3*\ltwo, -3+\lone);
      \draw[\colorto] (3-\ltwo, -\lone) -- (3-\ltwo, -3+\lone);
      \draw[\colorto] (3+\ltwo, -\lone) -- (3+\ltwo, -3+\lone);
      \draw[\colorto] (3+3*\ltwo, -\lone) -- (3+3*\ltwo, -3+\lone);

      \draw[\colorto] (\lone, -3-3*\ltwo) -- (3-\lone, -3-3*\ltwo);
      \draw[\colorto] (\lone, -3-\ltwo) -- (3-\lone, -3-\ltwo);
      \draw[\colorto] (\lone, -3+\ltwo) -- (3-\lone, -3+\ltwo);
      \draw[\colorto] (\lone, -3+3*\ltwo) -- (3-\lone, -3+3*\ltwo);

      \draw[rotate around={45:(0.75,-0.75)},\colorfrom] (0.75,-0.75) ellipse (1.5cm and 0.5cm);
      \draw[rotate around={45:(2.25,-2.25)},\colorto] (2.25,-2.25) ellipse (1.5cm and 0.5cm);

      \fill[\colororig] (0,0) circle[radius=0.1] node {};
      \draw[\colorarrow] [->] (0.8,-0.8) -- (2.2,-2.2);
    \end{scope}
  \end{tikzpicture}


  \caption{Lower bounds from distance and robustness.
    The edges appearing in the figures are the neighbors of an edges or a vertex as shown in \Cref{fig:neighbor-intro}.
    The black arrows indicate the two sides of the inequality.
    The number of non-zero edges in the blue region is greater than the number of non-zero edges in the pink region times the distance or robustness parameter.
  }
  \label{fig:co-expansion}
\end{figure}

Second, we show systolic distance from co-systolic distance.
Recall the statement for systolic distance,
given $c_1 \in Z_1 - B_1 \subset (\F2^{m_a})^{E^-} \times (\F2^{m_b})^{E^|}$, the goal is to show $\norm{c_1} = \Theta(|X(1)|)$.
Essentially, what we want to show is that if $c_1$ has small weight, then we show $c_1 \in B_1$ by finding $c_2$ such that $c_1 = \partial_2 c_2$.
This is done by utilizing the local exact chain complex around $v$, 
$(\F2)^{F(v)} \xrightarrow{\partial_2} (\F2^{m_b})^{E^|(v)} \oplus (\F2^{m_a})^{E^-(v)} \xrightarrow{\partial_1} \F2^{m_a \times m_b}$,
where $F(v)$ are the faces incident to $v$, $E^|(v)$ and $E^-(v)$ are the vertical and horizontal edges incident to $v$.
To find $c_2$, one make a guess $s_2(v) \in (\F2)^{F(v)}$ for each vertex such that $c_1 = \partial_2 c_2$ is satisfied locally, i.e. $c_1(E(v)) = \partial_2 s_2(v)$.
So by construction, $c_2$ and $s_2$ have similar properties locally.
The main difference between $c_2$ and $s_2$ is that for a face $F$, $c_2$ has a unique assignment on $F$, but for $s_2$ it depends on which vertex is viewed from.
To ``correct'' the difference between different views of $s_2$, we define the difference as $t_2$.
Note that $t_2$ has domain in $E$, because $t_2$ is comparing the view between two neighboring vertices.
What is interesting is the next step.
If one studies the property of $t_2$, one realizes that it resembles $c^1$ for the co-chain complex $X(\cG, C_A^\perp, C_B^\perp)$.
(For now we just note that $c^1$ also has domain in $E$.)
This allows us to use the co-systolic result and write $c^1 = \delta^0 c^0$.
Finally, using $c^0$, one can then ``correct'' the difference between the views of $s_2$ and recover $c_2$.
}
\subsection{Related work}

\paragraph{Quantum LDPC Codes and LTCs.}

Our work fits into a line of recent works on quantum LDPC codes and LTCs \cite{dinur2021locally,panteleev2021asymptotically,lin2022c, leverrier2022quantum}.
The constructions for qLDPC codes and LTCs turn out to be quite similar because both problems utilize 3-term chain complexes with expansion properties.
We focus on the history of quantum LDPC codes.
The historical development of LTCs can be found in \cite{goldreich2010short} and a more recent development can be found in \cite{dinur2021locally}.
More discussion of qLDPC can be found in \cite{breuckmann2021quantum}.

The earliest family of qLDPC codes are Kitaev's toric codes and surface codes \cite{kitaev2003fault} with dimension $k=\Theta(1)$ and distance $d=\Theta(\sqrt{n})$.
Over time, better codes with increasing rate \cite{tillich2013quantum} $k=\Theta(n)$ and distance \cite{freedman2002z2, evra2020decodable, kaufman2020new} $d=\Theta(\textnormal{polylog}(n) \sqrt{n})$ have been discovered.
Only recently did \cite{hastings2021fiber} and following works \cite{panteleev2021quantum, breuckmann2021balanced} significantly break the square root barrier and achieve $d=\Theta(n/\log n)$.
Finally, \cite{panteleev2021asymptotically} showed the existence of good quantum LDPC codes with $k=\Theta(n)$ and $d=\Theta(n)$.
More recently, \cite{leverrier2022quantum} provide another construction of good quantum LDPC codes.

\paragraph{Decoders for Quantum LDPC Codes.}

Finding an efficient decoder is often the next question after knowing the distance of a code.
If one does not worry about the efficiency, 
  in exponential time,
  it is known that one can decode up to $(d-1)/2$ errors by finding the closest codeword.
Practically, it is more desirable to have a polynomial time or even a linear time decoder.

Existing decoders can be broadly separated into two families that focus on different type of qLDPC codes. 
The first family mainly decodes the surface codes, 
  while the second family mainly decodes expander codes.
Because the code structure is different, the corresponding decoding strategy is also very different.
The first family includes minimum-weight perfect matching \cite{dennis2002topological}, union-find \cite{delfosse2021almost}, and variants of belief propagation decoders \cite{duclos2010fast, panteleev2019degenerate}.
A more complete discussion can be found in \cite{breuckmann2021quantum}.

We now focus on the second family, which the decoder of this paper belongs to.
When the underlying graph has good expansion properties, often the greedy algorithm that flips the bits locally will work.
This includes the classical expander codes \cite{sipser1996expander}
and the corresponding small set flip decoder in \cite{leverrier2015quantum} for quantum codes.
The same decoder was also applied to \cite{evra2020decodable, lin2022good}.
In this work, we use the small set flip decoder to decode the direction of the co-chain complex (i.e. decode $Z$ errors), and additionally use a ``reconstruction'' procedure to decode the direction of the chain complex (i.e. decode $X$ errors).

\paragraph{High Dimensional Expanders.}
Our work can be case as a study of notions of expansion in chain complexes. 
This relates to the study of high dimensional expanders (HDX), which is about notions of expansion for high-dimensional objects.
The study of the HDX was introduced by Linial and Meshulam \cite{linial2006homological} to study random simplicial complexes and independently by Gromov \cite{gromov2010singularities} to study the topological overlapping principle.
These natural questions have led to impressive results across areas including coding theory \cite{jeronimo2021near,dinur2021locally,panteleev2021asymptotically,lin2022c}, approximate sampling \cite{kaufman2020high,anari2019log,alev2020improved,anari2020spectral}, analysis of Boolean functions \cite{dikstein2018boolean,bafna2021hypercontractivity,gur2021hypercontractivity}, agreement testing \cite{dinur2017high,dikstein2019agreement}, and sum-of-square lower bounds \cite{dinur2020explicit,hopkins2022explicit}.

The most studied type of HDX are called simplicial complexes.
On the other hand, the recent development of qLDPC codes is more related to the cubical complexes.
It would be interesting to see if one can translate the results from one to the other.
One recent success is the application of qLDPC codes to sum-of-square lower bounds \cite{hopkins2022explicit}.

\subsection{Further directions}

\paragraph{Decoders for Other Quantum LDPC Codes.}
Our code construction has rate up to $1/2$,
whereas the code construction in \cite{panteleev2021asymptotically,leverrier2022quantum} have rate up to $1$.
It would be interesting to see if these constructions with rate up to $1$ have linear time decoders,
potentially using the robust code with better parameter constructed in this paper.

\paragraph{Towards Quantum LTCs.}
Related to qLDPC codes and LTCs are quantum locally testable codes (qLTCs) \cite{aharonov2015quantum}. Our proof technique may extend to the analysis of higher-length chain complexes, on which such qLTCs could be built. The main challenge towards this seems to be the absence of higher dimensional robust codes. On the topic of robustness, even for two-dimensional robustness we note that the current proofs only provide existence of robust codes via a probabilistic argument. 
It could be useful to have a direct, explicit construction as this may generalize more easily to higher dimensions than the probabilistic argument. 

\paragraph{PCPs and Quantum PCPs.}
Probabilistic checkable proofs (PCPs) and locally testable codes are closely, though not formally, related. (See \cite{goldreich2010short} for a survey.) In the quantum complexity literature there is a quantum version of PCPs \cite{aharonov2013guest}, the existence of which remains open. 
It would be interesting to see if one can make progress on this question by leveraging the recent works on qLDPCs.

\section{Preliminaries} \label{sec:prelim}

\subsection{Chain complexes} \label{sec:chain-complexes}

Chain complexes provide a way to connect the study of quantum codes with high dimensional expanders.

\begin{definition}[Chain complex]
  A chain complex $X$ is a sequence of vector spaces $\F2^{X(i)}$ generated by sets $X(i)$ together with linear maps $\partial_i\colon \F2^{X(i)} \to \F2^{X(i-1)}$ called the boundary operators.
  These boundary operators satisfy
  \begin{equation*}
    \partial_{i-1} \partial_i  = 0\;.
  \end{equation*}
\end{definition}

Because $\F2^{X(i)}$ has a canonical choice of basis corresponding to the elements of $X(i)$,
one can define the associated co-boundary operators $\delta^i\coloneqq \partial_{i-1}^T \colon \F2^{X(i)} \to \F2^{X(i+1)}$, where $(\cdot)^T$ denotes the matrix transpose.
The co-boundary operators automatically satisfy 
\begin{equation*}
  \delta^{i+1} \delta^i  = 0\;.
\end{equation*}
We introduce some standard terminology.
Elements of the kernel of the (co)-boundary operators are called (co)-cycles
\begin{equation*}
  Z_i \coloneqq \Ker \partial_i = \{c_i \in \F2^{X(i)} : \partial_i c_i = 0\}\;, \qquad 
  Z^i \coloneqq \Ker \delta^i = \{c^i \in \F2^{X(i)} : \delta^i c^i = 0\}\;.
\end{equation*}
Elements of the image of the (co)-boundary operators are called (co)-boundaries
\begin{equation*}
  B_i \coloneqq \Ima \partial_{i+1} = \{\partial_{i+1} c_{i+1} : c_{i+1} \in \F2^{X(i+1)}\}\;, \qquad
  B^i \coloneqq \Ima \delta^{i-1} = \{\delta^{i-1} c^{i-1} : c^{i-1} \in \F2^{X(i-1)}\}\;.
\end{equation*}
Since $\partial_i \partial_{i+1} = 0$ it follows that $B_i \subset Z_i$.
When $B_i = Z_i$ the chain complex is said to be \emph{exact} at $i$.

\subsection{Classical and quantum error correcting codes}
\label{sec:ecc}

A classical linear code is specified by a $k$-dimensional linear subspace $C \subset \F2^n$.
Here, $n$ is called the length, $k$ is called the dimension, and $d \coloneqq \min_{c \in C} \abs{c}$ is called the distance, where $\abs{\cdot}$ is the Hamming weight, i.e. the number of non-zero entries.
We call $r=k/n$ the rate and $\delta=d/n$ the relative distance of the code.
A more explicit way of describing a classical linear code is by specifying a parity-check matrix $H\colon \F2^n \to \F2^m$ where $m=n-k$ and $C = \Ker H$ is the kernel of the matrix.


A quantum CSS code is specified by two classical codes $C_z = \Ker H_z \subset \F2^n$ and $C_x = \Ker H_x \subset \F2^n$ such that $C_x^\perp \subset C_z$, i.e.\ $H_x H_z^T = 0$.
This condition allows us to associate a 3-term chain complex to the quantum code,
\begin{equation*}
  X\colon \F2^{m_z} \xrightarrow{H_z^T} \F2^n \xrightarrow{H_x} \F2^{m_x}\;.
\end{equation*}
Here are the relevant quantities associated with the quantum code.
Elements of $C_x=Z^1$ (resp.\ $C_z=Z_1$) are called $X$ (resp.\ $Z$)-logical operators. Elements of $C_x^\perp=B_1$ (resp.\ $C_z^\perp=B^1$) are called $Z$ (resp.\ $X$)-stabilizers.
The dimension of the code is $k = \dim Z_1 - \dim B_1$.
The distance is $d = \min(d_x, d_z)$ where 
\begin{equation*}
  d_x = \min_{c^1 \in Z^1 - B^1} |c^1|\;, \qquad
  d_z = \min_{c_1 \in Z_1 - B_1} |c_1|
\end{equation*}
and $d_x, d_z$ are called the $X$-distance and $Z$-distance of the code respectively. 
The code is called a \emph{quantum low-density parity-check code} (qLDPC) if $H_x$ and $H_z$ have a bounded number of nonzero entries in each column and row.

Having defined a quantum code, we now describe the task of decoding.
The goal of the decoder is to recover the error pattern from the syndrome.
Under the stabilizer formalism, one can express any error pattern as a pair $(c^1, c_1)$ 
  where $c^1 \in \F2^n$ indicates coordinates with an $X$-error
  and $c_1 \in \F2^n$ indicates coordinates with a $Z$-error.
The decoder is given the syndrome $(\delta^1 c^1, \partial_1 c_1)$
  and is required to return a correction $(\tilde c^1, \tilde c_1)$
  such that the difference from the actual error is a stabilizer,
  i.e.\ $\tilde c^1 - c^1 \in B^1$ and $\tilde c_1 - c_1 \in B_1$.
This task can be divided into two independent tasks
  where one recovers $\tilde c^1$ from $\delta^1 c^1$ ($X$-error decoding)
  and the other recovers $\tilde c_1$ from $\partial_1 c_1$ ($Z$-error decoding).


\subsection{Expander graphs} \label{sec:expander-graphs}

Expander graphs are used to obtain various important results in theoretical computer science. 
The most important one in our context is the expander codes \cite{sipser1996expander}.
We refer the reader to \cite{hoory2006expander} for other applications of expander graphs.

\begin{definition} [Spectral Expander Graphs and Ramanujan Graphs]
Let $\cG = (V, E)$ be an undirected, $\Delta$-regular graph on $n$ vertices, and define $\lambda(\cG) \coloneqq \max\{|\lambda_2|, |\lambda_n|\}$ where $\Delta = \lambda_1 \ge \lambda_2 \ge ... \ge \lambda_n$ are the eigenvalues of the adjacency matrix of $\cG$. We say that $\cG$ is a $\lambda$-spectral expander if $\lambda(\cG) \le \lambda$.
\end{definition}

We use spectral expanders for two reasons. 
First, there are known explicit infinite families of spectral expanders \cite{philips1988ramanujan}.
Second, spectral expansion implies edge expansion which is a key ingredient to obtain our results. This property is captured in the following expander mixing lemma which first appeared in \cite{alon1988explicit}.

\begin{lemma} [Expander Mixing Lemma]
  Let $\cG$ be a $\Delta$-regular graph with $\lambda$-spectral expansion. Then for any subset $S, T \subset V$, we have
  \begin{equation*}
    |E(S,T)| \le \frac{\Delta}{|V|} |S| |T| + \lambda \sqrt{|S||T|}\;.
  \end{equation*}

  Moreover, for any vectors $x,y \in \RR^V$ we have
  \begin{equation*}
    x^T M y \le \frac{\Delta}{|V|} \norm{x}_1 \norm{y}_1 + \lambda \norm{x}_2 \norm{y}_2\;,
  \end{equation*}
	 where $M$ is the adjacency matrix of $\cG$ and for $z\in\RR^n$, $\norm{z}_1 = \sum_i |z_i|$ and $\norm{z}_2 = (\sum_i |z_i|^2)^{1/2}$ denote the $L_1$ and $L_2$ norm respectively.
\end{lemma}

\subsection{Left-right Cayley complexes} \label{sec:left-right-cayley-complexes}


Our code construction is based on the left-right Cayley complex introduced in \cite{dinur2021locally}. A similar structure also appeared in \cite{breuckmann2021balanced, panteleev2021asymptotically}.
The \emph{4-fold left-right Cayley complex} $\cG_2(G, A, B)$ is specified by a finite group $G$ and two sets of generators $A$ and $B$ which are closed under inverse. 
The complex is illustrated in \Cref{fig:left-right-cayley-complex}. It consists of vertices, edges, and faces as follows: 
\begin{itemize}
  \item The vertices are $V = V_{00} \cup V_{10} \cup V_{01} \cup V_{11}$ where $V_{00} \cong V_{10} \cong V_{01} \cong V_{11} \cong G$.
  \item The edges are $E = E^| \cup E^- = (E_{*0} \cup E_{*1}) \cup (E_{0*} \cup E_{1*})$ where
    \begin{align*}
      E_{*0} &= \{(g, ag) : g \in G, a \in A\} \subset V_{00} \times V_{10}\;, \\
      E_{*1} &= \{(gb, agb) : gb \in G, a \in A\} \subset V_{01} \times V_{11}\;, \\
      E_{0*} &= \{(g, gb) : g \in G, b \in B\} \subset V_{00} \times V_{01}\;, \\
      E_{1*} &= \{(ag, agb) : ag \in G, b \in B\} \subset V_{10} \times V_{11}\;.
    \end{align*}
  \item The faces are $F = \{(g, ag, gb, agb): g \in G, a \in A, b \in B\} \subset V_{00} \times V_{10} \times V_{01} \times V_{11}$\;.
\end{itemize}

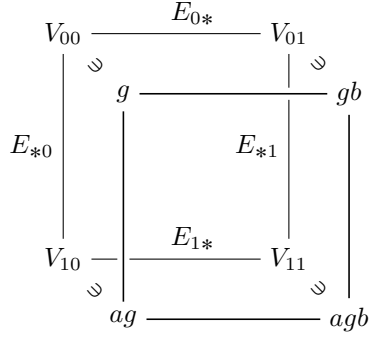
\begin{figure}
  \centering
  \begin{tikzpicture}
    \node (V00) at (0,0) {$V_{00}$};
    \node (V10) at (0,-3) {$V_{10}$};
    \node (V01) at (3,0) {$V_{01}$};
    \node (V11) at (3,-3) {$V_{11}$};

    \draw (V00) to node[auto, swap] {$E_{*0}$} (V10);
    \draw (V00) to node[auto, pos=0.6] {$E_{0*}$} (V01);
    \draw (V01) to node[auto, swap] {$E_{*1}$} (V11);
    \draw (V10) to node[auto, pos=0.6] {$E_{1*}$} (V11);

    \draw (0,0)+(0.8,-0.8) node (g) {$g$};
    \draw (0,-3)+(0.8,-0.8) node (ag) {$ag$};
    \draw (3,0)+(0.8,-0.8) node (gb) {$gb$};
    \draw (3,-3)+(0.8,-0.8) node (agb) {$agb$};

    \path (g) -- node [sloped] {$\ni$} (V00);
    \path (ag) -- node [sloped] {$\ni$} (V10);
    \path (gb) -- node [sloped] {$\ni$} (V01);
    \path (agb) -- node [sloped] {$\ni$} (V11);

    \draw [white, double=black, line width=2pt] (g) to (ag);
    \draw [white, double=black, line width=2pt] (g) to (gb);
    \draw [white, double=black, line width=2pt] (gb) to (agb);
    \draw [white, double=black, line width=2pt] (ag) to (agb);
  \end{tikzpicture}
  \caption{4-fold left-right Cayley complex.}
  \label{fig:left-right-cayley-complex}
\end{figure}

To clarify which vertex set, $V_{00}$, $V_{01}$, etc.\, a given vertex $g$ belongs to, we sometimes write the vertex as $(g,00)$ or $(g,01)$, etc. The same convention applies to edges. For example, $((g, ag), *0)$ is an edge in $E_{*0}$. 
Note that the edges and faces are labeled by ordered tuples instead of sets. Elements of $E^|$ are referred to as \emph{vertical edges}, and elements of $E^-$ as \emph{horizontal edges}. 
The appearance of faces crucially relies on the fact that the left action commutes with the right action, e.g.\ $a(gb)=(ag)b$.

We introduce the following important notation to describe the neighborhood relation between the vertices, edges and faces. For $v_{00}\in V_{00}$
we define $V_{10}(v_{00})$ as the set of vertices in $V_{10}$ neighbor to $v_{00}$
  and $V_{11}(v_{00})$ as the set of vertices in $V_{11}$ ``neighbor'' to $v_{00}$ by going through a horizontal edge and a vertical edge.
Similarly we define $E_{*0}(v_{00})$ as the set of edges in $E_{*0}$ incident to $v_{00}$ and $E_{1*}(v_{00})$ as the set of edges accessible by $v_{00}$ by first going through a vertical edge then choosing an adjacent horizontal edge.

More precisely, given $v_{00} = (g, 00)$ we define the following neighborhoods.
\begin{itemize}
  \item $V_{10}(v_{00}) = \{(ag, 10): a \in A\}$, $V_{01}(v_{00}) = \{(gb, 01): b \in B\}$, $V_{11}(v_{00}) = \{(agb, 11): a \in A, b \in B\}$,
  \item $E_{*0}(v_{00}) = \{((g, ag), *0): a \in A\}$, $E_{0*}(v_{00}) = \{((g, gb), 0*): b \in B\}$,
  \item $E_{*1}(v_{00}) = \{((gb, agb), *1): a \in A, b \in B\}$, $E_{1*}(v_{00}) = \{((ag, agb), 1*): a \in A, b \in B\}$,
  \item $E^|(v_{00}) = E_{*0}(v_{00})$, $E^-(v_{00}) = E_{0*}(v_{00})$, $E(v_{00}) = E^|(v_{00}) \cup E^-(v_{00})$,
  \item $F(v_{00}) = \{(g, ag, gb, agb): a \in A, b \in B\}$.
\end{itemize}
Given $e_{*0} = ((g, ag), *0)$, we define the following neighborhoods.
\begin{itemize}
  \item $E_{*1}(e_{*0}) = \{((gb, agb), *1): b \in B\}$,
  \item $E_{0*}(e_{*0}) = \{((g, gb), 0*): b \in B\}$, $E_{1*}(e_{*0}) = \{((ag, agb), 1*): b \in B\}$,
  \item $F(e_{*0}) = \{(g, ag, gb, agb): b \in B\}$.
\end{itemize}

\begin{figure}[H]
  \centering
  \begin{tikzpicture}
    \begin{scope} [scale=1.1]
      \draw (0, 0) -- (-0.2, -2.6) -- (2.4, -2.4) -- (2.6, 0.2) -- cycle;
      \draw (0, 0) -- (0, -3) -- (2.6, -2.8) -- (2.6, 0.2) -- cycle;
      \draw (0, 0) -- (-0.2, -2.6) -- (2.8, -2.6) -- (3, 0) -- cycle;
      \draw (0, 0) -- (0, -3) -- (3, -3) -- (3, 0) -- cycle;

      \draw (0, 0.0) node [anchor=south east] {$v_{00}$};
      \draw (0, -3) node [anchor=north east] {$V_{10}(v_{00})$};
      \draw (3, 0) node [anchor=south west] {$V_{01}(v_{00})$};
      \draw (3, -3) node [anchor=north west] {$V_{11}(v_{00})$};
      \path (-0.2, -1.5) node [anchor=east] {$E_{*0}(v_{00})$};
      \draw (1.5, 0.2) node [anchor=south] {$E_{0*}(v_{00})$};
      \draw (3, -1.5) node [anchor=west] {$E_{*1}(v_{00})$};
      \draw (1.5, -3) node [anchor=north] {$E_{1*}(v_{00})$};
      \draw (1.5, -1.5) node {$F(v_{00})$};
    \end{scope}

    \begin{scope} [scale=0.9, shift={(8.5,0)}]
      \draw (0, 0) -- (0, -3) -- (3.3, -3.6) -- (3.3, -0.6) -- cycle;
      \draw (0, 0) -- (0, -3) -- (3, -3) -- (3, 0) -- cycle;
      \draw (0, 0) -- (0, -3) -- (2.7, -2.4) -- (2.7, 0.6) -- cycle;

      \draw (0, 0) node [anchor=south east] {$v_{00}$};
      \draw (0, -3) node [anchor=north east] {$v_{10}$};
      \path (0, -1.5) node [anchor=east] {$e_{*0}$};
      \draw (3.4, -1.5) node [anchor=west] {$E_{*1}(e_{*0})$};
      \draw (1.5, 0.5) node [anchor=south] {$E_{0*}(e_{*0})$};
      \draw (1.5, -3.5) node [anchor=north] {$E_{1*}(e_{*0})$};
      \draw (1.5, -1.5) node {$F(e_{*0})$};
    \end{scope}
  \end{tikzpicture}
  \caption{(Left) The neighboring sets of a vertex $v_{00}$.
    (Right) The neighboring sets of an edge $e_{*0}$.
  }
  \label{fig:neighbor}
\end{figure}
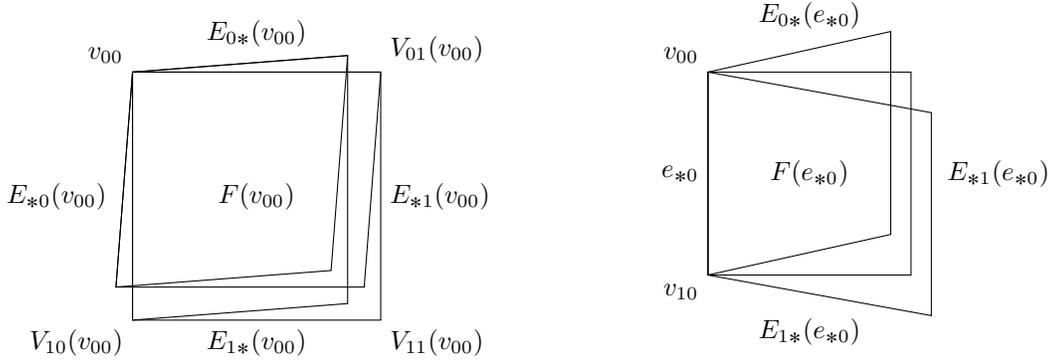

Finally we introduce subgraphs of the complex that will be used to define  Tanner codes in Section~\ref{sec:robust-tensor-codes}.
$\cG(E^|, F)$ is the bipartite graph that has $E^| = E_{*0} \cup E_{*1}$ as vertices and $F$ as the edges between them. 
More precisely, the edges are $F \cong \{((g, ag), (gb, agb)) : g \in G, a \in A, b \in B\} \subset E_{*0} \times E_{*1}$. The bipartite graph $\cG(E^-, F)$ is defined similarly. 
$\cG(V, E^|)$ is the bipartite graph that has $V = (V_{00} \cup V_{01}) \cup (V_{10} \cup V_{11})$ as vertices and $E^|$ as the edges between them. 
More precisely, the edges are $E^| = E_{*0} \cup E_{*1}$ where $E_{*0} \cong \{(g, ag) : g \in G, a \in A\} \subset V_{00} \times V_{10}$ and $E_{*1} \cong \{(g, ag) : g \in G, a \in A\} \subset V_{01} \times V_{11}$.
One defines the bipartite graph $\cG(V, E^-)$ similarly.

We conclude by discussing an explicit instance that is used in our construction. We use the Ramanujan graph constructed in \cite{philips1988ramanujan}.
Let $p$ and $q$ be unequal primes $\equiv 1 \mod 4$ and $\big(\frac{q}{p}\big) = 1$ where $\big(\frac{q}{p}\big)$ is the Legendre symbol.
Let $G = \textnormal{PSL}(2, \ZZ/q\ZZ)$ and $S=S^{-1}$ be the set of size $\Delta = p+1$ as defined in the paper.
The paper above shows that the Cayley graph $Cay(G, S)$ with vertex set $G$ and edge set $\{\{g, ag\}: g \in G, a \in S\}$ is a Ramanujan graph.
Finally, the 4-fold left-right Cayley complex we consider is $\cG_2(G, A=S, B=S)$.

\subsection{Expansion properties of left-right Cayley complexes}

We give three lemma that state expansion properties of operators defined on graphs obtained from the left-right Cayley complex. The first two lemma show expansion properties of two different random walks on the edges of $\cG_2(G, A, B)$.

\begin{lemma}\label{lem:M_1}
Let $M_1 \in \RR^{E\times E}$ be the adjacency matrix between opposing edges of the same face in $\cG_2(G, A, B)$, i.e.\ the adjacency matrix of the graph
    \begin{equation*}
      ((g, ag), *0) \sim ((gb, agb), *1)\;,\quad ((g, gb), 0*) \sim ((ag, agb), 1*)\qquad \forall g \in G, a \in A, b \in B\;.
    \end{equation*}
		Suppose that $\Cay(G,A)$ and $\Cay(G,B)$ are $\lambda$-spectral expanders. 
		Then for any subset $S\subseteq E$ it holds that 
		\begin{equation}\label{eq:M_1}
		1_S^T M_1 1_S \,\leq\, \lambda |S| + \frac{\Delta}{2|G|} |S|^2\;.
				\end{equation}
\end{lemma}
				
\begin{proof}
  $M_1$ is the disjoint union of $|G|$ copies of $\Cay^b(G, A)$ and $|G|$ copies of $\Cay^b(G, B)$ where $\Cay^b(G, A)$ and $\Cay^b(G, B)$ are the double covers of the $\lambda$-spectral expander graphs $\Cay(G, A)$ and $\Cay(G, B)$. 
  Let $S = \cup_i (S_i^0 \cup S_i^1)$ be a partition of $S$ according to each disjoint graph and their two vertex sets.
  Each disjoint graph satisfies,
  \begin{equation*}
    1_{S^0_i}^T M_1  1_{S^1_i} \le \lambda \sqrt{| {S^0_i}|| {S^1_i}|} + \frac{\Delta}{|G|} | {S^0_i}|| {S^1_i}|.
  \end{equation*}
  So
  \begin{align*} 
    1_S^T M_1 1_S
    &= 2 \sum_i 1_{S^0_i}^T M_1  1_{S^1_i} \nonumber \\
    &\le \lambda |S| + \frac{\Delta}{2|G|} |S|^2\;.
  \end{align*}
\end{proof}

\begin{lemma}\label{lem:M_0}
Let $M_0 \in \RR^{E\times E}$ be the adjacency matrix where two edges of $\cG_2(G, A, B)$ are connected if one of their endpoints are connected through an edge, i.e.\ 
    $M_0 = U M_0' D$ where $D\in \RR^{V\times E}$ and $U\in\RR^{E\times V}$ are the incidence matrices between the edges and the vertices
    and $M_0'$ is the adjacency matrix of the graph
    \begin{equation*}
      (g, 00) \sim (ag, 10),\, (g, 00) \sim (gb, 01),\, (ag, 10) \sim (agb, 11),\, (gb, 01) \sim (agb, 11)\quad \forall g \in G, a \in A, b \in B\;.
    \end{equation*}
			Suppose that $\Cay(G,A)$ and $\Cay(G,B)$ are $\lambda$-spectral expanders. 
		Then for any subset $S\subseteq E$ it holds that 
		  \begin{equation} \label{eq:M_0}
   1_S^T M_01_S  \,\leq\,  8\lambda \Delta |S| + \frac{2\Delta}{|G|} |S|^2\;.
	\end{equation}
		\end{lemma}
		
  Note that we allow multi-edges, so some entries of $M_0$ could be greater than $1$ when there are degeneracies.

\begin{proof}
  $M_0'$ is the union of two copies of $\Cay^b(G, A)$ and two copies of $\Cay^b(G, B)$. Let $\cV_{00} \in V_{00}$, $\cV_{10}\in V_{10}$, $\cV_{10}\in V_{10}$ and $\cV_{11}\in V_{11}$ be the vertices incident on $\cE$. 
  Because each edge is connected to two vertices,
  \begin{equation*}
    \norm{\cV_{00}}_1 + \norm{\cV_{10}}_1 + \norm{\cV_{01}}_1 + \norm{\cV_{11}}_1 \le 2 |\cE|.
  \end{equation*}
  Because each vertex is connected by at most $2\Delta$ edges,
    $\norm{\cV_{00}}_\infty \le 2\Delta$, so
  \begin{equation*}
    \norm{\cV_{00}}_2^2 + \norm{\cV_{10}}_2^2 + \norm{\cV_{01}}_2^2 + \norm{\cV_{11}}_2^2 \le 2|\cE| \cdot 2\Delta.
  \end{equation*}
  The expander subgraph $\Cay^b(G, A)$ between $\cV_{00}$ and $\cV_{10}$ gives
  \begin{align*}
    1_{\cV_{00}}^T M_0' 1_{\cV_{10}} 
    &\le \lambda \norm{\cV_{00}}_2 \norm{\cV_{10}}_2 + \frac{\Delta}{|G|} \norm{\cV_{00}}_1 \norm{\cV_{10}}_1 \\
    &\le \lambda \frac{\norm{\cV_{00}}_2^2 + \norm{\cV_{10}}_2^2}{2} + \frac{\Delta}{|G|} \norm{\cV_{00}}_1 \norm{\cV_{10}}_1.
  \end{align*}
  By combining with other expander subgraphs, we have
  \begin{align*} \label{eq:M_0}
    1_{\cE}^T M_0 1_\cE &= 2 (1_{\cV_{00}}^T  M_0' 1_{\cV_{10}} + 1_{\cV_{00}}^T M_0' 1_{\cV_{01}} + 1_{\cV_{10}}^T M_0' 1_{\cV_{11}} + 1_{\cV_{01}}^T M_0' 1_{\cV_{11}}) \nonumber \\
    &\le 2 \lambda (\norm{\cV_{00}}_2^2 + \norm{\cV_{10}}_2^2 + \norm{\cV_{01}}_2^2 + \norm{\cV_{11}}_2^2) \nonumber \\
    &+ \frac{2\Delta}{|G|} (\norm{\cV_{00}}_1 \norm{\cV_{10}}_1 + \norm{\cV_{00}}_1 \norm{\cV_{01}}_1 + \norm{\cV_{10}}_1 \norm{\cV_{11}}_1 + \norm{\cV_{01}}_1 \norm{\cV_{11}}_1) \nonumber \\
    &\le 8\lambda \Delta |\cE| + \frac{2\Delta}{|G|} |\cE|^2\;.
  \end{align*}
\end{proof}

The third lemma shows co-expansion of an associated graph. 

\begin{lemma} [Co-Expansion $\F2^{X(1)} \gets \F2^{X(0)}$] \label{lem:co-expansion-1d}
    Given $\Delta$-regular $\lambda$-spectral expander graphs $\Cay(G, A)$, $\Cay(G, B)$
    and linear codes $C_A^\perp, C_B^\perp$ of length $\Delta$ with distance $d_1$.
    
    Then the map
    \begin{equation*}
      (\F2^{m_a})^{E^-} \times (\F2^{m_b})^{E^|} \xleftarrow{\delta^0} (\F2^{m_a \times m_b})^V
    \end{equation*}
    satisfies
    \begin{equation*}
      \norm{\delta^0 c^0}_E \ge 2 (d_1-\lambda) \norm{c^0}_V - \frac{\Delta}{2} \frac{\norm{c^0}_V^2}{|G|}\;.
    \end{equation*}
  \end{lemma}
  
  \begin{proof}
    To show the expansion, one consider each component $c^1(E_{*0}) = \delta^0 c^0(V_{00}) + \delta^0 c^0(V_{10})$ separately.
    Because of code distance, each non-zero vertices in $V_{00}$ contribute to at least $d_1$ distinct non-zero edges in $\delta^0 c^0(V_{00})$.
    Same for $\delta^0 c^0(V_{10})$.
    What is left is to bound the number of cancellations in $\delta^0 c^0(V_{00}) + \delta^0 c^0(V_{10})$.
    Because $(V_{00}, V_{10}, E_{*0})$ is the double cover of the $\lambda$-spectral expander $\Cay(G, A)$,
      the number of cancellation is at most 
      $\lambda \sqrt{\norm{c^0(V_{00})}_V\norm{c^0(V_{10})}_V} + \frac{\Delta}{|G|} \norm{c^0(V_{00})}_V\norm{c^0(V_{10})}_V$.
    So
    \begin{align*}
      \norm{c^1(E_{*0})}_E 
      &\ge d_1 (\norm{c^0(V_{00})}_V + \norm{c^0(V_{10})}_V) - 2 (\lambda \sqrt{\norm{c^0(V_{00})}_V\norm{c^0(V_{10})}_V} + \frac{\Delta}{|G|} \norm{c^0(V_{00})}_V\norm{c^0(V_{10})}_V) \\
      &\ge (d_1-\lambda) (\norm{c^0(V_{00})}_V + \norm{c^0(V_{10})}_V) - \frac{2\Delta}{|G|} \norm{c^0(V_{00})}_V\norm{c^0(V_{10})}_V.
    \end{align*}

    Now we combine the four contributions and use AM-GM inequality to obtain
    \begin{align*}
      \norm{\delta^0 c^0}_E 
      &= \norm{c^1(E_{*0})}_E + \norm{c^1(E_{*0})}_E + \norm{c^1(E_{*0})}_E + \norm{c^1(E_{*0})}_E \\
      &\ge 2 (d_1-\lambda) \norm{c^0}_V - \frac{\Delta}{2} \frac{\norm{c^0}_V^2}{|G|}.
    \end{align*}
  \end{proof}

\subsection{Tensor codes and robustness} \label{sec:robust-tensor-codes}

Robust codes were first studied in \cite{ben2004robust} and \cite{dinur2006robust} in the context of locally testable codes (LTC).
Similar variants are applied to the construction LTC and qLDPC in \cite{dinur2021locally,panteleev2021asymptotically,leverrier2022quantum}.
In this paper, the definition of robustness is identical to agreement testability up to a normalization constant.
We first give the definition, then discuss its equivalence to agreement testability, and finally state our result stating robustness of the tensor product of random tensor codes.

Given 2 linear codes $C_A, C_B$ of length $n_a, n_b$ let $C_A \otimes C_B$ be the set of $n_a \times n_b$ matrices where each column vector belongs to $C_A$ and each row vector belongs to $C_B$.
Let $\tensorcode(C_A, C_B) \coloneqq C_A \otimes \F2^{n_b} + \F2^{n_a} \otimes C_B$ be the set of matrices that can be expressed as a sum of two $n_a \times n_b$ matrices, where the first has each column in $C_A$ and the second has each row in $C_B$. We introduce convenient notation for measuring different variations on the Hamming weight of a matrix: by entries, by rows, or by columns. 

\begin{definition}\label{def:hamming}
Given a matrix $c \in \F2^{n_a \times n_b}$, we let 
\begin{align*}
\norm{c}_{[n_a \times n_b]} &= \abs{\{(i, j) : f_{i,j} \ne 0\}}\;,\\
\norm{c}_{[n_b]} &= \abs{\{j : f_{\cdot,j} \ne 0\}}\;,\\
\norm{c}_{[n_a]} &= \abs{\{i : f_{i,\cdot} \ne 0\}}\;,
\end{align*}
\end{definition}

This definition allows us to introduce the notion of robustness we make use of.

\begin{definition} [Robustness of Tensor Codes] \label{def:robust}
 Let $C_A, C_B$ be linear codes of length $n_a, n_b$ respectively and $d_2 \in \RR_+$.
  We say that $(C_A, C_B)$ is $d_2$-robust if
  for all $c \in \tensorcode (C_A, C_B) \subset \F2^{n_a \times n_b}$,
    there exists $c_a \in C_A \otimes \F2^{n_b}$ and $c_b \in \F2^{n_a} \otimes C_B$
    such that $c = c_a + c_b$
    and 
  \begin{equation*}
    \norm{c}_{[n_a] \times [n_b]} \ge d_2 (\norm{c_a}_{[n_b]} + \norm{c_b}_{[n_a]})\;.
  \end{equation*}
\end{definition}

The notion of robustness can be understood as boundary expansion for a chain complex naturally associated with the pair of codes $(C_A,C_B)$. To see this define a $3$-term chain complex 
  \begin{equation}\label{lem:exact}
    Y(H_A, H_B) \colon \F2^{n_a \times n_b} \xrightarrow{\partial_2} \F2^{n_a \times m_b + m_a \times n_b} \xrightarrow{\partial_1} \F2^{m_a \times m_b}
  \end{equation}
	through the maps
	\begin{equation*}
    \partial_2(c_2) = ((I_{[n_a]} \otimes H_B) c_2, (H_A \otimes I_{[n_b]}) c_2)
  \end{equation*}
  and 
  \begin{equation*}
    \partial_1(c_1=(c_a, c_b)) = (H_A \otimes I_{[m_b]}) c_a + (I_{[m_a]} \otimes H_B) c_b\;,
  \end{equation*}
  where for an integer $k\geq 1$, $I_{[k]}$ denotes the identity map of $\F2^k$. 
  Then it follows easily from the K\"unneth formula (see e.g.~\cite[Section 3.B]{hatcher2002algebraic}) that $Y(H_A, H_B)$ is \emph{exact}, i.e.\ any element in the kernel of $\partial_1$ is  in the image of $\partial_2$. 

Now consider the co-chain
\begin{equation*}
  Y(H_A^\perp, H_B^\perp) \colon \F2^{n_a \times n_b} \xleftarrow{\delta^1} \F2^{n_a \times k_b + k_a \times n_b} \xleftarrow{\delta^0} \F2^{k_a \times k_b}\;,
\end{equation*}
where $H_A^\perp \colon \F2^{n_a} \to \F2^{k_a}$ is the parity check matrix of the dual code $C_A^\perp$. Using this complex, \Cref{def:robust} can be reformulated as saying that 
for all $c^2 \in \tensorcode (C_A, C_B) = \Ima \delta^1$,
  there exists $c^1 \in \F2^{n_a \times m_b + m_a \times n_b}$
  such that $c^2 = \delta^1 c^1$ and 
  \begin{equation*}
    \norm{c^2}_{[n_a] \times [n_b]} \ge d_2 \norm{c^1}_{[n_a] \cup [n_b]}
  \end{equation*}
where the variables $c, c_a, c_b$ from \Cref{def:robust} correspond to the new variables $c^2, ((H_A^\perp)^T \otimes I_{[n_b]}) c^1_a, (I_{[n_b]} \otimes (H_B^\perp)^T) c^1_b$ where $c^1 = (c^1_a, c^1_b) \in \F2^{m_a \times n_b} \oplus \F2^{n_a \times m_b}$.
Here we used the fact that $\norm{c^1_a}_{[n_b]} = \norm{((H_A^\perp)^T \otimes I_{[n_b]}) c^1_a}_{[n_b]}$ because $(H_A^\perp)^T$ is injective. 

The perspective through chain complexes allows us to make the connection with agreement testability. Note that the definition below differs from \cite[Definition 2.8]{dinur2021locally} by a normalization factor.

\begin{definition}[Agreement Testability] \label{def:agreement} 
Let $C_A, C_B$ be linear codes of length $n_a, n_b$ respectively and $d'_2\in\RR_+$.
  We say that $C_A \otimes C_B$ is $d'_2$-agreement testable if
  for all $c_a \in C_A \otimes \F2^{n_b}$, $c_b \in \F2^{n_a} \otimes C_B$,
    there exists $c \in C_A \otimes C_B$
    such that
  \begin{equation*}
    \norm{c_a + c_b}_{[n_a] \times [n_b]} \ge
    d'_2 (\norm{c + c_a}_{[n_b]} + \norm{c + c_b}_{[n_a]})\;.
  \end{equation*} 
\end{definition}

Using the same notation as above, \Cref{def:agreement} is saying that
for all $c^1 \in \F2^{n_a \times m_b + m_a \times n_b}$
  there exists $c^0 \in \F2^{m_a \times m_b}$
  such that 
  \begin{equation*}
    \norm{\delta^1 c^1}_{[n_a] \times [n_b]} \ge d'_2 \norm{c^1 + \delta^0 c^0}_{[n_a] \cup [n_b]}\;,
  \end{equation*}
where now $c$ in \Cref{def:agreement} corresponds to $((H_A^\perp)^T \otimes (H_B^\perp)^T) c^0$ and $c_a,c_b$ are as before. 
Because the chain complex $Y$ is exact, the two definitions are identical with $d_2 = d'_2$.

Finally, we state our result on the robustness of random tensor codes.
We consider the case when $n_a$ and $n_b$ are equal, $n_a = n_b = \Delta$.
In \cite{leverrier2022quantum} it is shown that for for arbitrary $\epsilon > 0$ and  
 $C_A$ and $C_B$ chosen uniformly at random, the pair $(C_A,C_B)$ is $\Omega(\Delta^{1/2-\epsilon})$-robust with high probability. Using a different counting argument we show that a uniformly random pair of codes is $\Theta(\Delta)$-robust with high probability.

\begin{restatable} [Random codes are robust]{theorem}{Robust} \label{thm:random-tensor-codes-are-robust}
  Fix $\rho_a, \rho_b \in (0,1)$, let $\delta_1 \in (0, 1/2), \delta_2 \in (0, \delta_1(1-\delta_1/2)/8)$ 
  satisfy
  \begin{equation} \label{eq:GV-2d}
    2 h(\delta_1/2) + 2(1-\delta_1/2) h(\frac{4 \delta_2}{\delta_1(1-\delta_1/2)}) < \frac{3}{4} \frac{(1 - \delta_1/2 - \rho_a)(1 - \delta_1/2 - \rho_b)}{1 - \delta_1/2}
  \end{equation}
  where $h(p) = -p \log_2(p) - (1-p) \log_2(1-p)$.\footnote{The allowed range for $\delta_2$ is chosen such that the argument in $h(\cdot)$ is valued between $(0, 1/2)$.}
  Let $C_A, C_B$ be random codes sampled from the uniform distribution with length $\Delta$ and dimensions $\rho_a \Delta, \rho_b \Delta$.
  Then as $\Delta$ goes to infinity, with probability tending to $1$,
    $C_A$, $C_B$ have distance $d_1 = \delta_1 \Delta$ and
    $(C_A, C_B)$ is $d_2 = \delta_2 \Delta$-robust.
\end{restatable}

The theorem is shown in \Cref{sec:optimal-robust-tensor-codes}.
When $C_A$ is sampled uniformly among codes of dimension $\rho_a \Delta$, $C_A^\perp$ is sampled uniformly among codes of dimension $(1-\rho_a) \Delta$. 
So the same theorem applies to $C_A^\perp$ and $C_B^\perp$
  and through a union bound we obtain the following corollary.

\begin{corollary} \label{cor:robust-code-exist}
  Fix $\rho_a, \rho_b \in (0,1)$. There exist constants $\delta_1$ and  $\delta_2$
  such that for large enough $\Delta$
  there exist codes $C_A$ and $C_B$ of length $\Delta$ where
  \begin{enumerate}
    \item $\dim C_A = \rho_a\Delta$ and $\dim C_B = \rho_b\Delta$,
    \item $C_A, C_B, C_A^\perp, C_B^\perp$ have distance $d_1 = \delta_1 \Delta $,
    \item $(C_A, C_B)$ and $(C_A^\perp, C_B^\perp)$ are both $d_2 = \delta_2 \Delta$-robust.
  \end{enumerate}
\end{corollary}

\subsection{Tanner codes} \label{sec:tanner-codes}


The Tanner construction \cite{tanner1981recursive} is a method to obtain `large' code by combining a `large' graph and a `small' local code.
This allows one to find an infinite family of codes by combining an infinite family of graphs with a fixed local code.
As long as the graphs are explicit, the Tanner codes are also explicit, even if finding the desired local code requires brute force search.
When the underlying graph is an expander, the Tanner code often inherits desireable properties from the small code.
Later, we will not only be interested in the code but also in the parity-check matrix that generates the code, since the LDPC property is defined on the parity-check matrix.
Therefore, we sometimes abuse language and refer to the code and the linear map (parity-check matrix) interchangeably.

We consider a $\Delta$-regular bipartite graph $\cG = (V_0, V_1, E)$ 
  with the $1-1$ identification $E \times [2] \cong V \times [\Delta]$, 
  where the additional index on the edge indicates whether it is asking for the vertex on the side of $V_0$ or $V_1$,
  and the additional index on the vertex gives an ordering to the edges incident to the vertex.
For example, 
  for the double cover of the Cayley graph with 
  $V_0 \cong V_1 \cong G$ and $E = \{(g, ag) : g \in G, a \in A\}$,
  a choice of the identification is 
  $(e=(g, ag), 0) \leftrightarrow (v_0=(g, 0), a)$
  and $(e=(g, ag), 1) \leftrightarrow (v_1=(ag, 1), a)$.

Given a $\Delta$-regular bipartite graph $\cG$ and a local code $C$ with parity-check matrix $H\colon \F2^\Delta \to \F2^m$,
  the Tanner code $\tannercode(\cG, H)\colon \F2^E \to (\F2^{m})^V$ is defined through the composition
  \begin{equation*}
    \F2^E \to (\F2^{\Delta})^V \to (\F2^{m})^V
  \end{equation*}
  where the first map copies the value on the edge to the vertices incident to the edge
  and the second map applies $H$ to $\F2^\Delta \cong \F2^{\{(v, a): a \in [\Delta]\}}$ for each vertex $v$.

Another way to think about the map is though its submatrices.
This description will be helpful to prove that the construction in \Cref{sec:linear-distance} is a chain complex.
Given an edge $e \in E$ and a vertex $v \in V$, 
  consider the submatrix $\tannercode(\cG, H)_e^v \colon \F2 \to \F2^{m}$
  which is the restriction where the input vector is supported on $e$ and the output vector is restricted to $v$.
Describing $\tannercode(\cG, H)$ is the same as  describing $\tannercode(\cG, H)_e^v$ for each $e \in E$ and $v \in V$.
When $v$ and $e$ are not incident, $\tannercode(\cG, H)_e^v$ is simply $0$.
When $v$ and $e$ are incident, and suppose $(e, i) \leftrightarrow (v, a)$,
  then $\tannercode(\cG, H)_e^v = H(\bar a) \colon \F2 \to \F2^{m}$, where $\bar a$ is the basis vector of $\F2^\Delta$ corresponding to the element $a \in [\Delta]$.

\subsection{Expansion properties of chain complexes} \label{sec:expansion-properties}


The distance of a quantum code falls into a broader category of expansion properties of chain complexes.
This includes (co)-systolic distance (the one equivalent to quantum code distance), small set (co)-boundary expansion \cite{hopkins2022explicit}, and (co)-locally minimal expansion \cite{kaufman2014ramanujan,evra2016bounded}.
We discuss them together because heuristically they are of similar difficulty, that is a proof that works for one often implies the other.
On the other hand, in certain scenario they can be distinguished.
For example, locally testable code does not follow directly from systolic distance,
  but does follow from small set boundary expansion.
This is one of the motivations for considering small set boundary expansion.
See \cite{lubotzky2014ramanujan} for the history and more discussions on the study of these expansion properties.

We first define the different notions of expansion, then we discuss relations between them and with code properties.
As we will discuss more precisely in Section~\ref{sec:linear-distance},
  we consider a weight on elements of a complex that is different from the Hamming weight, and 
  which counts the number of non-zero geometric objects instead of non-zero bits.
This weight is denoted as $\norm{\cdot}$ and differs from the usual Hamming weight by a constant factor, i.e.\ $\norm{\cdot} = \Theta(\abs{\cdot})$ (because the chain complex we consider has bounded degree).

\begin{definition} [(Co)-Systolic Distance]
  We say that $X: \F2^{X(2)} \xrightarrow{\partial_2} \F2^{X(1)} \xrightarrow{\partial_1} \F2^{X(0)}$ has systolic distance $\alpha$ if
  \begin{equation*}
    \forall c_1 \in Z_1 - B_1 : \norm{c_1} \ge \alpha |X(1)|.
  \end{equation*}

  Similarly, $X$ has co-systolic distance $\alpha$ if
  \begin{equation*}
    \forall c^1 \in Z^1 - B^1 : \norm{c^1} \ge \alpha |X(1)|.
  \end{equation*}
\end{definition}
It is not hard to see and well-known that constant (co)-systolic distance of a chain complex is equivalent to linear $X$-distance and $Z$-distance of the corresponding quantum CSS code.

\begin{definition} [Small-Set (Co)-Boundary Expansion]
  We say that $X: \F2^{X(2)} \xrightarrow{\partial_2} \F2^{X(1)} \xrightarrow{\partial_1} \F2^{X(0)}$ is a $(\alpha,\beta,\gamma)$-small-set boundary expander if
  \begin{equation*}
    \forall c_1 \in \F2^{X(1)}, \norm{c_1} < \alpha |X(1)|: 
    \exists c_2 \in \F2^{X(2)}, \norm{\partial_1 c_1} \ge \beta \norm{c_1 + \partial_2 c_2}, \norm{c_2} \le \gamma \norm{c_1}.
  \end{equation*}

  Similarly, $X$ is a $(\alpha,\beta,\gamma)$-small-set co-boundary expander if
  \begin{equation*}
    \forall c^1 \in \F2^{X(1)}, \norm{c^1} < \alpha |X(1)|: 
    \exists c^0 \in \F2^{X(0)}, \norm{\delta^1 c^1} \ge \beta \norm{c^1 + \delta^0 c^0}, \norm{c^0} \le \gamma \norm{c^1}.
  \end{equation*}
\end{definition}
We made a modification from \cite{hopkins2022explicit} by including a bound on $\norm{c_2}$ and $\norm{c^0}$. This additional bound is needed to show local testability.

\begin{definition} [(Co)-Locally Minimal]
  We say that $c_1 \in \F2^{X(1)}$ is locally minimal if
  \begin{equation*}
    \forall e_2 \in \F2^{X(2)}, \norm{e_2} = 1 : \norm{c_1} \le \norm{c_1 + \partial_2 e_2}.
  \end{equation*}

  Similarly, we say $c^1 \in \F2^{X(1)}$ is co-locally minimal if
  \begin{equation*}
    \forall e^0 \in \F2^{X(0)}, \norm{e^0} = 1 : \norm{c^1} \le \norm{c^1 + \delta^0 e^0}.
  \end{equation*}
\end{definition}

\begin{definition} [Small-Set (Co)-Locally-Minimal Expansion]
  We say that $X: \F2^{X(2)} \xrightarrow{\partial_2} \F2^{X(1)} \xrightarrow{\partial_1} \F2^{X(0)}$ is a $(\alpha,\beta)$-small-set locally-minimal expander if
  \begin{equation*}
    \forall c_1 \in \F2^{X(1)}\ \textnormal{s.t. $c_1$ is locally minimal and}\ \norm{c_1} < \alpha |X(1)|: \norm{\partial_1 c_1} \ge \beta \norm{c_1}.
  \end{equation*}

  Similarly, $X$ is an $(\alpha,\beta)$-small-set co-locally-minimal expander if
  \begin{equation*}
    \forall c^1 \in \F2^{X(1)}\ \textnormal{s.t. $c^1$ is locally minimal and}\ \norm{c^1} < \alpha |X(1)|: \norm{\delta^1 c^1} \ge \beta \norm{c^1}\;.
  \end{equation*}
\end{definition}

For our construction in Section~\ref{sec:linear-distance} we will show that the  chain complex has small-set co-locally-minimal expansion but not small-set locally-minimal expansion.
This is roughly because in our construction $X(2)$, $X(1)$, and $X(0)$ correspond to the faces, edges, and vertices.
So $e_2$ corresponds to a face and $e^0$ corresponds to a vertex.
Flipping $\partial_2 e_2$ only affects the four edges incident to the face,
whereas $\delta^0 e^0$ affects the $2 \Delta$ edges incident to the vertex.
Roughly, this means there are more freedom when flipping using $\delta^0 e^0$ than $\partial_2 e_2$.
This is the rationale for why the chain complex does not (seem to) have small-set locally-minimal expansion.

Given the definitions, we now discuss their relations.
The first lemma is between the expanders. The second and third lemma show that small-set boundary expansion implies systolic distance and local testability. 
  
\begin{lemma} [Small-Set (Co)-Locally-Minimal Expansion $\to$ Small-Set (Co)-Boundary Expansion] \label{lem:small-set-locally-minimal-implies-small-set-boundary}
  Let $c_2\in \F2^{X(2)}$ be such that $\norm{\partial_2 c_2} \le \mu \norm{c_2}$. Assume the gap between the possible values that $\norm{c_1}$ can take, for $c_1 \in \F2^{X(1)}$, is at least $\nu$ (i.e. $|\norm{c_1} - \norm{c'_1}| \ge \nu$ for any $c_1,c'_1$ such that $\norm{c_1} \ne \norm{c'_1}$.)
 
  If $X$ has $(\alpha, \beta)$-small-set locally-minimal expansion, then $X$ has $(\alpha/(1+\mu/\nu), \beta, 1/\nu)$-small-set boundary expansion.
\end{lemma}

The assumptions in the lemma often hold when the chain complex has bounded degree.

\begin{proof}
  Given $c_1$, consider the local flipping process of the decoder of the expander code \cite{sipser1996expander} which outputs $c_2$.

  \begin{algorithm}[H]
    \SetAlgoLined
    \begin{enumerate}
      \item (Initialization) $c_1^0 \coloneqq c_1$.
      \item (Main loop) In the $i$-th iteration, if there is $e_2^i$ with $\norm{e_2^i}=1$ such that $\norm{c_1^i + \partial_2 e_2^i} < \norm{c_1^i}$,
        set $c_1^{i+1} \coloneqq c_1^i + \partial_2 e_2^i$ and repeat.
      \item (End) Output $c_2 \coloneqq \sum e_2^i$.
    \end{enumerate}
    \caption{Local flip decoder. (Input: $c_1 \in \F2^{X(1)}$)}
    \label{alg:local-flip-decoder}
  \end{algorithm}
  
  We show that $c_2$ satisfies the desired properties: $\norm{\partial_1 c_1} \ge \beta \norm{c_1 + \partial_2 c_2}$ and $\norm{c_2} \le \gamma \norm{c_1}$.

  We first show $\norm{c_2} \le \gamma \norm{c_1}$.
  Because $\norm{c_2} \le \sum \norm{e_2^i}$ is bounded by the number of iterations,
  and each iteration reduces $\norm{c_1^i}$ by at least $\nu$,
  we have $\norm{c_2} \le 1/\nu \norm{c_1}$.

  We now show $\norm{\partial_1 c_1} \ge \beta \norm{c_1 + \partial_2 c_2}$
  Because the decoder cannot find $e_2$ and stops at $c_1 + \partial_2 c_2$, that means $c_1 + \partial_2 c_2$ is locally minimal.
  To apply small set locally minimal expansion, we suffice to show $c_1 + \partial_2 c_2$ has small size.
  Because $\norm{c_1 + \partial_2 c_2} \le \norm{c_1} + \norm{\partial_2 c_2} \le \norm{c_1} + \mu \norm{c_2} \le (1 + \mu/\nu) \norm{c_1}$,
  when $\norm{c_1} < \frac{\alpha}{1 + \mu/\nu} |X(1)|$, 
  $\norm{c_1 + \partial_2 c_2}$ satisfies the small set condition.
  Therefore, $\norm{\partial_1 c_1} \ge \beta \norm{c_1 + \partial_2 c_2}$.
\end{proof}

\begin{lemma} [Small-Set (Co)-Boundary Expansion $\to$ (Co)-Systolic Distance] \label{lem:small-set-boundary-implies-distance}
  If $X$ has $(\alpha, \beta, \gamma)$-small-set boundary expansion, then $X$ has systolic distance $\alpha$.
\end{lemma}

When the chain complex has bounded degree, this is equivalent to linear distance.

\begin{proof}
  Suppose $c_1 \in Z_1$ and $\norm{c_1} < \alpha |X(1)|$.
  Then by small set boundary expansion, there exists $c_2$, such that $0 = \norm{\partial_1 c_1} \ge \beta \norm{c_1 + \partial_2 c_2}$.
  This means $c_1 = \partial_2 c_2 \in B_1$.
  Therefore, for $c_1 \in Z_1 - B_1$ we have $\norm{c_1} \ge \alpha |X(1)|$.
\end{proof}

\begin{lemma} [Small-Set (Co)-Boundary Expansion $\to$ (Co)-Locally Testable Code] \label{lem:small-set-boundary-implies-local-testability}
  If $X$ has $(\alpha, \beta, \gamma)$-small-set boundary expansion, 
  then the classical code $C$ with parity check matrix $H = \partial_2\colon \F2^{X(2)} \to \F2^{X(1)}$ satisfies
  \begin{equation*}
    \norm{H v} \ge \min\Big(\frac{1}{\gamma}, \frac{\alpha |X(1)|}{|X(2)|}\Big) \min_{c \in C} \norm{v - c}.
  \end{equation*}
\end{lemma}

When the chain complex has bounded degree, this is equivalent to the condition for local testability.

\begin{proof} 
  Denote $c_2 = v$. Let $c_1 = \partial_2 c_2 \in Z_1$.
  When $\norm{c_1} < \alpha |X(1)|$, by small set boundary expansion, there exists $c'_2$, such that $0 = \norm{\partial_1 c_1} \ge \beta \norm{c_1 + \partial_2 c'_2}$ and $\norm{c'_2} \le \gamma \norm{c_1}$.
  This means $\partial_2 c'_2 = c_1$ and $\partial_2 (c_2 + c'_2) = 0$.
  That is $c \coloneq c_2 + c'_2 \in C$
  and $\gamma \norm{c_1} \ge \norm{c_2 - c}$.

  When $\norm{c_1} \ge \alpha |X(1)|$, we set $c=0$, and we have $\norm{c_1} \ge (\alpha |X(1)|/|X(2)|) \norm{c_2}$.
  Overall, we have $\norm{c_1} \ge \min(1/\gamma, \alpha |X(1)| / |X(2)|) \min_{c \in C} \norm{c_2 - c}$.
\end{proof}

\section{Linear dimension and linear distance} \label{sec:linear-distance}


We give our construction of a quantum code and show that it leads to a family of quantum LDPC codes with linear rate an distance.
Additionally, we show that the associated chain complexes have various kinds of good expansion properties.

\subsection{Construction}
\label{sec:construction}

Let $G$ be a finite group and $A$ and $B$ sets of generators for $G$ that are closed under inverse and have cardinality $|A|=n_a$, $|B|=n_B$. Throughout we assume that $n_A=n_B$ and write $\Delta = n_A = n_B$. 
The construction uses Tanner codes over 
  the 4-fold left-right Cayley complex $\cG_2(G, A, B)$ with $|A|=|B|=\Delta$ and 
  local tensor codes $C_A, C_B$ with parity-check matrices $H_A\colon \F2^\Delta \to \F2^{m_a}$, $H_B\colon \F2^\Delta \to \F2^{m_b}$.
The idea is to construct four Tanner codes and then combine them into a chain complex. 
We use the graphs $\cG(E^-, F), \cG(E^|, F), \cG(V, E^-), \cG(V, E^|)$ induced from the left-right Cayley complex defined in \Cref{sec:left-right-cayley-complexes}
and the Tanner code construction in \Cref{sec:tanner-codes}.
The four Tanner codes we make use of are
\begin{equation*}
  \tannercode(\cG(E^-, F), H_A)\colon \F2^F \to (\F2^{m_a})^{E^-},
\end{equation*}
\begin{equation*}
  \tannercode(\cG(E^|, F), H_B)\colon \F2^F \to (\F2^{m_b})^{E^|},
\end{equation*}
\begin{equation*}
  \tannercode(\cG(V, E^-), H_B)\colon (\F2^{m_a})^{E^-} \to (\F2^{m_a \times m_b})^V,
\end{equation*}
\begin{equation*}
  \tannercode(\cG(V, E^|), H_A)\colon (\F2^{m_b})^{E^|} \to (\F2^{m_a \times m_b})^V.
\end{equation*}
To clarify the notation we explicitly spell out the map $\tannercode(\cG(E^-, F), H_A)$.
By the definition of the Tanner construction, this map is the composition
  \begin{equation*}
    \F2^F \to (\F2^{\Delta})^{E^-} \to (\F2^{m_a})^{E^-}
  \end{equation*}
  where the first map copies the value on the face to the horizontal edges incident to the face (each horizontal edge is incident to $|A|=\Delta$ faces, so each horizontal edge is valued in $\F2^{\Delta}$)
  and the second map applies $H_A$ to $\F2^\Delta$ for each horizontal edge.

The resulting chain complex is
\begin{equation}\label{eq:complex-1}
  X\colon \F2^F \xrightarrow{\partial_2} (\F2^{m_a})^{E^-} \oplus (\F2^{m_b})^{E^|} \xrightarrow{\partial_1} (\F2^{m_a \times m_b})^V,
\end{equation}
where 
\begin{equation*}
  \partial_2(c_2) = (\tannercode(\cG(E^-, F), H_A)(c_2), \tannercode(\cG(E^|, F), H_B)(c_2))
\end{equation*}
and
\begin{equation*}
  \partial_1(c^-_1, c^|_1) = \tannercode(\cG(V, E^-), H_B)(c^-_1) + \tannercode(\cG(V, E^|), H_A)(c^|_1)
\end{equation*}
where $c_2 \in \FF_2^F, c^-_1 \in \FF_2^{E^-}, c^|_1 \in \FF_2^{E^|}$.

\begin{figure}
  \centering
  \begin{tikzpicture}
    \draw (0,0)node(F){$\F2^F$};
    \draw (4,0)node(Ev){$(\F2^{m_b})^{E^|}$};
    \draw (0,-2)node(Eh){$(\F2^{m_a})^{E^-}$};
    \draw (4,-2)node(V){$(\F2^{m_a \times m_b})^V$};

    \draw[->] (F) --node[auto,swap] {$\tannercode(\cG(E^-, F), H_A)$} (Eh);
    \draw[->] (F) --node[auto] {$\tannercode(\cG(E^|, F), H_B)$} (Ev);
    \draw[->] (Eh) --node[auto,swap] {$\tannercode(\cG(V, E^-), H_B)$} (V);
    \draw[->] (Ev) --node[auto] {$\tannercode(\cG(V, E^|), H_A)$} (V);
  \end{tikzpicture}
  \caption{The chain complex as a composition of the Tanner codes.}
  \label{fig:the-chain-complex}
\end{figure}

We denote this chain complex as $X(\cG_2, C_A, C_B)$, where $\cG_2$ is a shorthand for $\cG_2(G, A, B)$.
(Later in the analysis we also consider the chain complex $X(\cG_2, C_A^\perp, C_B^\perp)$ with the same graph but a different local code.) We use $\cC(\cG_2,C_A,C_B)$ to denote the associated quantum CSS code (see Section~\ref{sec:ecc}), and often write only $\cC$ for simplicity.

We end this section by commenting on the way to obtain an explicit family of groups and generating sets that satisfy all the expansion properties required for the quantum code $\cC$ to have linear distance and linear-time decoding, as shown in the following sections.
This relies on having an explicit construction of large Ramanujan graphs \cite{philips1988ramanujan} and the existence of (at least) one good local code pair \Cref{cor:robust-code-exist}.
First, we discuss the graph.
The graphs depend on the group $G$ and generators $A, B$.
The group $G$ belongs to an infinite family of groups with generators $A, B$ of fixed size $\Delta$ such that $\Cay(G, A)$, $\Cay(G, B)$ are $\lambda=2 \sqrt{\Delta-1}$-spectral expanders.
Second, we discuss the base codes. As shown in Section~\ref{sec:distance}, to show constant systolic and co-systolic distance we need $(C_A, C_B)$ and its dual code $(C_A^\perp, C_B^\perp)$ to have distance $d_1$ and robustness $d_2$ satisfying $d_1 d_2 - \lambda d_2 - 8 \lambda \Delta > 0$.
From \Cref{cor:robust-code-exist} we know that for fixed $\rho_a, \rho_b$ there exist constants $\delta_1, \delta_2$ such that for large enough $\Delta$, $C_A, C_B, C_A^\perp, C_B^\perp$ have distance $\delta_1 \Delta$ and $(C_A, C_B), (C_A^\perp, C_B^\perp)$ have robustness $\delta_2 \Delta$.
Because of the scaling $\lambda = \Theta(\Delta^{1/2}), d_1 = \Theta(\Delta), d_2 = \Theta(\Delta)$, for some large but fixed $\Delta$, there exists a good local code pair $(C_A, C_B)$.
This good code pair can be found by brute forcing all the possible code pairs.
Because $\Delta$ is fixed, the family of chain complexes remains explicit.

\subsection{Notation}

The following important notations are used for the analysis. First, we describe the notation that extracts the local structure. Given $c_2 \in \F2^F$, we denote $c_2(f) \in \F2$ as the value of $c_2$ at $f \in F$. Similarly, for $c_1 \in (\F2^{m_a})^{E^-} \oplus (\F2^{m_b})^{E^|}$ and $c_0 \in (\F2^{m_a \times m_b})^V$, one has $c_1(e^-) \in \F2^{m_a}$ for $e^- \in E^-$, $c_1(e^|) \in \F2^{m_b}$ for $e^| \in E^|$, and $c_0(v) \in \F2^{m_a \times m_b}$ for $v \in V$.
We also write $c^1(E_{*0}(V_{00})) \in \F2^{m_a \times n_a}$ to denote the entries on $E_{*0}(V_{00})$, where recall that this set is defined in Section~\ref{sec:left-right-cayley-complexes}. Notice that $E_{*0}(V_{00})$ contains $n_a$ edges and each edge gives a vector of size $m_a$. 

Second, we describe notation for measuring the size, or norm, of elements of the complex $X$. The norm is defined as the number of non-zero geometric objects, i.e. $\norm{c_2}_F = |\set{f \in F: c_2(f) \ne 0}|$, $\norm{c_1}_E = |\set{e \in E: c_1(e) \ne 0}|$, $\norm{c_0}_V = |\set{v \in V: c_0(v) \ne 0}|$.
We also write $\norm{c_1(E_{*0}(v_{00}))}_E = |\set{e \in E_{*0}(v_{00}): c_1(e) \ne 0}|$.

An element $c_2 \in \F2^F$ is usually indexed by $F$, leading to the norm $\|c_2\|_F$ as defined above, it can also naturally be indexed by $E_{*0}$ through $c_2(e_{*0}) = c_2(F(e_{*0}))$. This allows us to define $\norm{c_2}_{E_{*0}}$. The difference between the two norms is analogous to the difference between the different variants of the Hamming norm introduced in Definition~\ref{def:hamming}. Similarly, an element $s^2 \in  (\F2^{n_a})^{E^-} \times (\F2^{n_b})^{E^|}$ is indexed by $E$, but notice that for any $e_{*0}\in E_{*0}$, $s^2(e_{*0})$ can be indexed by $F(e_{*0})$.
This allows us to write $s^2(e_{*0}, f) \in \F2$ for $f \in F(e_{*0})$.
This leads to the definition $EF = E^-F \cup E^|F = \set{(e, f) \in E \times F: f \in F(e)}$, where $E^-F$ and $E^|F$ specialize to horizontal and vertical edges.
So $s^2$ can be indexed by $EF$ and this leads to the norm $\norm{s^2}_{EF} = |\set{(e, f) \in EF: s^2(e, f) \ne 0}|$. We will also write $\norm{s^2(E^{*0})}_F$ for $\norm{s^2(E^{*0})}_{EF}$; this is because when the edges are restricted to $E^{*0}$ we have $E^{*0}F \cong F$. One can similarly define $VE$, $VF$ and their corresponding norms.

Finally, the last notation we discuss is with regard to $H_A$ and $H_B$.
By thinking of $F$ as being indexed by $E_{0*}$, we have
$H_A^\uparrow \colon \F2^F \cong (\F2^{n_a})^{E_{0*}} \to (\F2^{m_a})^{E_{0*}}$.
Similarly, by thinking of $F$ as being indexed by $E_{1*}$, we have
$H_A^\downarrow \colon \F2^F \cong (\F2^{n_a})^{E_{1*}} \to (\F2^{m_a})^{E_{1*}}$.
We can also define $H_B^\leftarrow$ and $H_B^\rightarrow$.
When the context is clear, we sometime hide the arrows.

\subsection{Dimension and low density}

Before measuring the dimension of the quantum code based on $X$, we verify that $X$ is a well-defined chain complex. For this it suffices to show that for each $f \in F$ and $v \in V$, the restriction $(\partial_1 \partial_2)_f^v\colon \F2 \to \F2^{m_a \times m_b}$ is $0$.
To do so, we first recall the submatrices of the Tanner code described in \Cref{sec:tanner-codes}.

Given elements $e^- \in E^-$ and $f \in F$,
  the submatrix $\tannercode(\cG(E^-, F), H_A)_f^{e^-} \colon \F2 \to \F2^{m_a}$ 
  is $0$ when $e^-$ and $f$ are not incident.
  When $e^-$ and $f$ are incident, say $e^- = ((g, gb), 0*)$, $f = (g, ag, gb, agb)$, we have 
  \begin{equation*}
    \tannercode(\cG(E^-, F), H_A)_f^{e^-} = H_A(\bar a) \colon \F2 \to \F2^{m_a}
  \end{equation*}
  where $\bar a$ is the basis vector of $\F2^A \cong \F2^\Delta$
  corresponding to the element $a \in A$.

Similarly, given elements $v \in V$ and $e^- \in E^-$,
  the submatrix $\tannercode(\cG(V, E^-), H_B)_{e^-}^v \colon \F2^{m_a} \to \F2^{m_a \times m_b}$ is $0$ when $v$ and $e^-$ are not incident.
  When $v$ and $e^-$ are incident, say $v = (g, 00)$, $e^- = ((g, gb), 0*)$, we have
  \begin{equation*}
    \tannercode(\cG(V, E^-), H_B)_{e^-}^v = {-} \otimes H_B(\bar b) \colon \F2^{m_a} \to \F2^{m_a \times m_b}
  \end{equation*}
  where $\bar b$ is the basis vector of $\F2^B \cong \F2^\Delta$ 
  and ${-}$ is the placeholder where ${-} \otimes H_B(\bar b) \colon v \mapsto v \otimes H_B(\bar b)$.

\begin{lemma}
$X$ is a well-defined chain complex, i.e. 
\[(\partial_1 \partial_2)_f^v\colon \F2 \to \F2^{m_a \times m_b} \,=\, 0\;.\]
\end{lemma}

\begin{proof}
Because $\partial_1 \partial_2 = \tannercode(\cG(V, E^-), H_B) \tannercode(\cG(E^-, F), H_A) + \tannercode(\cG(V, E^|), H_A) \tannercode(\cG(E^|, F), H_B)$
 it suffices to compute $(\tannercode(\cG(V, E^-), H_B) \tannercode(\cG(E^-, F), H_A))_f^v$ and  $(\tannercode(\cG(V, E^|), H_A) \tannercode(\cG(E^|, F), H_B))_f^v$.
Now, by matrix multiplication,
  $(\tannercode(\cG(V, E^-), H_B) \tannercode(\cG(E^-, F), H_A))_f^v = \sum_{e^- \in E^-} \tannercode(\cG(V, E^-), H_B)_f^{e^-} \tannercode(\cG(E^-, F), H_A)_{e^-}^v$.
We consider the following two cases.

When $v$ and $f$ are not incident, there is no $e^-$ for both $\tannercode(\cG(V, E^-), H_B)_f^{e^-}$ and $\tannercode(\cG(E^-, F), H_A)_{e^-}^v$ to be non-zero, so $(\tannercode(\cG(V, E^-), H_B) \tannercode(\cG(E^-, F), H_A))_f^v = 0$.
Similarly, $(\tannercode(\cG(V, E^|), H_A) \tannercode(\cG(E^|, F), H_B))_f^v = 0$.
So $(\partial_1 \partial_2)_f^v = 0$ in the case when $v$ and $f$ are not incident.

When $v$ and $f$ are incident, suppose $v = (g, 00)$ and $f = (g, ag, gb, agb)$.
We define $e^- = ((g, gb), 0*)$ and $e^| = ((g, ag), *0)$.
Because $e^-$ is the only edge in $E^-$ that is incident to both $v$ and $f$, we have
\begin{equation*}
  (\tannercode(\cG(V, E^-), H_B) \tannercode(\cG(E^-, F), H_A))_f^v = \tannercode(\cG(V, E^-), H_B)_{e^-}^v \tannercode(\cG(E^-, F), H_A)_f^{e^-} = H_A(\bar a) \otimes H_B(\bar b).
\end{equation*}
Similarly, 
\begin{equation*}
  (\tannercode(\cG(V, E^|), H_A) \tannercode(\cG(E^|, F), H_B))_f^v = \tannercode(\cG(V, E^|), H_A)_{e^|}^v \tannercode(\cG(E^|, F), H_B)_f^{e^|} = H_A(\bar a) \otimes H_B(\bar b).
\end{equation*}
This implies $(\partial_1 \partial_2)_f^v = 0$ and implies $X$ is a chain complex.
\end{proof}

We now check that the boundary maps $\partial_2$ and $\partial_1$ have bounded number of non-zero entries in each column and row.

\begin{lemma}
The code $\cC$ is low density, i.e.\ 
the maps $\partial_2$ and $\partial_1$ have at most $4\Delta$ nonzero entries in each row and column. 
\end{lemma}

\begin{proof}
The result follows because the left-right Cayley graph has bounded degree and the non-zero entry appears only when there is an incident relation.
We call $F$, $E^- \times [m_a] \cup E^| \times [m_b]$, and $V \times [m_a] \times [m_b]$ the face bits, the edge bits, and the vertex bits.
And we say that a face bit is incident to an edge bit if the corresponding entry in the boundary map is non-zero.

We first consider $\partial_2$.
Each face is incident to $4$ edges and each edge is incident to $\Delta$ faces.
Now \\ $\tannercode(\cG(E^-, F), H_A)_f^{e^-} \colon \F2 \to \F2^{m_a}$ and $\tannercode(\cG(E^|, F), H_B)_f^{e^|} \colon \F2 \to \F2^{m_b}$ have $\le 1$ non-zero entry in each row and $\le \max(m_a, m_b)$ non-zero entries in each column.
So each face bit is incident to $\le 4 \max(m_a, m_b)$ edge bits
  and each edge bit is incident to $\le \Delta$ face bits.

We now consider $\partial_1$.
Each edge is incident to $2$ vertices and each vertex is incident to $2\Delta$ edges.
Now $\tannercode(\cG(V, E^-), H_B)_{e^-}^v \colon \F2^{m_a} \to \F2^{m_a \times m_b}$ and $\tannercode(\cG(V, E^|), H_A)_{e^|}^v \colon \F2^{m_b} \to \F2^{m_a \times m_b}$ have $\le 1$ non-zero entry in each row and $\le \max(m_a, m_b)$ non-zero entries in each column.
So each edge bit is incident to $\le 2 \max(m_a, m_b)$ vertex bits
  and each vertex bit is incident to $\le 2 \Delta$ edge bits.
	\end{proof}
	
	It is easy to check that the quantum code has linear dimension.
	
	\begin{lemma}
	The  code $\cC$ has rate at least
	\[ \frac{- (2 \rho_a - 1) (2 \rho_b - 1)}{2(2-\rho_a-\rho_b)}\;.\]
	\end{lemma}
	
	\begin{proof}
The rate is at least
\[ \frac{|X(1)| - |X(2)| - |X(0)|}{|X(1)|} = \frac{- (\Delta - 2 m_a) (\Delta - 2 m_a) |G|}{2 (m_a + m_b) \Delta |G|} = \frac{- (2 \rho_a - 1) (2 \rho_b - 1)}{2(2-\rho_a-\rho_b)}\;.\]
\end{proof}

Note that one can achieve any rate in $(0, 1/2)$ by choosing corresponding $\rho_a$ and $\rho_b$.

\subsection{Distance}
\label{sec:distance}

A quantum CSS code has linear distance if and only if the chain complex $X$ has constant systolic and co-systolic distance. We start with a general theorem, Theorem~\ref{thm:co-expansion} that shows a certain co-expansion property of the complex  $X(\cG_2, C_A, C_B)$ defined in~\eqref{eq:complex-1}. The property of having linear $X$-distance for $\cC$, i.e.\ linear co-systolic distance of $X$, follows almost immediately and is shown in Corollary~\ref{cor:co-distance}. The argument for showing linear $Z$-distance for $\cC$, i.e.\ linear systolic distance for $X$, is more involved and proceeds by reduction to the co-systolic distance. This is shown in Theorem~\ref{thm:co-expansion-implies-expansion}. After having shown the distance properties, we show that the co-expansion property shown in Theorem~\ref{thm:co-expansion} also implies small-set expansion properties for $X$. This is shown in Corollary~\ref{cor:co-expansion}.

\subsubsection{Co-expansion and co-systolic distance}

We start with the main theorem on co-expansion.   

\begin{theorem} [Co-Expansion] \label{thm:co-expansion}
  Given $\Delta$-regular $\lambda$-spectral expander graphs $\Cay(G, A)$, $\Cay(G, B)$
  and linear codes $C_A^\perp, C_B^\perp$ of length $\Delta$ with distance $d_1$ and $(C_A^\perp, C_B^\perp)$ with robustness $d_2$. 
	If $c^1\in \F2^{X(1)}$ is co-locally minimal, then
  \begin{equation} \label{eq:co-locally-minimal-expansion}
    \norm{\delta^1 c^1}_F \ge \frac{d_1 d_2 - \lambda d_2 - 8 \lambda \Delta}{4 d_2 + 8 \Delta} \norm{c^1}_E - \frac{\Delta d_2/2 + 2\Delta}{4 d_2 + 8 \Delta} \frac{\norm{c^1}_E^2}{|G|}\;.
  \end{equation}
	\end{theorem}
		
	\begin{corollary}\label{cor:co-distance}
	Under the same assumptions as Theorem~\ref{thm:co-expansion}, suppose further that  $d_1 d_2 - \lambda d_2 - 8 \lambda \Delta > 0$.
Then the co-chain complex~\eqref{eq:complex-1} has co-systolic distance at least $\frac{\eta}{2\Delta (m_a + m_b)}$, where $\eta \coloneqq \frac{d_1 d_2 - \lambda d_2 - 8 \lambda \Delta}{\Delta d_2/2 + 2\Delta}$.
\end{corollary}

\begin{proof}
 When $c^1$ is a non-zero co-cycle $c^1 \in Z^1-0$,~\eqref{eq:co-locally-minimal-expansion} implies
  \begin{equation*}
    \frac{\Delta d_2/2 + 2\Delta}{4 d_2 + 8 \Delta} \frac{\norm{c^1}_E^2}{|G|} \ge \frac{d_1 d_2 - \lambda d_2 - 8 \lambda \Delta}{4 d_2 + 8 \Delta} \norm{c^1}_E\;,
  \end{equation*}
  which gives
  \begin{equation*}
    \norm{c^1}_E \ge \frac{d_1 d_2 - \lambda d_2 - 8 \lambda \Delta}{\Delta d_2/2 + 2\Delta} |G| \coloneqq \eta |G| = \frac{\eta}{2\Delta (m_a + m_b)} |X(1)|\;.
  \end{equation*}
	\end{proof}

We now move on to prove the theorem. 

\begin{proof} [Proof of \Cref{thm:co-expansion}]
Let $c^1$ be co-locally minimal and $c^2 = \delta^1 c^1$. 
  Let $\cE \subset E$ be the support of $c^1$, i.e. $\cE = \set{e \in E : c^1(e) \ne 0}$. 
  (Recall that $c^1(e_{*0}), c^1(e_{*1}) \in \F2^{m_b}$ and $c^1(e_{0*}), c^1(e_{1*}) \in \F2^{m_a}$.)
  The proof strategy is to count the number of ``neighbors'' between $\cE$, for some appropriate neighborhood structure.
  The expansion of the graph gives an upper bound on the number of ``neighbors''
  and the distance and the robustness of the local code give a lower bound.
  Comparing the two bounds gives~\Cref{eq:co-locally-minimal-expansion}.

  \paragraph{Step 1: Define ``neighbors'' $M$.}
	Recall the adjacency matrices $M_0$ and $M_1$ defined in Lemma~\ref{lem:M_0} and Lemma~\ref{lem:M_1} respectively. 
  We describe the neighborhood structure through the matrix 
	\[ M \,=\, d_2 M_1 + M_0 \in  \RR^{E\times E}\;.\]
 Let $1_{\cE} \in \F2^{E}$ be the indicator vector for $\cE$. 
  
  \paragraph{Step 2: Upper bound from expansion.}
Combining Lemma~\ref{lem:M_0} and Lemma~\ref{lem:M_1}, 
  \begin{equation}\label{eq:dist1-lb}
    1_\cE^T M 1_\cE \,\le\, \lambda (d_2 + 8 \Delta) |\cE| + \frac{\Delta}{|G|}\Big(\frac{d_2}{2}+2\Big) |\cE|^2\;.
  \end{equation}

  \paragraph{Step 3: Lower bound from distance and robustness.}
  We show a lower bound on $1_\cE^T M 1_\cE$ using the distance and the robustness of the local tensor code. We start with two claims. The first claim uses the distance property of $C_B^\perp$. 
 
	\begin{claim}\label{claim:64}
	For any edge $e_{*0} \in E_{*0}$ it holds that
  \begin{equation} \label{eq:ineq2-from-distance}
    \norm{c^2(F(e_{*0}))}_F + \norm{c^1(E_{*1}(e_{*0}))}_E + \norm{c^1(E_{0*}(e_{*0}))}_E + \norm{c^1(E_{1*}(e_{*0}))}_E \ge d_1 \norm{c^1(e_{*0})}_E\;.
  \end{equation}
	\end{claim}
	
	\begin{proof}
	The distance property of $C_B^\perp$ immediately implies that 
	 \begin{equation} \label{eq:ineq-from-distance}
    \norm{s^2(e_{*0})}_F \ge d_1 \norm{c^1(e_{*0})}_E\;.
  \end{equation}
	Recall that 
		 \begin{equation} \label{eq:ineq-from-distance-b}
	c^2 = s^2(E_{*0})+s^2(E_{*1})+s^2(E_{0*})+s^2(E_{1*})\;.
	\end{equation}
	Thus each non-zero entry $f=(g, ag, gb, agb)$ of $s^2(e_{*0})$, where $e_{*0}=(g,ag)$, is either a non-zero entry in $c^2$ or is canceled by a term in $s^2(E_{*1})$, $s^2(E_{0*})$, or $s^2(E_{1*})$ that contributes to the entry at $f$.
  Such a term, say $s^2(e_{*1}) \ne 0$, must have its edge $e_{*1}$ incident to the face $f$ which is incident to $e_{*0}$, i.e. $e_{*1} \in E_{*1}(e_{*0})$.
  Therefore, each cancellation contributes to $\norm{s^2(E_{*1}(e_{*0}))}_E$, $\norm{s^2(E_{0*}(e_{*0}))}_E$, or $\norm{s^2(E_{1*}(e_{*0}))}_E$.
  Notice that $s^2(e_{*1}) \ne 0 \iff c^1(e_{*1}) \ne 0$. Thus from~\eqref{eq:ineq-from-distance-b} we get that
\[	\norm{c^2(F(e_{*0}))}_F + \norm{c^1(E_{*1}(e_{*0}))}_E + \norm{c^1(E_{0*}(e_{*0}))}_E + \norm{c^1(E_{1*}(e_{*0}))}_E \ge\norm{c^1(E_{*0}(e_{*0}))}_E\;.\]
Combined with~\eqref{eq:ineq-from-distance}, this shows the claim.  
	\end{proof}
	
	The second claim uses the robustness property of $\Sigma(C_A^T,C_B^T)$. 
	
	\begin{claim}\label{claim:65}
	For any vertex $v_{00}\in V_{00}$ it holds that
  \begin{equation} \label{eq:ineq2-from-robust}
    \norm{c^2(F(v_{00}))}_F + \norm{c^1(E_{*1}(v_{00}))}_E + \norm{c^1(E_{1*}(v_{00}))}_E \ge d_2 (\norm{c^1(E_{*0}(v_{00}))}_E + \norm{c^1(E_{0*}(v_{00}))}_E)\;.
  \end{equation}
	Similarly, for any vertex $v_{10}\in V_{10}$ it holds that
	  \begin{equation} \label{eq:ineq2b-from-robust}
    \norm{c^2(F(v_{10}))}_F + \norm{c^1(E_{*1}(v_{10}))}_E + \norm{c^1(E_{0*}(v_{10}))}_E \ge d_2 (\norm{c^1(E_{*0}(v_{10}))}_E + \norm{c^1(E_{1*}(v_{10}))}_E)\;.
  \end{equation}
	\end{claim}

\begin{proof}
The proof is similar to the previous claim, using the robustness property of $(C_A^\perp,C_B^\perp)$ instead of the distance property of $C_B$. 
\end{proof}

Let $e_{*0} \in E_{*0}$ have endpoints $v_{00} \in V_{00}$ and $v_{10} \in V_{10}$. Using the definition of $M=d_2M_1+M_0$,
  \begin{align*}
    & 1_\cE^T M 1_{e_{*0}} \\
    &= 1_\cE^T \big( d_2 M_1 + M_0 \big) 1_{e_{*0}} \\
    &= d_2 \norm{c^1(E_{*1}(e_{*0}))}_E + \norm{c^1(E_{*1}(v_{00}))}_E + \norm{c^1(E_{1*}(v_{00}))}_E + \norm{c^1(E_{*1}(v_{10}))}_E + \norm{c^1(E_{0*}(v_{10}))}_E \\
    &\ge d_2 \norm{c^1(E_{*1}(e_{*0}))}_E + d_2 \norm{c^1(E_{*0}(v_{00}))}_E + d_2 \norm{c^1(E_{0*}(v_{00}))}_E + d_2 \norm{c^1(E_{*0}(v_{10}))}_E + d_2 \norm{c^1(E_{1*}(v_{10}))}_E \\
    &- |c^2(F(v_{00}))| - |c^2(F(v_{10}))| \\
    &\ge d_2 \norm{c^1(E_{*1}(e_{*0}))}_E + d_2 \norm{c^1(E_{0*}(v_{00}))}_E + d_2 \norm{c^1(E_{1*}(v_{10}))}_E - \norm{c^2(F(v_{00}))}_F - \norm{c^2(F(v_{10}))}_F \\
    &= d_2 \norm{c^1(E_{*1}(e_{*0}))}_E + d_2 \norm{c^1(E_{0*}(e_{*0}))}_E + d_2 \norm{c^1(E_{1*}(e_{*0}))}_E - \norm{c^2(F(v_{00}))}_F - \norm{c^2(F(v_{10}))}_F \\
    &\ge d_1 d_2 \norm{c^1(e_{*0})}_E - d_2 \norm{c^2(F(e_{*0}))}_F - \norm{c^2(F(v_{00}))}_F - \norm{c^2(F(v_{10}))}_F\;.
  \end{align*}
Here, the first inequality uses~\eqref{eq:ineq2-from-robust} and~\eqref{eq:ineq2b-from-robust}, the second inequality drops non-negative terms, and the last inequality follows from~\eqref{eq:ineq2-from-distance}.
Summing over all edges $e_{*0}$ and analogous inequalities shown for edges $e_{*1}$, $e_{0*}$ and $e_{1*}$ we obtain
  \begin{equation}\label{eq:dist1-ub}
    1_\cE^T M 1_\cE \ge d_1 d_2 \norm{c^1}_E - 4 d_2 \norm{c^2}_F - 8 \Delta \norm{c^2}_F\;,
  \end{equation}
  where the factor of $4$ is because each face is counted $4$ times by the $4$ edges incident to the face,
  and the factor of $8 \Delta$ because each vertex is summed over $2 \Delta$ times by the $2 \Delta$ edges incident to the vertex and each face is incident to $4$ vertices.

  \paragraph{Step 4: Combine the upper and lower bounds.}
  Combining~\eqref{eq:dist1-lb} and~\eqref{eq:dist1-ub},
  \begin{equation*}
    d_1 d_2 \norm{c^1}_E - (4 d_2 + 8 \Delta) \norm{c^2}_F \le   1_\cE^T M 1_\cE
    \le \lambda (d_2 + 8 \Delta) \norm{c^1}_E + \frac{\Delta}{|G|}(\frac{d_2}{2}+2) \norm{c^1}_E^2\;,
  \end{equation*}
  which implies
  \begin{equation*}
    \norm{c^2}_F \ge \frac{d_1 d_2 - \lambda d_2 - 8 \lambda \Delta}{4 d_2 + 8 \Delta} \norm{c^1}_E - \frac{\Delta d_2/2 + 2\Delta}{4 d_2 + 8 \Delta} \frac{\norm{c^1}_E^2}{|G|}\;.
  \end{equation*}
  This concludes the proof of~\eqref{eq:co-locally-minimal-expansion}.
\end{proof}

\subsubsection{Expansion and systolic distance}

The second main theorem for this section shows that systolic distance follows from cp-systolic distance. In fact, we will prove a stronger statement which also shows that expansion follows from co-expansion. 

\begin{restatable} [Co-Expansion $\to$ Expansion] {theorem}{CoExpansionToExpansion} \label{thm:co-expansion-implies-expansion}
  If $X(\cG_2, C_A^\perp, C_B^\perp)$ has co-systolic distance $\frac{\eta}{2\Delta(k_a + k_b)}$,
  then $X(\cG_2, C_A, C_B)$ has systolic distance $\frac{\eta}{2\Delta(m_a + m_b)}$.

  Furthermore, if $X(\cG_2, C_A^\perp, C_B^\perp)$ is a $(\frac{\eta}{4\Delta(k_a + k_b)}, \beta, \gamma)$-small-set co-boundary expander, 
  then $X(\cG_2, C_A, C_B)$ is a \\ $(\frac{\eta}{4\Delta(m_a + m_b)}, \frac{1}{\Delta^2 + \Delta + \frac{\Delta^3}{\beta}}, \Delta + \Delta^2 \gamma)$-small-set boundary expander.
\end{restatable}

To use the second part of this theorem we need to know that $X$ is a small-set co-boundary expander. This will be shown in Corollary~\ref{cor:co-expansion}. 

We now have all the ingredients to state the property of linear distance for our quantum code. 

\begin{corollary}\label{cor:distance}
Assume that $\Cay(G, A), \Cay(G, B)$ are $\lambda = \Theta(\sqrt{\Delta})$-spectral expanders and that $C_A, C_B$ have distance $d_1, d_2 = \Theta(\Delta)$. Then the quantum code $\cC$ has linear distance. 
\end{corollary}

\begin{proof}
The assumptions made in the corollary imply that for large enough $\Delta$, $d_1 d_2 - \lambda d_2 - 8 \lambda \Delta >0$.
  By \Cref{thm:co-expansion}, $X(\cG_2, C_A, C_B)$ and $X(\cG_2, C_A^\perp, C_B^\perp)$ have linear co-systolic distance.
  By \Cref{thm:co-expansion-implies-expansion}, $X(\cG_2, C_A, C_B)$ has linear systolic distance.
  Therefore, $\cC$ has distance $\frac{d_1 d_2 - \lambda d_2 - 8 \lambda \Delta}{\Delta^2 (m_a + m_b) (d_2 + 4)} n$.
\end{proof}

Now, we move on to the proof of the second theorem.

\begin{proof} [Proof of \Cref{thm:co-expansion-implies-expansion}]
We start with an introductory discussion in which we focus on the simpler case of systolic distance, for which we have $\partial_1 c_1 = 0$. For the proof we will show the general case of expansion. 

  Recall that systolic distance is saying that given $c_1$ with small weight there exists $c_2$ such that $\partial_2 c_2 = c_1$.
  So the task is to find $c_2$.
  We will do so by first making local guesses of $c_2$ around each vertex.
  For each vertex $v_{00}$ we focus on the local exact chain complex around $v_{00}$,
  \begin{equation*}
    (\F2)^{F(v_{00})} \xrightarrow{\partial_2} (\F2^{m_b})^{E_{*0}(v_{00})} \oplus (\F2^{m_a})^{E_{0*}(v_{00})} \xrightarrow{\partial_1} \F2^{m_a \times m_b}.
  \end{equation*}
  Note that this is isomorphic to the exact chain complex $Y(H_A, H_B) \colon \F2^{n_a \times n_b} \xrightarrow{\partial_2} \F2^{n_a \times m_b + m_a \times n_b} \xrightarrow{\partial_1} \F2^{m_a \times m_b}$ defined in~\eqref{lem:exact}.
  Because $\partial_2 c_2 = c_1$, when restricted on this local region we have $\partial_2 c_2(F(v_{00})) = c_1(E_{*0}(v_{00}) \cup E_{0*}(v_{00}))$.
  At this moment we do not know what $c_2(F(v_{00}))$ is, but we can always write $c_2(F(v_{00})) = s_2(v_{00}) + w_2(v_{00})$ where $\partial_2 s_2(v_{00}) = c_1(E_{*0}(v_{00}) \cup E_{0*}(v_{00}))$ and $\partial_2 w_2(v_{00}) = 0$.
  We now fix an arbitrary $s_2(v_{00})$ then worry about finding $w_2(v_{00})$.

  We know two things about $w_2$.
  First, they are codewords $w_2(v_{00}) \in C_A \otimes C_B$.
  Second, because $s_2(V_{00}) + w_2(V_{00}) = c_2 = s_2(V_{10}) + w_2(V_{10})$,
    the sum is known $w_2(V_{00}) + w_2(V_{10}) = s_2(V_{00}) + s_2(V_{10})$.
    ($s_2(V_{00})$ denotes the concatenation of $s_2(v_{00})$.)
  Turns out this task can be interpreted as finding $c^0$ (related to $w_2$) given $c^1 = \delta^0 c^0$ (related to $w_2(V_{00}) + w_2(V_{10})$).
  Therefore, using the hypothesis on co-expansion, we can recover $w_2$ and find $c_2$.

  We now comment more on what $w_2$ is doing.
  A useful interpretation is that the task of finding $w_2$ can be interpreted as finding the ``corrections'' that make $s_2$ consistent.
  Notice that $s_2$ is somewhat like $c_2$ with the same value after $\partial_2$, $\partial_2 s_2(v_{00}) = \partial_2 c_2(F(v_{00}))$, except that $s_2$ is not defined on $F$.
  This is because the value of a face is dependent on which vertex it views from.
  That is each of $s_2(V_{00})$, $s_2(V_{10})$, $s_2(V_{01})$, $s_2(V_{11})$ defines a bit string on $F$, but they are not guaranteed to be consistent.
  In this view point, $w_2$ is the ``correction'' that makes $s_2$ consistent through $s_2(V_{00}) + w_2(V_{00}) = c_2 = s_2(V_{10}) + w_2(V_{10})$.

  We now describe the detailed construction in steps. Remember we are now showing boundary expansion and we let $c_0 = \partial_1 c_1$.

  \paragraph{Step 1: Construct $s_2$ by guessing from the vertex.}
  Let $VE^-$ and $VE^|$ be the set of vertex, edge pair 
    $VE^- = \{(v, e^-) : v \in e^-, e^- \in E^-\}$ and
    $VE^| = \{(v, e^|) : v \in e^|, e^| \in E^|\}$.
  We find $s_2 \in (\F2^{n_a})^{VE^-} \times (\F2^{n_b})^{VE^|}$ (i.e. for each $(v, e^-)$ we assign a value in $\F2^{n_a}$ and 
    for each $(v, e^|)$ we assign a value in $\F2^{n_b}$)
    such that 
  \begin{equation} \label{eq:map1}
    H_B^\leftarrow s_2(v_{00}, e_{*0}) = c_1(e_{*0}), H_A^\uparrow s_2(v_{00}, e_{0*}) = c_1(e_{0*}).
  \end{equation}
  We further require $s_2$ to satisfy two more conditions.
  When $c_1(e_{*0}) = 0$ we set $s_2(v_{00}, e_{*0}) = 0$
  \begin{equation} \label{eq:approx1-1}
    c_1(e_{*0}) = 0 \implies s_2(v_{00}, e_{*0}) = 0.
  \end{equation}
  When $c_0(v_{00}) = 0$, by exactness of~\eqref{lem:exact},
    we can require $s_2(v_{00}, E_{*0}(v_{00})) = s_2(v_{00}, E_{0*}(v_{00}))$.
  Therefore, $\norm{s_2(v_{00}, E_{*0}(v_{00})) + s_2(v_{00}, E_{0*}(v_{00}))}_F \le \Delta^2 \norm{c_0(v_{00})}_V$, i.e.
  \begin{equation} \label{eq:approx1-2}
    \norm{s_2(V_{00}, E_{*0}) + s_2(V_{00}, E_{0*})}_F \le \Delta^2 \norm{c_0(V_{00})}_V
  \end{equation}
  and $\norm{s_2(v_{00}, E_{*0}(v_{00})) + s_2(v_{00}, E_{0*}(v_{00}))}_{E_{0*}} \le \Delta \norm{c_0(v_{00})}_V$, i.e.
  \begin{equation} \label{eq:approx1-3}
    \norm{s_2(V_{00}, E_{*0}) + s_2(V_{00}, E_{0*})}_{E_{0*}} \le \Delta \norm{c_0(V_{00})}_V.
  \end{equation}
  Recall $\norm{s_2(v_{00}, E_{*0}(v_{00})) + s_2(v_{00}, E_{0*}(v_{00}))}_{E_{0*}}$ means we think of $s_2(v_{00}, E_{*0}(v_{00})) + s_2(v_{00}, E_{0*}(v_{00})) \in \F2^{F(v_{00})} \cong (\F2^{n_a})^{E_{0*}(v_{00})}$ as being indexed by $E_{0*}(v_{00})$.
  
  \begin{figure}[H]
    \centering
    \begin{tikzpicture} [scale=1.2, every node/.style={font=\footnotesize}]
      \begin{scope}
        \draw[dashed] (0.25, 0.25) rectangle node{$s_2(v_{00}, E_{0*}(v_{00}))$} ++(2, 2);
        \draw[dashed] (-0.25, -0.25) rectangle node{$s_2(v_{00}, E_{*0}(v_{00}))$} ++(-2, -2);
  
        \draw (0, 0) node[rotate=45]{$\approx$};
  
        \draw[->] (1.25, 2.5) -- (1.25, 2.75);
        \draw[->] (-2.5, -1.25) -- (-2.75, -1.25);

        \draw (0.25, 3) rectangle node{$c_1(v_{00}, E_{0*}(v_{00}))$} ++(2, 1);
        \draw (-3, -0.25) rectangle node[rotate=90]{$c_1(v_{00}, E_{*0}(v_{00}))$} ++(-1, -2);

        \draw (-4, 4) rectangle node{$\approx 0$} ++(1, -1);
        \draw[->] (0, 3.5) -- (-2.75, 3.5);
        \draw[->] (-3.5, 0) -- (-3.5, 2.75);
      \end{scope}
    \end{tikzpicture}
    \caption{Construct $s_2$ using $c_1$ that satisfies \Cref{eq:map1,eq:approx1-1,eq:approx1-2,eq:approx1-3}.}
    \label{fig:expansion-step1}
  \end{figure}
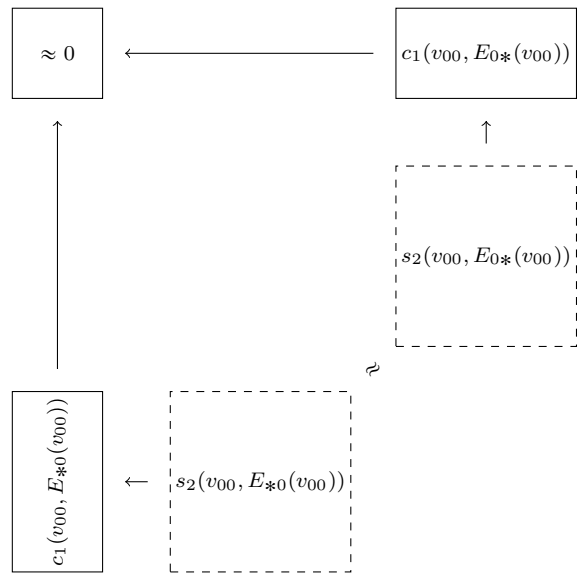

  \paragraph{Step 2: Construct $t_2$ from the sum of $s_2$.}
  $t_2 \in (\F2^{n_a})^{E^-} \times (\F2^{n_b})^{E^|}$ is defined as the ``sum'' of $s_2$ viewed from its two endpoints.
  That is
  \begin{equation*}
    t_2(e_{*0}) = s_2(v_{00}, e_{*0}) + s_2(v_{10}, e_{*0}).
  \end{equation*}
  By \Cref{eq:map1},
  \begin{equation} \label{eq:map2}
    H_B^\leftarrow t_2(e_{*0}) = c_1(e_{*0}) + c_1(e_{*0}) = 0, H_A^\uparrow t_2(e_{0*}) = 0
  \end{equation}
  which implies $t_2(e_{*0}) \in C_B$ and $t_2(e_{*0}) \in C_A$.
  By \Cref{eq:approx1-1} $c_1(e_{*0}) = 0$ implies both $s_2(v_{00}, e_{*0}) = 0$ and $s_2(v_{10}, e_{*0}) = 0$, 
  which implies $t_2(e_{*0}) = 0$.
  So
  \begin{equation} \label{eq:approx2-1}
    \norm{t_2}_E \le \norm{c_1}_E.
  \end{equation}
  By construction,
  \begin{align*}
    t_2(E_{*0}) + t_2(E_{*1}) + t_2(E_{0*}) + t_2(E_{1*})
    &= (s_2(V_{00}, E_{*0}) + s_2(V_{00}, E_{0*})) 
    + (s_2(V_{10}, E_{*0}) + s_2(V_{10}, E_{1*})) \\
    &+ (s_2(V_{01}, E_{*1}) + s_2(V_{01}, E_{0*})) 
    + (s_2(V_{11}, E_{*1}) + s_2(V_{11}, E_{1*})).
  \end{align*}
  Combine with \Cref{eq:approx1-2}, we have 
  \begin{equation} \label{eq:approx2-2}
    \norm{t_2(E_{*0}) + t_2(E_{*1}) + t_2(E_{0*}) + t_2(E_{1*})}_F
    \le \Delta^2 \norm{c_0}_V.
  \end{equation}

  \begin{figure}[H]
    \centering
    \begin{tikzpicture} [scale=1, every node/.style={font=\footnotesize}]
      \begin{scope}
        \draw (0, 0) rectangle node[rotate=90]{$c_1(e_{*0})$} ++(-1, -2);
        \draw (2.75, 0) rectangle node{$s_2(v_{00}, e_{*0})$} ++(-2, -2);

        \draw[->] (0.5, -1) -- (0.25, -1);
      \end{scope}
      \begin{scope} [shift={(0,-2.5)}]
        \draw[dashed] (0, 0) rectangle node[rotate=90]{$0$} ++(-1, -2);
        \draw[dashed] (2.75, 0) rectangle node{$t_2(e_{*0})$} ++(-2, -2);

        \draw[->] (0.5, -1) -- (0.25, -1);
      \end{scope}
      \begin{scope} [shift={(0,-5)}]
        \draw (0, 0) rectangle node[rotate=90]{$c_1(e_{*0})$} ++(-1, -2);
        \draw (2.75, 0) rectangle node{$s_2(v_{10}, e_{*0})$} ++(-2, -2);

        \draw[->] (0.5, -1) -- (0.25, -1);
      \end{scope}
    \end{tikzpicture}
    \caption{Construct $t_2$ using $s_1$ that satisfies \Cref{eq:map2,eq:approx2-1,eq:approx2-2}.}
    \label{fig:expansion-step2}
  \end{figure}
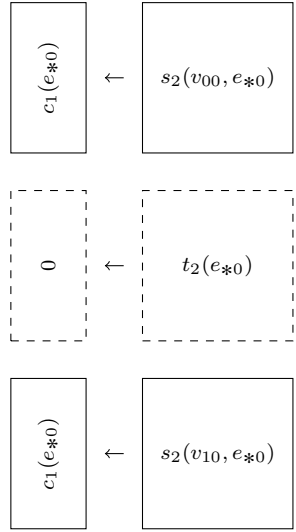

  \paragraph{Step 3: Apply co-expansion assumption.}
  This is the main step where we use the co-expansion assumption and find the ``corrections''.
  
  Let $H_A^\perp\colon \F2^{n_a} \to \F2^{k_a}$, $H_B^\perp\colon \F2^{n_b} \to \F2^{k_b}$ be the parity-check matrices of $C_A^\perp$ and $C_B^\perp$ where $k_a=n_a-m_a$ and $k_b=n_b-m_b$. 
  The following chain complex is the one we will apply the co-expansion assumption
  \begin{equation*}
    X(\cG_2, C_A^\perp, C_B^\perp)\colon \F2^F \xleftarrow{\delta_1} (\F2^{k_a})^{E^-} \oplus (\F2^{k_b})^{E^|} \xleftarrow{\delta_0} (\F2^{k_a \times k_b})^V.
  \end{equation*}
  The overall strategy is to first move to this new chain complex by constructing $c^1$. 
  Then, use co-expansion to write $c^1 = c^1{}' + \delta^0 c^0$.
  Finally, we move back to the original chain complex.

  We first construct $c^1$.
  \Cref{eq:map2} implies that one can find $c^1\in (\F2^{k_a})^{E^-} \times (\F2^{k_b})^{E^|}$ such that 
  \begin{equation*}
    t_2(e_{*0}) = (H_B^\perp)^T c^1(e_{*0}), t_2(e_{0*}) = (H_A^\perp)^T c^1(e_{0*}).
  \end{equation*}

  To apply the hypothesis of co-boundary expansion
    we need to check $c^1$ and $c^2 \coloneqq \delta^1 c^1$ have small weight.
  We need small weight of $c^1$ to satisfy the small set hypothesis
    and we need small weight of $c^2$ for the bound of $\norm{c^1{}'}_E$ to be meaningful.
  We first check $c^1$ has small weight.
  Because $(H_B^\perp)^T$ is injective, $\norm{c^1(e_{*0})}_E \le \norm{t_2(e_{*0})}_E$,
    together with \Cref{eq:approx2-1} implies 
  \begin{equation*}
    \norm{c^1}_E \le \norm{c_1}_E.
  \end{equation*}
  Now, we check $c^2$ has small weight.
  Because $c^2 = t_2(E_{*0}) + t_2(E_{*1}) + t_2(E_{0*}) + t_2(E_{1*})$
  by \Cref{eq:approx2-2} we have
  \begin{equation*}
    \norm{c^2}_F \le \Delta^2 \norm{c_0}_V.
  \end{equation*}

  Given that $c^1$ and $c^2$ have small weight, we can now use the hypothesis of co-boundary expansion and write $c^1 = c^1{}' + \delta^0 c^0$ where $c^1{}' \in (\F2^{k_a})^{E^-} \oplus (\F2^{k_b})^{E^|}$ and $c^0 \in (\F2^{k_a \times k_b})^V$ such that
  \begin{equation*}
    \norm{c^2}_F \ge \beta \norm{c^1{}'}_E
  \end{equation*}
  and
  \begin{equation*}
    \norm{c^0}_V \le \gamma \norm{c^1}_E.
  \end{equation*}

  To move back to the original chain complex, 
  we define $u_2 \in (C_A)^{E^-} \times (C_B)^{E^|}$
    and $w_2 \in (C_A \otimes C_B)^V$
  where
  \begin{equation*}
    u_2(e_{*0}) = (H_B^\perp)^T c^1{}'(e_{*0}), u_2(e_{0*}) = (H_A^\perp)^T c^1{}'(e_{0*})
  \end{equation*}
  and 
  \begin{equation*}
    w_2(v_{00}) = (H_A^\perp \otimes H_B^\perp)^T c^0(v_{00}).
  \end{equation*}

  By construction, $u_2$ and $w_2$ inherit various properties from $c^1{}'$ and $c^0$.
  First, $u_2$ and $w_2$ are codewords
  \begin{equation} \label{eq:map3-1}
    H_B^{\leftarrow} u_2(e_{*0}) = 0, H_A^{\uparrow} u_2(e_{0*}) = 0.
  \end{equation}
  and 
  \begin{equation} \label{eq:map3-2}
    H_B^{\leftarrow} w_2(v_{00}) = 0, H_A^{\uparrow} w_2(v_{00}) = 0.
  \end{equation}
  Second, $u_2$ and $w_2$ have small weights.
  Because $\norm{u_2}_E \le \norm{c^1{}'}_E$ and $\norm{w_2}_V \le \norm{c^0}_V$, we have 
  \begin{equation} \label{eq:approx3-1}
    \norm{c^2}_F \ge \beta \norm{u_2}_E
  \end{equation}
  and
  \begin{equation} \label{eq:approx3-2}
    \norm{w_2}_V \le \gamma \norm{c^1}_E.
  \end{equation}
  Finally, $u_2$ and $w_2$ form a decomposition of $t_2$.
  Because $c^1 = c^1{}' + \delta^0 c^0$,
  \begin{equation} \label{eq:equal3}
    t_2(E_{*0}) = w_2(V_{00}) + u_2(E_{*0}) + w_2(V_{10}).
  \end{equation}
  
  These $u_2$ and $w_2$ can be understood as the ``corrections'' that bridge the inconsistent local guesses, $s_2$.
  And the upper bound implies, one can use a few ``corrections'' to make the local guesses consistent.

  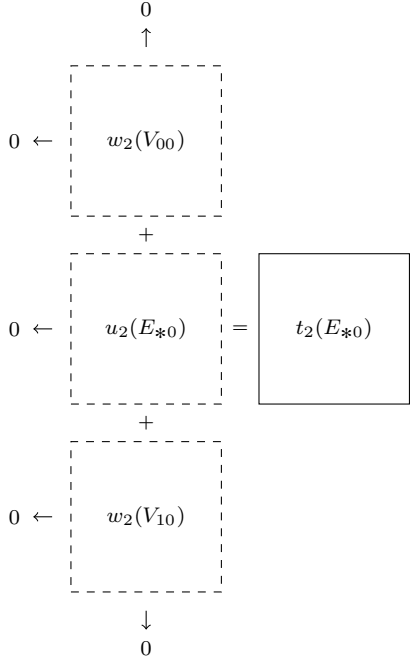
\begin{figure}[H]
    \centering
    \begin{tikzpicture} [scale=1, every node/.style={font=\footnotesize}]
      \begin{scope}
        \draw[dashed] (0, 0) rectangle node{$w_2(V_{00})$} ++(2, -2);
        \draw[->] (-0.25, -1) -- (-0.5, -1);
        \draw (-0.75, -1) node{$0$};
        \draw[->] (1, 0.25) -- (1, 0.5);
        \draw (1, 0.75) node{$0$};

        \draw (1, -2.25) node{$+$};

        \draw[dashed] (0, -2.5) rectangle node{$u_2(E_{*0})$} ++(2, -2);
        \draw[->] (-0.25, -3.5) -- (-0.5, -3.5);
        \draw (-0.75, -3.5) node{$0$};

        \draw (1, -4.75) node{$+$};

        \draw[dashed] (0, -5) rectangle node{$w_2(V_{10})$} ++(2, -2);
        \draw[->] (-0.25, -6) -- (-0.5, -6);
        \draw (-0.75, -6) node{$0$};
        \draw[->] (1, -7.25) -- (1, -7.5);
        \draw (1, -7.75) node{$0$};

        \draw (2.25, -3.5) node{$=$};

        \draw (2.5, -2.5) rectangle node{$t_2(E_{*0})$} ++(2, -2);
      \end{scope}
    \end{tikzpicture}
    \caption{Construct $u_2$ and $w_2$ using the co-expansion assumption. $u_2$ and $w_2$ satisfy \Cref{eq:map3-1,eq:map3-2,eq:approx3-1,eq:approx3-2,eq:equal3}.}
    \label{fig:expansion-step3}
  \end{figure}

  \paragraph{Step 4: Obtain $c_2$.}
  Finally, we set
  \begin{equation*}
    c_2 \coloneqq s_2(V_{00}, E_{*0}) + w_2(V_{00}).
  \end{equation*}
  Here, $s_2(V_{00}, E_{*0})$ can be interpreted as the initial guess and $w_2(V_{00})$ as the ``correction''.
  Because we write $c_1 = c_1' + \partial_2 c_2$, we have $c_1' \coloneqq c_1 + \partial_2 c_2$.
  This concludes the construction.

  To show expansion, we suffice to show $c_1'$ and $c_2$ have small weight. This is a straightforward computation using the fact that the number of ``corrections'' is small.
  Note that when $c_1 \in Z_1$, $c_1'$ is simply $0$.
  This is why this long computation does not appear in \cite{dinur2021locally}.

  \paragraph{Show $c_1'$ is small.}
  The weight of $c_1'$ has four components 
  $\norm{c_1'}_E = \norm{c_1'(E_{*0})}_E + \norm{c_1'(E_{*1})}_E + \norm{c_1'(E_{0*})}_E + \norm{c_1'(E_{1*})}_E$
  and we will bound them individually.
  The general strategy is to first use the identities
  \[c_1(E_{*0}) = H_B^\leftarrow s_2(V_{00}, E_{*0}),\]
  \[c_1(E_{0*}) = H_A^\uparrow s_2(V_{00}, E_{0*}),\]
  \[c_1(E_{*1}) = H_B^\rightarrow s_2(V_{01}, E_{*1}),\]
  \[c_1(E_{1*}) = H_A^\downarrow s_2(V_{10}, E_{1*}),\]
  then express the difference between the $s_2$ in the identities and the $s_2$ in the definition of $c_2 = s_2(V_{00}, E_{*0}) + w_2(V_{00})$ as something of small weight.
  Be aware that some of the arrow directions in $H_A$ and $H_B$ below may be different from what one might expect, for example $H_B^\rightarrow s_2(V_{00}, E_{*0})$.

  \paragraph{Bound $\norm{c_1'(E_{*0})}_E$.}
  \begin{align*}
    c_1'(E_{*0})
    &= c_1(E_{*0}) + H_B^\leftarrow s_2(V_{00}, E_{*0}) + H_B^\leftarrow w_2(V_{00}) \\
    &= H_B^\leftarrow s_2(V_{00}, E_{*0}) + H_B^\leftarrow s_2(V_{00}, E_{*0}) + 0 \\
    &= 0.
  \end{align*}
  So $\norm{c_1'(E_{*0})}_E = 0$.

  \paragraph{Bound $\norm{c_1'(E_{0*})}_E$.}
  \begin{align*}
    c_1'(E_{0*})
    &= c_1(E_{0*}) + H_A^\uparrow s_2(V_{00}, E_{*0}) + H_A^\uparrow w_2(V_{00}) \\
    &= H_A^\uparrow s_2(V_{00}, E_{0*}) + H_A^\uparrow s_2(V_{00}, E_{*0}) + 0 \\
    &= H_A^\uparrow (s_2(V_{00}, E_{0*}) + s_2(V_{00}, E_{*0})).
  \end{align*}
  By \Cref{eq:approx2-2}, $\norm{s_2(V_{00}, E_{0*}) + s_2(V_{00}, E_{*0})}_{E_{0*}} \le \norm{c_0(V_{00})}_V$, so
  \begin{equation*}
    \norm{c_1'(E_{0*})}_E \le \norm{s_2(V_{00}, E_{0*}) + s_2(V_{00}, E_{*0})}_{E_{0*}} \le \Delta \norm{c_0(V_{00})}_V.
  \end{equation*}

  \paragraph{Bound $\norm{c_1'(E_{*1})}_E$.}
  \begin{align*}
    c_1'(E_{*1})
    =& c_1(E_{*1}) + H_B^\rightarrow s_2(V_{00}, E_{*0}) + H_B^\rightarrow w_2(V_{00}) \\
    =& H_B^\rightarrow s_2(V_{01}, E_{*1}) + H_B^\rightarrow s_2(V_{00}, E_{*0}) + H_B^\rightarrow w_2(V_{00}) \\
    =& H_B^\rightarrow (s_2(V_{01}, E_{*1}) + s_2(V_{01}, E_{0*}))
    + H_B^\rightarrow (s_2(V_{01}, E_{0*}) + s_2(V_{00}, E_{0*})) \\
      &+ H_B^\rightarrow (s_2(V_{00}, E_{0*}) + s_2(V_{00}, E_{*0})) 
      + H_B^\rightarrow w_2(V_{00}) \\
    =& H_B^\rightarrow (s_2(V_{01}, E_{*1}) + s_2(V_{01}, E_{0*}))
    + H_B^\rightarrow t_2(E_{0*}) \\
      &+ H_B^\rightarrow (s_2(V_{00}, E_{0*}) + s_2(V_{00}, E_{*0})) 
      + H_B^\rightarrow w_2(V_{00}) \\
    =& H_B^\rightarrow (s_2(V_{01}, E_{*1}) + s_2(V_{01}, E_{0*}))
    + H_B^\rightarrow (w_2(V_{00}) + u_2(E_{0*}) + w_2(V_{01})) \\
      &+ H_B^\rightarrow (s_2(V_{00}, E_{0*}) + s_2(V_{00}, E_{*0})) 
      + H_B^\rightarrow w_2(V_{00}) \\
    =& H_B^\rightarrow (s_2(V_{01}, E_{*1}) + s_2(V_{01}, E_{0*}))
    + H_B^\rightarrow u_2(E_{0*})
    + H_B^\rightarrow (s_2(V_{00}, E_{0*}) + s_2(V_{00}, E_{*0})).
  \end{align*}
  By \Cref{eq:approx2-2} and \Cref{eq:approx2-1}, $\norm{s_2(V_{01}, E_{*1}) + s_2(V_{01}, E_{0*})}_{E_{*1}} \le \Delta \norm{c_0(V_{01})}_V$ and $\norm{s_2(V_{00}, E_{*0}) + s_2(V_{00}, E_{0*})}_F \le \Delta^2 \norm{c_0(V_{00})}_V$, so
  \begin{align*}
    \norm{c_1'(E_{*1})}_E 
    &\le \norm{s_2(V_{01}, E_{*1}) + s_2(V_{01}, E_{0*})}_{E_{*1}} + \norm{u_2(E_{0*})}_F + \norm{s_2(V_{00}, E_{*0}) + s_2(V_{00}, E_{0*})}_F \\
    &\le \Delta \norm{c_0(V_{01})}_V + \Delta \norm{u_2(E_{0*})}_E + \Delta^2 \norm{c_0(V_{00})}_V.
  \end{align*}
  Notice we do not get better bound by considering $\norm{s_2(V_{00}, E_{*0}) + s_2(V_{00}, E_{0*})}_{E_{*1}}$
    because each vertex in $V_{00}$ share a face with $\Delta^2$ distinct edges in $E_{*1}$.

  \paragraph{Bound $\norm{c_1'(E_{1*})}_E$.}
  \begin{align*}
    c_1'(E_{1*})
    &= c_1(E_{1*}) + H_A^\downarrow s_2(V_{00}, E_{*0}) + H_A^\downarrow w_2(V_{00}) \\
    &= H_A^\downarrow s_2(V_{10}, E_{1*}) + H_A^\downarrow s_2(V_{00}, E_{*0}) + H_A^\downarrow w_2(V_{00}) \\
    &= H_A^\downarrow (s_2(V_{10}, E_{1*}) + s_2(V_{10}, E_{*0})) + H_A^\downarrow (s_2(V_{10}, E_{*0}) + s_2(V_{00}, E_{*0})) + H_A^\downarrow w_2(V_{00}) \\
    &= H_A^\downarrow (s_2(V_{10}, E_{1*}) + s_2(V_{10}, E_{*0})) + H_A^\downarrow t_2(E_{*0}) + H_A^\downarrow w_2(V_{00}) \\
    &= H_A^\downarrow (s_2(V_{10}, E_{1*}) + s_2(V_{10}, E_{*0})) + H_A^\downarrow (w_2(V_{00}) + u_2(E_{*0}) + w_2(V_{10})) + H_A^\downarrow w_2(V_{00}) \\
    &= H_A^\downarrow (s_2(V_{10}, E_{1*}) + s_2(V_{10}, E_{*0})) + H_A^\downarrow u_2(E_{*0}).
  \end{align*}
  By \Cref{eq:approx2-2}, $\norm{s_2(V_{10}, E_{1*}) + s_2(V_{10}, E_{*0})}_{E_{1*}} \le \Delta \norm{c_0(V_{10})}_V$, so
  \begin{equation*}
    \norm{c_1'(E_{1*})}_E \le \norm{s_2(V_{10}, E_{1*}) + s_2(V_{10}, E_{*0})}_{E_{1*}} + \norm{u_2(E_{*0})}_F \le \Delta \norm{c_0(V_{10})}_V + \Delta \norm{u_2(E_{*0})}_E.
  \end{equation*}

  So overall,
  \begin{equation*}
    \norm{c_1'}_E \le (\Delta^2 + \Delta) \norm{c_0}_V + \Delta \norm{u}_E.
  \end{equation*}
  Combine with the inequality obtained from the co-expansion theorem \Cref{eq:approx3-1}, $\norm{u}_E \le \frac{1}{\beta} \norm{c^2}_F \le \frac{\Delta^2}{\beta} \norm{c_0}_V$, we obtain
  \begin{equation*}
    \norm{c_1'}_E \le (\Delta^2 + \Delta + \frac{\Delta^3}{\beta}) \norm{c_0}_V.
  \end{equation*}

  \paragraph{Show $c_2$ is small.}
  Recall $c_2 = s_2(V_{00}, E_{*0}) + w_2(V_{00})$, so
  \begin{equation*}
    \norm{c_2}_F \le \norm{s_2(V_{00}, E_{*0})}_F + \norm{w_2(V_{00})}_F
    \le \norm{s_2(V_{00}, E_{*0})}_F + \Delta^2 \norm{w_2(V_{00})}_V.
  \end{equation*}
  Because \Cref{eq:approx1-1},
    $\norm{s_2(V_{00}, E_{*0})}_F \le \Delta \norm{c_1(E_{*0})}_E$.
  Combine with the inequality from the co-expansion theorem \Cref{eq:approx3-2}, $\norm{w}_V \le \gamma \norm{c^1}_E \le \gamma \norm{c_1}_E$, we obtain
  \begin{equation*}
    \norm{c_2}_F \le (\Delta + \Delta^2 \gamma) \norm{c_1}_E.
  \end{equation*}

  This implies that the chain complex is a $(\alpha', \beta', \gamma')$-small-set boundary expander with
  \begin{equation*}
    \alpha' = \frac{\eta}{4\Delta(m_a + m_b)}, \beta' = \frac{1}{\Delta^2 + \Delta + \frac{\Delta^3}{\beta}}, \gamma' = \Delta + \Delta^2 \gamma.
  \end{equation*}
  ($\alpha'$ is chosen so that $\norm{c_1}_E \le \alpha' |X(\cG, C_A, C_B)(1)| \implies \norm{c^1}_E \le \alpha |X(\cG, C_A^\perp, C_B^\perp)(1)|$. 
  Note that \\ $|X(\cG, C_A, C_B)(1)| = 2\Delta(m_a + m_b)|G|$ and $|X(\cG, C_A^\perp, C_B^\perp)(1)| = 2\Delta(k_a + k_b)|G|$.)

  To show systolic distance, one follows the same argument with a few modifications.
  In Step 3, from the expansion assumption, we instead have $u_2 = 0$. 
  In the very end where we show $c'_1$ is small, we instead obtain $c'_1 = 0$, i.e. $c_1 = \partial_2 c_2$.
  Therefore, if $X(\cG_2, C_A^\perp, C_B^\perp)$ has co-systolic distance $\frac{\eta}{2\Delta(k_a + k_b)}$, then $X(\cG_2, C_A, C_B)$ has systolic distance $\frac{\eta}{2\Delta(m_a + m_b)}$.
\end{proof}

\subsubsection{Small-set (co)-expansion}

\begin{corollary}\label{cor:co-expansion}
Under the same assumptions as Theorem~\ref{thm:co-expansion}, suppose further that  $d_1 d_2 - \lambda d_2 - 8 \lambda \Delta > 0$.
Then the co-chain complex~\eqref{eq:complex-1} has the following properties. 
\begin{itemize}
    \item It is a $(\frac{\eta}{4\Delta (m_a + m_b)}, \frac{\eta}{2})$-small-set co-locally-minimal expander. 
    \item It is a $(\frac{\eta}{4\Delta (m_a + m_b)}, \frac{\eta}{2}, \frac{1}{d_1-\lambda-\frac{1}{2}})$-small-set co-boundary expander, where   $\eta = \frac{d_1 d_2 - \lambda d_2 - 8 \lambda \Delta}{\Delta d_2/2 + 2\Delta}$.
  \end{itemize}
More generally given any co-chain $c^1$ it is possible to write $c^1 = c^1{}' + \delta^0 c^0$ in such a way that
  \begin{equation} \label{eq:co-boundary-expansion}
    \begin{gathered}
      \norm{\delta^1 c^1}_F \ge \frac{d_1 d_2 - \lambda d_2 - 8 \lambda \Delta}{4 d_2 + 8 \Delta} \norm{c^1{}'}_E - \frac{\Delta d_2/2 + 2\Delta}{4 d_2 + 8 \Delta} \frac{\norm{c^1{}'}_E^2}{|G|}\;, \\
      \norm{c^1}_E \ge (d_1-\lambda) \norm{c^0}_V - \frac{\Delta}{4} \frac{\norm{c^0}_V^2}{|G|}\;.
    \end{gathered}
  \end{equation}
\end{corollary}

\begin{proof}
We start with~\eqref{eq:co-boundary-expansion}. The first inequality follows directly from~\eqref{eq:co-locally-minimal-expansion}. To show the second, let $c^1$ be any co-chain and let let $c^1{}'$ be of minimal weight among elements $c^1 + B^1$. Then $\norm{c^1{}'}_E \le \norm{c^1}_E$, and so by the triangle inequality we have $\norm{\delta^0 c^0}_E = \norm{c^1{}' + c^1}_E \le \norm{c^1{}'}_E + \norm{c^1}_E \le 2 \norm{c^1}_E$.
  Combined with \Cref{lem:co-expansion-1d} we obtain the second inequality in \Cref{eq:co-boundary-expansion}. 

  Next, we plug in specific numbers in \Cref{eq:co-locally-minimal-expansion} and 
    \Cref{eq:co-boundary-expansion} to obtain 
    small-set co-locally-minimal expansion
    and small-set co-boundary expansion.

 First, small-set co-locally-minimal expansion. For general chain with $\norm{c^1}_E \le \frac{1}{2} \eta |G|$, we have the expansion
  \begin{equation*}
    \norm{c^2}_F \ge \frac{1}{2} \eta \norm{c^1}_E.
  \end{equation*}

  Next, small-set co-boundary expansion. For general chain with $\norm{c^1}_E \le \frac{1}{2} \eta |G|$, 
  we here simplify the inequality 
  $\norm{c^1}_E \ge (d_1-\lambda) \norm{c^0}_V - \frac{\Delta}{4} \frac{\norm{c^0}_V^2}{|G|}$.
  Because $\eta = \frac{d_1 d_2 - \lambda d_2 - 8 \lambda \Delta}{\Delta d_2/2 + 2\Delta} \le \frac{d_1 - \lambda}{\Delta/2}$,
  we have $(d_1-\lambda) \norm{c^0}_V - \frac{\Delta}{4} \frac{\norm{c^0}_V^2}{|G|} \le \norm{c^1}_E \le \frac{d_1-\lambda}{\Delta} |G|$.
  This implies $\norm{c^0}_V \le \frac{2|G|}{\Delta}$ for large enough $\Delta$.
  When plug back in $\norm{c^1}_E \ge (d_1-\lambda) \norm{c^0}_V - \frac{\Delta}{4} \frac{\norm{c^0}_V^2}{|G|}$, we obtain
  \begin{equation*}
    \norm{c^1}_E \ge (d_1-\lambda-\frac{1}{2}) \norm{c^0}_V.
  \end{equation*}
  So overall
  \begin{equation*}
    \norm{c^2}_F \ge \frac{1}{2} \eta \norm{c^1{}'}_E
  \end{equation*}
  and 
  \begin{equation*}
    \norm{c^0}_V \le \frac{1}{d_1-\lambda-\frac{1}{2}} \norm{c^1}_E.
  \end{equation*}
\end{proof}

\section{Linear time decoder} \label{sec:decoder}

In this section we construct a linear time decoder for the quantum code $\cC$ introduced in Section~\ref{sec:construction}.
As discussed in the introduction, one can separate the task of decoding into two.
We call one of them the decoder and the other the co-decoder: 
the decoder recovers $\tilde c_1$ given the syndrome $\partial_1 c_1$ such that $\tilde c_1 \in c_1 + B_1$;
the co-decoder recovers $\tilde c^1$ given the syndrome $\delta^1 c^1$ such that $\tilde c^1 \in c^1 + B^1$.

This section parallels the section on distance with similar proof techniques.
We first show the existence of a linear time co-decoder.
Then we use the linear time co-decoder to obtain a linear time decoder.

\begin{restatable} [Co-Decoder]{theorem}{CoDecoder} \label{thm:co-decoder}
  $X(\cG_2, C_A, C_B)$ has a linear time co-decoder up to distance 
  $\frac{\kappa}{2 \Delta (m_a+m_b)} |X(1)|$
  where $\kappa = \frac{\Delta d_2/2 + 2 \Delta}{8 \Delta d_2 + 16 \Delta^2} \eta' \eta$, $\eta = \frac{d_1 d_2 - \lambda d_2 - 8 \lambda \Delta}{\Delta d_2/2 + 2\Delta}$, $\eta' = \frac{d_1 d_2/4 - \lambda d_2/2 - 8 \lambda \Delta}{\Delta d_2/4 + 2 \Delta}$.
\end{restatable}

\begin{restatable} [Co-Decoder $\to$ Decoder]{theorem}{CoDecoderToDecoder} \label{thm:co-expansion-implies-expansion2}
  If $X(\cG_2, C_A^\perp, C_B^\perp)$ has a linear time co-decoder up to distance $\eta'' |G|$,
  then $X(\cG_2, C_A, C_B)$ has a linear time decoder up to distance $\frac{\eta''}{6 + 4 \Delta/d_2} |G|$.
\end{restatable}

Together 
  we obtain a linear time decoder for $\cC(\cG_2, C_A, C_B)$ up to distance $\frac{\kappa}{4 \Delta (m_a+m_b)(6 + 4 \Delta/d_2)} |X(1)|$. 

\subsection{Co-Decoder}

\Cref{thm:co-decoder} is the main theorem we will show in this subsection. 
We discuss the construction, the correctness, and the running time of the decoder which together proves the theorem.

\paragraph{Construction:}

The co-decoder in the direction of the co-chain complex $\F2^{X(2)} \gets \F2^{X(1)} \gets \F2^{X(0)}$
  is the small-set-flip decoder introduced in \cite{leverrier2015quantum}.
The small-set-flip decoder is a generalization of the local-flip decoder for the expander codes \cite{sipser1996expander}
  where the decoder observes a local region and make local changes that reduce the weight of the syndrome.

\begin{algorithm}[H]
  \SetAlgoLined
  \begin{enumerate}
    \item (Initialization) $c^2_0 \coloneqq c^2$.
    \item (Main loop) In the $i$-th iteration, if there is a vertex $v_i$ with $e^1_i$ supported on $E(v_i)$ such that $\abs{c^2_i + \delta^1 e^1_i} < \abs{c^2_i}$,
      set $c^2_{i+1} \coloneqq c^2_i + \delta^1 e^1_i$ and repeat.
    \item (End) Output $\tilde c^1 \coloneqq \sum e^1_i$.
  \end{enumerate}
  \caption{Simple small-set-flip decoder. (Input: $c^2 \in \F2^{X(2)}$)}
  \label{alg:simple-small-set-flip-decoder}
\end{algorithm}

Besides these variables, we define other variables not directly known by the decoder.
Let $c^1_0$ be the minimal chain in $c^1 + B^1$ and $c^1_{i+1}$ be the minimal chain in $c^1_i + e^1_i + B^1$.
One can interpret $c^1_i$ as the error at the $i$-th iteration and $c^2_i$ as corresponding syndrome.
Note that the decoder knows the syndrome $c^2_i$ but not the error $c^1_i$.

Recall that the final syndrome of a local-flip decoder is locally minimal.
Similarly, the final syndrome of a small-set-flip decoder satisfies a similar property which we call extended local minimality.

\begin{definition} [Extended Co-Locally Minimal]
  We say $c^2 \in \F2^{X(2)}$ is co-locally minimal from $X(0)$ if
  \begin{equation*}
    \forall v \in V, c^1 \in \F2^{X(1)}, \supp(c^1) \subset E(v) : \norm{c^2}_F \le \norm{c^2 + \delta^1 c^1}_F,
  \end{equation*}
  where $E(v) \subset E$ are the edges incident to $v$.
\end{definition}

\paragraph{Correctness:}

Following the standard proof strategy for local-flip decoder
  we show two lemmas. The first lemma, \Cref{lem:co-expansion2}, shows that whenever there is a non-zero error with small weight the decoder is able to continue running and reduce the weight of the syndrome. The second lemma, \Cref{lem:short-remains-short}, shows that as long as the initial error is sufficiently small the error at each iteration remains small.
Together the two lemma imply that there is no syndrome when the decoder stops, and that the decoder correctly corrects the error, i.e.\ $\tilde c^1 = \sum e^1_i \in c^1 + B^1$.

\begin{lemma} [Reducible if Short] \label{lem:co-expansion2}
  Given $\Delta$-regular $\lambda$-spectral expander graphs $\Cay(G, A)$, $\Cay(G, B)$
  and linear codes $C_A^\perp, C_B^\perp$ of length $\Delta$ with distance $d_1$ and $(C_A^\perp, C_B^\perp)$ with robustness $d_2$. Assume further that  
	\[ d_1 d_2/4 - \lambda d_2/2 - 8 \lambda \Delta > 0\;.\]
Consider the co-chain complex
  \begin{equation*}
    X(\cG_2, C_A, C_B)\colon \F2^F \xleftarrow{\delta^1} (\F2^{m_a})^{E^-} \oplus (\F2^{m_b})^{E^|} \xleftarrow{\delta^0} (\F2^{m_a \times m_b})^V.
  \end{equation*}
	Then for every co-locally minimal $c^1 \in \F2^{X(1)}$ such that $c^2 = \delta c^1$ is co-locally minimal from $X(0)$, if
	\[\norm{c^1}_E < \frac{d_1 d_2/4 - \lambda d_2/2 - 8 \lambda \Delta}{\Delta d_2/4 + 2 \Delta}  \frac{|X(1)| }{2 \Delta (m_a+m_b)} \]
	then $c^1 = 0$.
\end{lemma}

\begin{lemma} [Short Remains Short] \label{lem:short-remains-short}
Let $c^1\in \F2^{X(1)}$ be such that 
\[\norm{c^1}_E \le \frac{\kappa}{2 \Delta (m_a+m_b)} |X(1)|\;.\]
For $i\geq 1$ let $c^1_i$ be as defined below Algorithm 1. Then 
\[\norm{c^1_i}_E \le \frac{\eta'}{2 \Delta (m_a+m_b)} |X(1)|\;,\]
 where 
\[\kappa = \frac{\Delta d_2/2 + 2 \Delta}{8 \Delta d_2 + 16 \Delta^2} \eta' \eta\;,\quad \eta = \frac{d_1 d_2 - \lambda d_2 - 8 \lambda \Delta}{\Delta d_2/2 + 2\Delta}\;,\quad \eta' = \frac{d_1 d_2/4 - \lambda d_2/2 - 8 \lambda \Delta}{\Delta d_2/4 + 2 \Delta}\;.\]
\end{lemma}

We first assume the two lemmas and show the correctness.

\begin{proof} [Proof of Correctness]
  Suppose the co-decoder stops at the $T$-th iteration.
  By \Cref{lem:short-remains-short}, because $\norm{c^1}_E \le \\ \frac{\kappa}{4 \Delta (m_a+m_b)} |X(1)|$, we have $\norm{c^1_i}_E \le \frac{\eta'}{4 \Delta (m_a+m_b)} |X(1)|$.
  By \Cref{lem:co-expansion2}, the co-decoder stops when $c^1_T = 0$.
  Notice $c^1_T \in c^1 + \sum_{i=0}^{T-1} e^1_i + B^1$,
    so $\tilde c^1 = \sum_{i=0}^{T-1} e^1_i \in c^1 + B^1$.
  Because the output differs from the error by a co-boundary,
    this implies the co-decoder decodes correctly.
\end{proof}

Now we prove the two lemmas. We first show the second lemma which is simpler.

\begin{proof} [Proof of \Cref{lem:short-remains-short}]
  Because the number of syndrome strictly decreases, we have $\norm{c^2_i}_F \le \norm{c^2}_F$.
  Because each edge is incident to $\Delta$ faces, we have $\norm{c^2}_F \le \Delta \norm{c^1}_E$.
  Combine the two equations with \Cref{eq:co-locally-minimal-expansion}
  we have
  \begin{equation*}
    \Delta \norm{c^1}_E
    \ge \frac{d_1 d_2 - \lambda d_2 - 8 \lambda \Delta}{4 d_2 + 8 \Delta} \norm{c^1_i}_E - \frac{\Delta d_2/2 + 2\Delta}{4 d_2 + 8 \Delta} \frac{\norm{c^1_i}_E^2}{|G|}
    = \frac{\Delta d_2/2 + 2 \Delta}{4 d_2 + 8 \Delta} \norm{c^1_i}_E (\eta-\frac{\norm{c^1_i}_E}{|G|}).
  \end{equation*}
  Therefore, to have $\norm{c^1_i}_E \le \eta' |G|$, we suffice to set
  \begin{equation*}
    \norm{c^1}_E \le \frac{\Delta d_2/2 + 2 \Delta}{4 \Delta d_2 + 8 \Delta^2} \eta' (\eta-\eta') |G|
    \le \frac{\Delta d_2/2 + 2 \Delta}{8 \Delta d_2 + 16 \Delta^2} \eta' \eta |G|
  \end{equation*}
  where the last inequality follows from $\eta' < \frac{1}{2} \eta$.
\end{proof}

We end with the proof of  \Cref{lem:co-expansion2}, which is very similar to the proof of \Cref{thm:co-expansion}.

\begin{proof} [Proof of \Cref{lem:co-expansion2}]
 Let $c^1 \in \F2^{X(1)}$ be co-locally minimal and furthermore such that $c^2 = \delta c^1$ is co-locally minimal from $X(0)$. Let $\cE \subset E$ be the support of $c^1$.
 
 \paragraph{Step 1: Define ``neighbors'' M.}
  We define 
  \begin{equation*}
    M = \frac{d_2}{2} M_1 + M_0.
  \end{equation*}

  \paragraph{Step 2: Upper bound from expansion.}
  Using Lemma~\ref{lem:M_0} and Lemma~\ref{lem:M_1}, 
		\begin{equation*}
    1_\cE^T M 1_\cE \,\le \,\lambda(\frac{d_2}{2} + 8\Delta) |\cE| + \frac{\Delta}{|G|}(\frac{d_2}{4} + 2)|\cE|^2\;.
  \end{equation*}

  \paragraph{Step 3: Lower bound from distance and robustness.}
  We use two corollaries to Claim~\ref{claim:64} and~\ref{claim:65} respectively, which make use of the additional assumption that $c^2 = \delta c^1$ is co-locally minimal from $X(0)$. First, we show that
  \begin{equation} \label{eq:ineq2-from-distance-decoder}
    \norm{c^1(E_{*1}(e_{*0}))}_E + \norm{c^1(E_{0*}(e_{*0}))}_E + \norm{c^1(E_{1*}(e_{*0}))}_E \ge \frac{d_1}{2} \norm{c^1(e_{*0})}_E\;.
  \end{equation}
	This inequality is shown as follows. Using~\eqref{eq:ineq2-from-distance},
	\[\norm{c^2(F(e_{*0}))}_F \ge d_1 \norm{c^1(e_{*0})}_E - (\norm{c^1(E_{*1}(e_{*0}))}_E + \norm{c^1(E_{0*}(e_{*0}))}_E + \norm{c^1(E_{1*}(e_{*0}))}_E)\;.\]
	On the other hand, by definition of $c^2 = \delta^1 c^1$,
	\begin{align*}
	\norm{(c^2 + \delta^1 c^1(e_{*0}))(F(e_{*0}))}_F &\leq \norm{c^1(E_{*1}(e_{*0}))}_E + \norm{c^1(E_{0*}(e_{*0}))}_E + \norm{c^1(E_{1*}(e_{*0}))}_E\;.
	\end{align*}
		Using that $c_2$ is co-locally minimal, 
		\[ \norm{(c^2 + \delta^1 c^1(e_{*0}))(F(e_{*0}))}_F \,\geq\, \norm{c^2(F(e_{*0}))}_F\;.\]
		Thus~\eqref{eq:ineq2-from-distance-decoder} follows. Similar reasoning shows the following as a consequence of~\eqref{eq:ineq2-from-robust}.
	  \begin{equation} \label{eq:ineq2-from-robust-decoder}
    \norm{c^1(E_{*1}(v_{00}))}_E + \norm{c^1(E_{1*}(v_{00}))}_E \ge \frac{d_2}{2} (\norm{c^1(E_{*0}(v_{00}))}_E + \norm{c^1(E_{0*}(v_{00}))}_E)\;.
  \end{equation}
Combining \Cref{eq:ineq2-from-distance-decoder} and \Cref{eq:ineq2-from-robust-decoder}, 
  \begin{align*}
    & 1_\cE^T |M|1_{e_{*0}}\\
    &= 1_\cE^T \Big(\frac{d_2}{2} M_1 + M_0\Big)1_{e_{*0}} \\
    &= \frac{d_2}{2} \norm{c^1(E_{*1}(e_{*0}))}_E + \norm{c^1(E_{*1}(v_{00}))}_E + \norm{c^1(E_{1*}(v_{00}))}_E + \norm{c^1(E_{*1}(v_{10}))}_E + \norm{c^1(E_{0*}(v_{10}))}_E \\
    &\ge \frac{d_2}{2} \norm{c^1(E_{*1}(e_{*0}))}_E + \frac{d_2}{2} \norm{c^1(E_{*0}(v_{00}))}_E + \frac{d_2}{2} \norm{c^1(E_{0*}(v_{00}))}_E + \frac{d_2}{2} \norm{c^1(E_{*0}(v_{10}))}_E + \frac{d_2}{2} \norm{c^1(E_{1*}(v_{10}))}_E \\
    &\ge \frac{d_2}{2} \norm{c^1(E_{*1}(e_{*0}))}_E + \frac{d_2}{2} \norm{c^1(E_{0*}(v_{00}))}_E + \frac{d_2}{2} \norm{c^1(E_{1*}(v_{10}))}_E \\
    &= \frac{d_2}{2} \norm{c^1(E_{*1}(e_{*0}))}_E + \frac{d_2}{2} \norm{c^1(E_{0*}(e_{*0}))}_E + \frac{d_2}{2} \norm{c^1(E_{1*}(e_{*0}))}_E \\
    &\ge \frac{d_1 d_2}{4} \norm{c^1(e_{*0})}_E\;.
  \end{align*}

  \paragraph{Step 4: Combine the upper bound and the lower bound.}
  Finally, we combine the upper bound and the lower bound.
  \begin{equation*}
    \frac{d_1 d_2}{4} \norm{c^1}_E \le 1_\cE^T M 1_\cE
    \le \lambda (\frac{d_2}{2} + 8 \Delta) \norm{c^1}_E + \frac{\Delta}{|G|}(\frac{d_2}{4}+2) \norm{c^1}_E^2\;.
  \end{equation*}
This implies that $c^1 = 0$ whenever
  \begin{equation*}
    \norm{c^1}_E < \frac{d_1 d_2/4 - \lambda d_2/2 - 8 \lambda \Delta}{\Delta d_2/4 + 2\Delta} |G|\;,
  \end{equation*}
	as desired.
\end{proof}

\begin{remark}
  A similar argument applies even when there are measurement errors in $c^2$.
  Suppose that the co-decoder receives input $c^2 + z^2$ instead of $c^2 = \delta^1 c^1$ where $z^2$ is the measurement error with small weight.
  Then by replacing \Cref{eq:ineq2-from-distance-decoder} with
  \begin{equation*}
    \norm{z^2(F(e_{*0}))}_F + \norm{c^1(E_{*1}(e_{*0}))}_E + \norm{c^1(E_{0*}(e_{*0}))}_E + \norm{c^1(E_{1*}(e_{*0}))}_E \ge \frac{d_1}{2} \norm{c^1(e_{*0})}_E
  \end{equation*}
  and replacing \Cref{eq:ineq2-from-robust-decoder} with
  \begin{equation*}
    \norm{z^2(F(v_{00}))}_F + \norm{c^1(E_{*1}(v_{00}))}_E + \norm{c^1(E_{1*}(v_{00}))}_E \ge \frac{d_2}{2} (\norm{c^1(E_{*0}(v_{00}))}_E + \norm{c^1(E_{0*}(v_{00}))}_E)
  \end{equation*}
  the rest of the argument still holds by replacing the corresponding parameters.
  The final result is that the remaining error after the co-decoder stop is less than the number of measurement errors times a constant.
  This is the so called error reduction property used in the construction of Spielman code (or cascade code) with linear time encoder and decoder \cite{spielman1996linear}.
\end{remark}

\paragraph{Linear Time:}

The decoder presented in \Cref{alg:simple-small-set-flip-decoder} runs in quadratic time.
To get linear time, we perform additional preprocessing.
A vertex $v$ is said to be flippable if there exist a small set flip $e^1$ supported on $E(v)$ which decrease the weight of the current syndrome $c^2_i$.

\begin{algorithm}[H]
  \SetAlgoLined
  \begin{enumerate}
    \item (Initialization) Set $c^2_0 \coloneqq c^2$. Create a list $Q$ which contains all the flippable vertices.
    \item (Main loop) In the $i$-th iteration, take a vertex $v_i$ and its flip $e^1_i$ from the list $Q$. 
    Set $c^2_{i+1}  \coloneqq c^2_i + \delta^1 e^1_i$. 
    Update the list $Q$ by checking the flippability of the vertices neighbor to $v_i$, i.e. $V_{00}(v_i) \cup V_{10}(v_i) \cup V_{01}(v_i) \cup V_{11}(v_i)$. 
    Repeat.
    \item (End) Output $\tilde c^1 \coloneqq \sum c^1_i$.
  \end{enumerate}
  \caption{Full small-set-flip decoder. (Input: $c^2 \in \F2^{X(2)}$)}
  \label{alg:full-small-set-flip-decoder}
\end{algorithm}

Note that this algorithm has a similar behavior from the simple small set flip decoder.
Because the syndrome only updates on $F(v_i)$, the updated list contains all the flippable vertices of the current syndrome.

We now compute the time complexity. Checking flippability of a vertex $v$ requires to try $2^{\Delta (m_a + m_b)}$ possible values of $e^1$ and each requires $O(\Delta^2 (m_a + m_b))$ to check if the weight decreases $\norm{c^2 + \delta^1 e^1}_F < \norm{c^2}_F$.
So the initialization takes time $O(\Delta^2 (m_a + m_b) 2^{\Delta (m_a + m_b)} |V|)$.
For the main loop, there are at most $\norm{c^2}_F$ iterations.
Each iteration checks the flippability of $(\Delta+1)^2$ many vertices.
So the main loop takes time $O(\Delta^4 (m_a + m_b) 2^{\Delta (m_a + m_b)} \norm{c^2}_F)$.
Overall the complexity is linear $\Theta(|X(1)|)$.

\subsection{Decoder}

Now, for decoder, we want to find $\tilde c_1 \in c_1 + B_1$ given $c_0 = \partial_1 c_1$.
Different from the co-decoder, small-set-flip decoder does not work directly.
This is roughly because there are too few information around a face to make decisions.
Nevertheless, similar to the section on expansion, one can translate the results from co-decoder to decoder.
Here is the main theorem for this subsection.

\CoDecoderToDecoder*

\paragraph{Construction:}

We now describes the decoding process. 
Similar to the proof of expansion, the main idea is to make local guesses around the vertex, then correct the inconsistencies between the local guesses.
It is crucial to keep track of the variables that exist but not known to the decoder.
We use $\hat c_1$ to denote such variables
and we use $\tilde c_1$ to denote a good approximation to $\hat c_1$.

\paragraph{Step 1: Construct $s_1$ and $\hat s'_2$ by guessing from the vertex.}

Using the local chain complex around $v_{00}$
\begin{equation*}
  (\F2)^{F(v_{00})} \xrightarrow{\partial_2} (\F2^{m_b})^{E_{*0}(v_{00})} \oplus (\F2^{m_a})^{E_{0*}(v_{00})} \xrightarrow{\partial_1} \F2^{m_a \times m_b}
\end{equation*}
we find $s_1 \in (\F2^{m_a})^{VE^-} \times (\F2^{m_b})^{VE^|}$
where $s_1(v_{00}) \in (\F2^{m_b})^{E_{*0}(v_{00})} \times (\F2^{m_a})^{E_{0*}(v_{00})}$ is the minimal chain
such that
\begin{equation*}
  \partial_1(s_1(v_{00})) = 
  H_A^\uparrow s_1(v_{00}, E_{*0}(v_{00})) + H_B^\leftarrow s_1(v_{00}, E_{0*}(v_{00})) = c_0(F(v_{00})).
\end{equation*}

Note that $s_1$ is a guess of $\hat c_1$ with the same local property
\begin{equation} \label{eq:dec-map1-1}
  \partial_1(\hat c_1(E(v_{00}))) = c_0(F(v_{00})).
\end{equation}
Because $\hat c_1(E(v_{00}))$ is a valid candidate for $s_1(v_{00})$ and $s_1(v_{00})$ has the minimal weight among them, we have $\norm{s_1(v_{00})}_E \le \norm{\hat c_1(E(v_{00}))}_E$, i.e.
\begin{equation} \label{eq:dec-approx1-1}
  \norm{s_1}_{VE} \le 2 \norm{\hat c_1}_E.
\end{equation}
The factor $2$ appears because each edge for $s_1$ has two values from the two endpoints.

Further, because they have the same local property and because the local chain complex is exact,
the difference between $s_1$ and $\hat c_1$ can be expressed using $\hat s'_2 \in (\F2^{n_a \times n_b})^{V}$
where
\begin{equation} \label{eq:dec-map1-2}
  \partial_2(\hat s'_2(v_{00})) + s_1(v_{00}) = \hat c_1(v_{00}).
\end{equation}
Additionally, using \Cref{lem:robust-testability} below
\footnote{In fact, it is sufficient to use the trivial bound $\norm{c_2}_{[n_a] \cup [n_b]} \le \Delta \norm{c_1}_{[n_a] \cup [n_b]}$. The bound in the lemma is to improve $\Delta$ to $\Theta(1)$.}
we can further require
  $\norm{\hat s'_2(v_{00})}_E \le (1 + \frac{\Delta}{d_2}) \norm{s_1(v_{00}) + \hat c_1(E(v_{00}))}_E \le (2 + \frac{2 \Delta}{d_2}) \norm{\hat c_1(E(v_{00}))}_E$, i.e.
\begin{equation} \label{eq:dec-approx1-2}
  \norm{\hat s'_2}_{VE} \le (4 + \frac{4 \Delta}{d_2}) \norm{\hat c_1}_E.
\end{equation}

\begin{restatable} {lemma}{RobustTestability} \label{lem:robust-testability}
  Given two linear codes $C_A, C_B$ with parity-check matrices $H_A, H_B$.
  Let
  \begin{equation*}
    Y(H_A, H_B)\colon \F2^{n_a \times n_b} \xrightarrow{\partial_2} \F2^{n_a \times m_b + m_a \times n_b} \xrightarrow{\partial_1} \F2^{m_a \times m_b}
  \end{equation*}
  be the exact chain complex in \Cref{lem:exact}.
  If $(C_A^\perp, C_B^\perp)$ is $d_2$-robust, then for all $c_1 \in \Ima(\partial_2)$ there exist $c_2 \in \F2^{n_a \times n_b}$ such that
  $\partial_2 c_2 = c_1$ and
  \begin{equation*}
    \norm{c_2}_{[n_a] \cup [n_b]} \le (1 + \frac{\Delta}{d_2}) \norm{c_1}_{[n_a] \cup [n_b]}.
  \end{equation*}
\end{restatable}
\begin{proof}
  The idea is to construct $c_2$ using $c_1$.
  The construction has a similar flavor as for the proof of expansion where one first guess then correct the inconsistencies.

  Let $c_1 = (c_1^a, c_1^b) \in \F2^{m_a \times n_b} \oplus \F2^{n_a \times m_b}$.
  Construct $s_2^a, s_2^b \in \F2^{n_a \times n_b}$
    such that $H_A s_2^a = c_1^a$ and $H_B s_2^b = c_1^b$.
  We pick $s_2^a(i_a) = 0$ when $c_i^a(i_a) = 0$. Similarly for $s_2^b$.
  Therefore,
  \begin{equation*}
    \norm{s_2^a}_{[n_a]} = \norm{c_1^a}_{[n_a]}, \norm{s_2^b}_{[n_b]} = \norm{c_1^b}_{[n_b]}.
  \end{equation*}

  Now, we use robustness to correct the inconsistencies between $s_2^a$ and $s_2^b$.
  Let $t_2 = s_2^a + s_2^b \in \F2^{n_a \times n_b}$.
  We have 
  \begin{equation*}
    \norm{t_2}_{[n_a \times n_b]} \le \Delta (\norm{s_2^a}_{[n_a]} + \norm{s_2^b}_{[n_b]}).
  \end{equation*}
  Using the robustness of the co-chain complex
  $\F2^{n_a \times n_b} \xrightarrow{\delta^1} \F2^{n_a \times k_b + k_a \times n_b} \xrightarrow{\delta^0} \F2^{k_a \times k_b}$
    there exist $u^1_a \in \F2^{k_a \times n_b}, u^1_b \in \F2^{n_a \times k_b}$
    such that $t_2 = (H_A^\perp)^T u^1_a + (H_B^\perp)^T u^1_b$
    with 
  \begin{equation*}
    d_2 (\norm{u^1_a}_{[n_a]} + \norm{u^1_a}_{[n_b]}) \le \norm{t_2}_{[n_a \times n_b]}.
  \end{equation*}

  Finally, we set $c_2 = s_2^a + (H_A^\perp)^T u^1_a = s_2^b + (H_B^\perp)^T u^1_b$.
  It is easy to check $H_A c_2 = H_A s_2^a = c_1^a$ and $H_B c_2 = H_B s_2^b = c_1^b$.
  Because $\norm{c_2}_{[n_a]} \le \norm{s_2^a}_{[n_a]} + \norm{u^1_a}_{[n_a]} \le (1 + \frac{\Delta}{d_2}) \norm{c_1}_{[n_a]}$ and $\norm{c_2}_{[n_b]} \le \norm{s_2^b}_{[n_b]} + \norm{u^1_b}_{[n_b]}  \le (1 + \frac{\Delta}{d_2}) \norm{c_1}_{[n_b]}$,
    we have the desired result
  \begin{equation*}
    \norm{c_2}_{[n_a] \cup [n_b]} \le (1 + \frac{\Delta}{d_2}) \norm{c_1}_{[n_a] \cup [n_b]}.
  \end{equation*}
\end{proof}

\begin{figure}[H]
  \centering
  \begin{tikzpicture} [scale=1.2, every node/.style={font=\footnotesize}]
    \begin{scope}
      \draw (0, 0) rectangle node{$c_2(v_{00})$} ++(1, -1);
      \draw (3.5, 0) rectangle node{$\hat c_1(E_{0*}(v_{00}))$} ++(2, -1);
      \draw (0, -3.5) rectangle node[rotate=90]{$\hat c_1(E_{*0}(v_{00}))$} ++(1, -2);

      \draw[dashed] (3.5, -1.5) rectangle node{$s_1(v_{00}, E_{0*}(v_{00}))$} ++(2, -1);
      \draw[dashed] (1.5, -3.5) rectangle node[rotate=90]{$s_1(v_{00}, E_{*0}(v_{00}))$} ++(1, -2);
      \draw[dashed] (3.5, -3.5) rectangle node{$\hat s'_2(v_{00})$} ++(2, -2);

      \draw[->] (3, -0.5) -- (1.5, -0.5);
      \draw[->] (0.5, -3) -- (0.5, -1.5);

      \draw (1.25, -4.5) node{$=$};
      \draw (4.5, -1.25) node[rotate=90]{$=$};
      \draw (2.75, -4.5) node{$+$};
      \draw (4.5, -2.75) node[rotate=90]{$+$};

      \draw[->] (4.5, -3.25) -- (4.5, -3);
      \draw[->] (3.25, -4.5) -- (3, -4.5);
    \end{scope}
  \end{tikzpicture}
  \caption{Construct $s_1$ and $\hat s'_2$ using $c_2$ and $\hat c_1$. $s_1$ and $\hat s'_2$ satisfy \Cref{eq:dec-map1-1,eq:dec-approx1-1,eq:dec-map1-2,eq:dec-approx1-2}.}
  \label{fig:decoder-step1}
\end{figure}
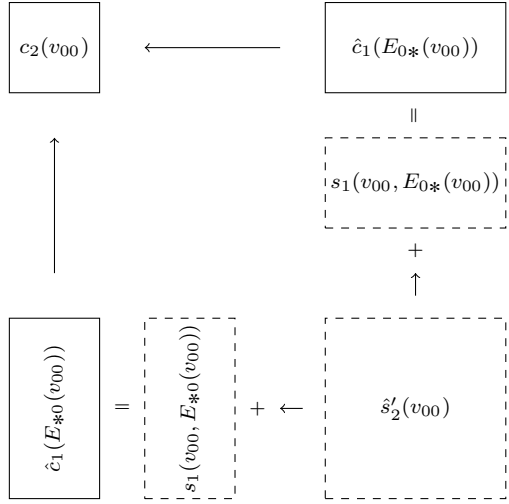

\paragraph{Step 2: Construct $t_1$ and $\hat t'_2$ from the difference of $s_1$ and $\hat s'_2$.}

$t_1 \in (\F2^{m_a})^{E^-} \times (\F2^{m_b})^{E^|}$ is defined as the ``sum'' of $s_1$ viewed from its two endpoints.
That is
\begin{equation*}
  t_1(e_{*0}) = s_1(v_{00}, e_{*0}) + s_1(v_{10}, e_{*0}).
\end{equation*}
We have 
\begin{equation} \label{eq:dec-approx2-1}
  \norm{t_1}_E \le \norm{s_1}_{VE}.
\end{equation}

Similarly, we define
$\hat t'_2 \in (\F2^{n_a})^{E^-} \times (\F2^{n_b})^{E^|}$
as
\begin{equation*}
  \hat t'_2(e_{*0}) = \hat s'_2(v_{00}, e_{*0}) + \hat s'_2(v_{10}, e_{*0}).
\end{equation*}
We have 
\begin{equation} \label{eq:dec-approx2-2}
  \norm{\hat t'_2}_E \le \norm{\hat s'_2}_{VE}.
\end{equation}

Notice that by construction, we have
\begin{equation} \label{eq:dec-map2-1}
  H_B^\leftarrow \hat t'_2(e_{*0}) = t_1(e_{*0}),
  H_A^\uparrow \hat t'_2(e_{0*}) = t_1(e_{0*}),
\end{equation}
and
\begin{equation} \label{eq:dec-map2-2}
  \begin{split}
    \hat t'_2(E_{*0}) + \hat t'_2(E_{0*}) + \hat t'_2(E_{*1}) + \hat t'_2(E_{1*})
    &= (\hat s'_2(V_{00}) + \hat s'_2(V_{10})) + (\hat s'_2(V_{00}) + \hat s'_2(V_{01})) \\
    &+ (\hat s'_2(V_{01}) + \hat s'_2(V_{11})) + (\hat s'_2(V_{10}) + \hat s'_2(V_{11})) = 0.
  \end{split}
\end{equation}

\begin{figure}[H]
  \centering
  \begin{tikzpicture} [scale=1.1, every node/.style={font=\footnotesize}]
    \begin{scope}
      \draw (0, 0) rectangle node[rotate=90]{$\hat c_1(e_{*0})$} ++(1, -2);
      \draw (1.5, 0) rectangle node[rotate=90]{$s_1(v_{00}, e_{*0})$} ++(1, -2);
      \draw (3.5, 0) rectangle node{$\hat s'_2(v_{00}, e_{*0})$} ++(2, -2);

      \draw (1.25, -1) node{$=$};
      \draw (2.75, -1) node{$+$};

      \draw[->] (3.25, -1) -- (3, -1);
    \end{scope}
    \begin{scope} [shift={(0,-2.5)}]
      \draw[dashed] (0, 0) rectangle node[rotate=90]{$0$} ++(1, -2);
      \draw[dashed] (1.5, 0) rectangle node[rotate=90]{$t_1(e_{*0})$} ++(1, -2);
      \draw[dashed] (3.5, 0) rectangle node{$\hat t'_2(e_{*0})$} ++(2, -2);

      \draw (1.25, -1) node{$=$};
      \draw (2.75, -1) node{$+$};

      \draw[->] (3.25, -1) -- (3, -1);
    \end{scope}
    \begin{scope} [shift={(0,-5)}]
      \draw (0, 0) rectangle node[rotate=90]{$\hat c_1(e_{*0})$} ++(1, -2);
      \draw (1.5, 0) rectangle node[rotate=90]{$s_1(v_{10}, e_{*0})$} ++(1, -2);
      \draw (3.5, 0) rectangle node{$\hat s'_2(v_{10}, e_{*0})$} ++(2, -2);

      \draw (1.25, -1) node{$=$};
      \draw (2.75, -1) node{$+$};

      \draw[->] (3.25, -1) -- (3, -1);
    \end{scope}
  \end{tikzpicture}
  \caption{Construct $t_1$ and $\hat t'_2$ using $s_1$ and $\hat s'_2$. $t_1$ and $\hat t'_2$ satisfy \Cref{eq:dec-map2-1,eq:dec-map2-2,eq:dec-approx2-1,eq:dec-approx2-2}.}
  \label{fig:decoder-step2}
\end{figure}
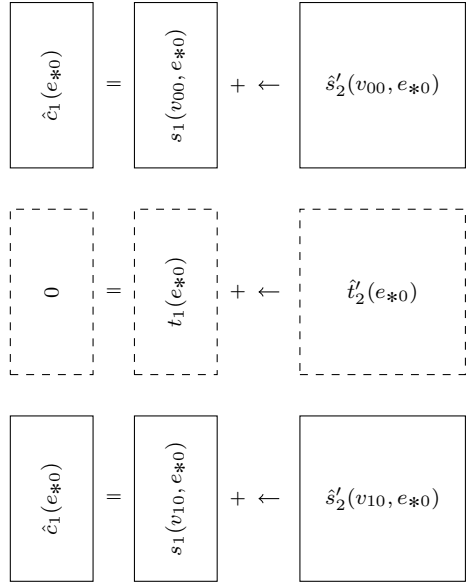

\paragraph{Step 3: Construct $u_2$ by guessing from the edge.}

We now make a second guess where we guess $\hat t'_2$.
We find $u_2 \in (\F2^{n_a})^{E^-} \times (\F2^{n_b})^{E^|}$
such that
\begin{equation*}
  H_B^\leftarrow u_2(e_{*0}) = t_1(e_{*0}), 
  H_A^\uparrow u_2(e_{0*}) = t_1(e_{0*}).
\end{equation*}
We pick $u_2(e_{*0})$ to be $0$ whenever $t_1(e_{*0})=0$, so we have $\norm{u_2}_E = \norm{t_1}_E$, which together with \Cref{eq:dec-approx1-1,eq:dec-approx2-1} implies
\begin{equation} \label{eq:dec-approx3-1}
  \norm{u_2}_E \le 2 \norm{\hat c_1}_E.
\end{equation}

We again track the difference
$\hat u'_2 \in (\F2^{n_a})^{E^-} \times (\F2^{n_b})^{E^|}$
where
\begin{equation} \label{eq:dec-equal3-1}
  \hat u'_2 = \hat t'_2 + u_2
\end{equation}
which satisfies
\begin{equation} \label{eq:dec-map3-1}
  H_B^\leftarrow \hat u'_2(e_{*0}) = 0,
  H_A^\uparrow \hat u'_2(e_{0*}) = 0,
\end{equation}
and
\begin{equation} \label{eq:dec-equal3-2}
  \hat u'_2(E_{*0}) + \hat u'_2(E_{0*}) + \hat u'_2(E_{*1}) + \hat u'_2(E_{1*}) = u_2(E_{*0}) + u_2(E_{0*}) + u_2(E_{*1}) + u_2(E_{1*}).
\end{equation}
Combining the bound of the weight using \Cref{eq:dec-approx1-2,eq:dec-approx2-2} and \Cref{eq:dec-approx3-1}, we have 
\begin{equation} \label{eq:dec-approx3-2}
  \norm{\hat u'_2}_E \le \norm{\hat t'_2}_E + \norm{u_2}_E \le (6 + \frac{4 \Delta}{d_2}) \norm{\hat c_1}_E.
\end{equation}

\paragraph{Step 4: Apply co-decoder assumption.}

This is the main step where we use the co-decoder assumption and obtain an approximation of $\hat u'_2$ with $\tilde u'_2$.

Let $H_A^\perp\colon \F2^{n_a} \to \F2^{k_a}$, $H_B^\perp\colon \F2^{n_b} \to \F2^{k_b}$ be the parity-check matrices of $C_A^\perp$ and $C_B^\perp$ where $k_a=n_a-m_a$ and $k_b=n_b-m_b$. 
The following chain complex is the one we will apply the co-decoder assumption
\begin{equation*}
  X(\cG_2, C_A^\perp, C_B^\perp)\colon \F2^F \xleftarrow{\delta_1} (\F2^{k_a})^{E^-} \oplus (\F2^{k_b})^{E^|} \xleftarrow{\delta_0} (\F2^{k_a \times k_b})^V.
\end{equation*}
The overall strategy is to first move to this new chain complex by constructing $\hat c^1$ and $c^2$. 
Then, use co-decoder to obtain $\tilde c^1 \in \hat c^1 + B^1$ from $c^2$.
Finally, we move back to the original chain complex.

We first construct $\hat c^1$.
\Cref{eq:dec-map3-1} implies one can find $\hat c^1\in (\F2^{k_a})^{E^-} \times (\F2^{k_b})^{E^|}$ such that 
\begin{equation*}
  \hat u_2'(e_{*0}) = (H_B^\perp)^T \hat c^1(e_{*0}),
  \hat u'_2(e_{0*}) = (H_A^\perp)^T \hat c^1(e_{0*}).
\end{equation*}
We now construct $c^2 \in (\F2)^F$ where
\begin{equation*}
  c^2 = u_2(E_{*0}) + u_2(E_{0*}) + u_2(E_{*1}) + u_2(E_{1*}).
\end{equation*}
By \Cref{eq:dec-equal3-2}, $u_2(E_{*0}) + u_2(E_{0*}) + u_2(E_{*1}) + u_2(E_{1*}) = \hat u'_2(E_{*0}) + \hat u'_2(E_{0*}) + \hat u'_2(E_{*1}) + \hat u'_2(E_{1*})$, and by construction $\hat u'_2(E_{*0}) + \hat u'_2(E_{0*}) + \hat u'_2(E_{*1}) + \hat u'_2(E_{1*}) = \delta^1 \hat c^1$, together we have
\begin{equation*}
  c^2 = \delta^1 \hat c^1.
\end{equation*}

We are almost ready to apply the co-decoder.
To apply the co-decoder
  we need to check $\hat c^1$ has small weight.
Because $(H_B^\perp)^T$ is injective, $\norm{\hat c^1(e_{*0})}_E \le \norm{\hat u'_2(e_{*0})}_E$,
  together with \Cref{eq:dec-approx3-1} implies 
\begin{equation*}
  \norm{\hat c^1}_E \le \norm{\hat u'_2}_E \le (6 + \frac{4 \Delta}{d_2}) \norm{\hat c_1}_E.
\end{equation*}

Given that $\hat c^1$ has small weight, we can now use the hypothesis of the co-decoder and obtain $\tilde c^1 \in \hat c^1 + B^1$, i.e.
  there exists $c^0 \in (\F2^{k_a \times k_b})^V$
  such that $\tilde c^1 = \hat c^1 + \delta^0 \hat c^0$.

To move back to the original chain complex, 
we define $\tilde u'_2 \in (\F2^{n_a})^{E^-} \times (\F2^{n_b})^{E^|}$ and $\hat w_2 \in (\F2^{n_a \times n_b})^V$ where
\begin{equation*}
  \tilde u'_2(e_{*0}) = (H_B^\perp)^T \tilde c^1(e_{*0}),
  \tilde u'_2(e_{0*}) = (H_A^\perp)^T \tilde c^1(e_{0*}),
\end{equation*}
\begin{equation*}
  \hat w_2(v_{00}) = (H_A^\perp \otimes H_B^\perp)^T \hat c^0(v_{00}).
\end{equation*}
Because $\tilde c^1 = \hat c^1 + \delta^0 \hat c^0$, we have
\begin{equation} \label{eq:dec-equal4-1}
  \tilde u'_2(E_{*0}) = \hat w_2(V_{00}) + \hat u'_2(E_{*0}) + \hat w_2(V_{10}).
\end{equation}


\paragraph{Step 5: Recover $\tilde t'_2$ from $\tilde u'_2$.}

We set $\tilde t'_2 = u_2 + \tilde u'_2$.
By \Cref{eq:dec-equal3-1} $\hat t'_2 = u_2 + \hat u'_2$ and \Cref{eq:dec-equal4-1}, we have
\begin{equation} \label{eq:dec-equal5-1}
  \tilde t'_2(E_{*0}) = \hat w_2(V_{00}) + \hat t'_2(E_{*0}) + \hat w_2(V_{10}).
\end{equation}

\paragraph{Step6: Obtain $\tilde c_1$.}

Finally, we set
\[ \tilde c_1(E_{*0}) = s_1(V_{00}, E_{*0}), \]
\[ \tilde c_1(E_{0*}) = s_1(V_{00}, E_{0*}), \]
\[ \tilde c_1(E_{*1}) = s_1(V_{01}, E_{*1}) + H_B^\rightarrow \tilde t'_2(E_{0*}), \]
\[ \tilde c_1(E_{1*}) = s_1(V_{10}, E_{1*}) + H_A^\downarrow \tilde t'_2(E_{*0}). \]

This concludes the construction of the decoder.

\paragraph{Correctness:}

To show correctness, we need to show $\tilde c_1 \in \hat c_1 + B_1$.
We now show 
\begin{equation*}
  \tilde c_1 = \hat c_1 + \partial_2 (\hat s'_2(V_{00}) + \hat w_2(V_{00}))
\end{equation*}
by checking the equality on $\tilde c_1(E_{*0})$, $\tilde c_1(E_{0*})$, $\tilde c_1(E_{*1})$, $\tilde c_1(E_{1*})$ individually.

\paragraph{Check $\tilde c_1(E_{*0})$.}

\begin{equation*}
  \hat c_1(E_{*0}) + H_B^\leftarrow \hat s'_2(V_{00}) + H_B^\leftarrow \hat w_2(V_{00})
  = \hat c_1(V_{00}, E_{*0}) + H_B^\leftarrow \hat s'_2(V_{00}) + 0
  = s_1(V_{00}, E_{*0})
  = \tilde c_1(E_{*0}).
\end{equation*}

\paragraph{Check $\tilde c_1(E_{0*})$.}

\begin{equation*}
  \hat c_1(E_{0*}) + H_A^\uparrow \hat s'_2(V_{00}) + H_A^\uparrow \hat w_2(V_{00})
  = \hat c_1(V_{00}, E_{0*}) + H_A^\uparrow \hat s'_2(V_{00}) + 0
  = s_1(V_{00}, E_{0*})
  = \tilde c_1(E_{0*}).
\end{equation*}

\paragraph{Check $\tilde c_1(E_{*1})$.}

\begin{align*}
  \hat c_1(E_{*1}) + H_B^\rightarrow \hat s'_2(V_{00}) + H_B^\rightarrow \hat w_2(V_{00})
  &= \hat c_1(V_{01}, E_{*1}) + H_B^\rightarrow \hat s'_2(V_{00}) + H_B^\rightarrow \hat w_2(V_{00}) \\
  &= s_1(V_{01}, E_{*1}) + H_B^\rightarrow \hat s'_2(V_{01}) + H_B^\rightarrow \hat s'_2(V_{00}) + H_B^\rightarrow \hat w_2(V_{00}) \\
  &= s_1(V_{01}, E_{*1}) + H_B^\rightarrow \hat t'_2(E_{0*}) + H_B^\rightarrow \hat w_2(V_{00}) \\
  &= s_1(V_{01}, E_{*1}) + H_B^\rightarrow \tilde t'_2(E_{0*}) + H_B^\rightarrow \hat w_2(V_{01}) \\
  &= s_1(V_{01}, E_{*1}) + H_B^\rightarrow \tilde t'_2(E_{0*}) + 0 \\
  &= \tilde c_1(E_{*1}).
\end{align*}

\paragraph{Check $\tilde c_1(E_{1*})$.}

\begin{align*}
  \hat c_1(E_{1*}) + H_A^\downarrow \hat s'_2(V_{00}) + H_A^\downarrow \hat w_2(V_{00})
  &= \hat c_1(V_{10}, E_{1*}) + H_A^\downarrow \hat s'_2(V_{00}) + H_A^\downarrow \hat w_2(V_{00}) \\
  &= s_1(V_{10}, E_{1*}) + H_A^\downarrow \hat s'_2(V_{10}) + H_A^\downarrow \hat s'_2(V_{00}) + H_A^\downarrow \hat w_2(V_{00}) \\
  &= s_1(V_{10}, E_{1*}) + H_A^\downarrow \hat t'_2(E_{*0}) + H_A^\downarrow \hat w_2(V_{00}) \\
  &= s_1(V_{10}, E_{1*}) + H_A^\downarrow \tilde t'_2(E_{*0}) + H_A^\downarrow \hat w_2(V_{10}) \\
  &= s_1(V_{10}, E_{1*}) + H_A^\downarrow \tilde t'_2(E_{*0}) + 0 \\
  &= \tilde c_1(E_{1*}).
\end{align*}

\paragraph{Linear Time:}

It is clear that the construction at each step takes linear time.
For the steps besides Step4 where we apply the co-decoder assumption,
  we only use simple operations such as linear map or inverse of a linear map.
And Step4 takes linear time by assumption.

\section{Optimal Robust Tensor Codes} \label{sec:optimal-robust-tensor-codes}

This section shows that random codes have linear robustness with high probability.
We improve on the result in \cite{panteleev2021asymptotically,leverrier2022quantum}
  by using the idea of puncturing and a new counting argument.

Recall that $C_A$ and $C_B$ are linear codes of length $n_a$ and $n_b$ with rate $\rho_a$ and $\rho_b$. For simplicity we assume that $n_a=n_b=\Delta$.
An $s$-punctured code of $C_A$ is obtained  by first choosing $s$ coordinates $I_a \subset [n_a]$ then consider the codewords of $C_A$ restricted to $[n_a]-I_a$.
Notice that $\norm{c}_{[n_a] \times [n_b]} = \abs{c}$ is identical to the Hamming weight, so we will mostly use $\abs{c}$ in this section to simplify the notation.

We now recall \Cref{thm:random-tensor-codes-are-robust}, which is the main theorem to prove in this section.

\Robust*


The proof follows a similar strategy in \cite{panteleev2021asymptotically,leverrier2022quantum}, where the key is to show that a codeword with a small weight $|c| < w = \Theta(\Delta^2)$ is ``structured'', i.e. $c$ is only supported on a few columns and a few rows, with high probability. 

To show codewords with small weights are ``structured'', we show all non-zero codewords in a random code have some column or row with large weight with high probability. Because the punctured code of a random code is still roughly a random code, the same property also applies to its punctured codes. 
Now, since we assumed the codeword $c\in \tensorcode(C_A, C_B)$ has small weight, we can remove a few columns and rows with large weights, such that the rest have small weight in all columns and rows.
We then apply the property of the punctured code above which implies the rest is $0$, so $c$ is only supported on those removed columns and rows, i.e. $c$ is ``structured''.

When $c$ is ``structured'', one can then find $c_a$ supported on the few columns and $c_b$ supported on the few rows. 
This means the cancellation in $c_a + c_b$ could only happen in the intersection of those columns and rows which is small. 
Since each column of $c_a$ is a codeword, when the distance is large, $\abs{c_a} \ge d_1 \norm{c_a}_{[n_b]} = \Theta(\Delta) \norm{c_a}_{[n_b]}$.
This implies codewords with small weight satisfy the inequality for robustness $\abs{c} = \abs{c_a} + \abs{c_b} - \textnormal{ small number of cancellations } \ge \Theta(\Delta) (\norm{c_a}_{[n_b]} + \norm{c_b}_{[n_a]})$.

When $c$ has large weight $|c| \ge w = \Theta(\Delta^2)$,
  because $\norm{c_a}_{[n_b]} + \norm{c_b}_{[n_a]} \le 2 \Delta$,
  the inequality for robustness $\abs{c} \ge d_2 (\norm{c_a}_{[n_b]} + \norm{c_b}_{[n_a]})$ is easily satisfied 
  by setting $d_2 = w/(2\Delta) = \Theta(\Delta)$.

We now go through the steps carefully below.
We first state two lemmas, show the theorem, then prove the lemmas.
The first lemma says each non-zero codeword has at least one row or column with large weight (which implies codewords with small weight are ``structured'').
The second lemma says ``structured'' codewords satisfy robustness.

\begin{lemma} \label{lem:large-distance-after-puncture}
  Fix $\rho_a, \rho_b \in (0,1)$, let $\sigma \in (0, 1), \tau \in (0, (1-\sigma)/2)$ satisfy
  \begin{equation*}
    2 h(\sigma) + 2 (1-\sigma)h(\frac{\tau}{1-\sigma}) < \frac{3}{4} \frac{(1 - \sigma - \rho_a)(1 - \sigma - \rho_b)}{1-\sigma}.
  \end{equation*}
  Let $C_A, C_B$ be random codes sampled from the uniform distribution with length $\Delta$ and dimensions $k_a=\rho_a \Delta, k_b=\rho_b \Delta$.
  Then as $\Delta$ goes to infinity, with probability tending to $1$, the following holds:
    for any $s = \sigma \Delta$-punctured code $C_A', C_B'$,
    all non-zero codewords in $\tensorcode(C_A', C_B')$ have at least one row or one column with weight $\ge t = \tau \Delta$.
  In other words, if a codeword in $\tensorcode(C_A', C_B')$ has all its rows and columns with weight $< t$, then the codeword is $0$.
\end{lemma}

\begin{remark}
  Previously, we said the robustness parameter can be thought of as the higher dimensional generalization of linear distance.
  Here, \Cref{lem:large-distance-after-puncture} provides an alternative generalization that could also be useful when trying to study higher dimensional tensor codes.
  Note that in one-dimension, these two definitions coincide.
\end{remark}

\begin{restatable} [Modification of {\cite[Lemma 8]{panteleev2021asymptotically}} or {\cite[Lemma 30]{leverrier2022quantum}}] {lemma}{RobustLemma} \label{lem:small-distance-is-structured}
  Suppose $C_A$ and $C_B$ have distance $d_1$.
  If $c \in \tensorcode(C_A, C_B)$ is supported on $I_a \times [n_b] \cup [n_a] \times I_b$
  and $|I_a|, |I_b| < d_1$,
  then there exists $c_a \in C_A \otimes \F2^{n_b}$ supported on $[n_a] \times I_b$ and $c_b \in \F2^{n_a} \otimes C_B$ supported on $I_a \times [n_b]$
  such that $c = c_a + c_b$.

  Furthermore, if $|I_a|, |I_b| < d_1/2$, we have
  \begin{equation*}
    \abs{c} \ge \frac{d_1}{2}(\norm{c_a}_{[n_b]} + \norm{c_b}_{[n_a]}).
  \end{equation*}
\end{restatable}

\begin{proof} [Proof of \Cref{thm:random-tensor-codes-are-robust}]
  We first show linear distance. The \Cref{eq:GV-2d} in the theorem statement implies $h(\delta_1) < 1 - \rho_a$ and $h(\delta_1) < 1 - \rho_b$ because $h(\delta_1) < 2 h(\delta_1/2)$ for $\delta_1 < 1/2$.
  By Gilbert-Varshamov bound, we know $C_A, C_B$ have distance $d_1 = \delta_1 \Delta$ with probability tending to $1$.

  Next, we prove that $(C_A,C_B)$ is $\delta_2 \Delta$-robust when $C_A, C_B$ have distance $d_1$ and the condition in \Cref{lem:large-distance-after-puncture} holds.
  More precisely, because the inequality in \Cref{lem:large-distance-after-puncture} holds for $\sigma = \delta_1/2$ and $\tau = 4 \delta_2/\delta_1$,
  we have that for all $s = \sigma \Delta$-punctured codes $C_A', C_B'$, if a codeword in $\tensorcode(C_A', C_B')$ has all columns and rows with weight $< t = \tau \Delta$,
  then the codeword is $0$.
  Note that the same holds for all punctured codes $C_A', C_B'$ with puncture $\le s$.
  
  Given $c \in \tensorcode(C_A, C_B)$, we consider the following two cases. First, assume that $c$ has small weight $|c| < d_1 t / 2$. We can remove $d_1/2-1=s-1$ columns $I_a$ and $s-1$ rows $I_b$ such that $c_{([n_a]-I_a)\times([n_b]-I_b)}$ has weight less than $t$ for each column and each row. 
  By \Cref{lem:large-distance-after-puncture}, 
    $c_{([n_a]-I_a)\times([n_b]-I_b)}$ has to be $0$.
  So $c$ is supported on at most $s-1$ columns and $s-1$ rows. 
  Now, because $s - 1 < s = d_1/2$, by \Cref{lem:small-distance-is-structured}
    \begin{equation*}
      \abs{c} \ge \frac{d_1}{2} (\norm{c_a}_{[n_b]} + \norm{c_b}_{[n_a]}).
    \end{equation*}

  Second, suppose $c$ has large weight $|c| \ge d_1 t / 2$. 
  Because $c \in \tensorcode(C_A, C_B)$, we can always write $c = c_a + c_b$. 
  For any choice of $c_a, c_b$, we have $\norm{c_a}_{[n_b]} + \norm{c_b}_{[n_a]} \le 2 \Delta$.
  Together with the fact that $c$ has large weight, we have
    \begin{equation*}
      \abs{c} \ge \frac{d_1 t}{4 \Delta} (\norm{c_a}_{[n_b]} + \norm{c_b}_{[n_a]}).
    \end{equation*}

  Therefore, the code pair $(C_A, C_B)$ is $\min(\frac{d_1}{2}, \frac{d_1 t}{4 \Delta}) = \frac{d_1 t}{4 \Delta} = \delta_2 \Delta$ robust.
\end{proof}

We now prove the first lemma.
The proof is similar to the proof of the Gilbert-Varshamov bound where we count the number of bad events and take the union bound.
Here, we organize the bad events by the rank of the matrices.
We first count the number of bad matrices with rank $r$.
Then, we compute the probability for a rank $r$ matrix to be a codeword. Finally, we take the union bound.

We again first state the two claims, show the lemma, then prove the claims.

\begin{claim} [Number of bad matrices is small] \label{claim:number-of-bad-matrix-is-small}
  Let $\cM(\Delta, r, t)$ be the set of matrices in $\F2^{\Delta \times \Delta}$ with rank $r$ where each row and each column has weight $< t$. Assume $t \le \Delta/2$.
  Then 
  \begin{equation*}
    \abs{\cM(\Delta, r, t)} \le 2^{2r\Delta h(\frac{t}{\Delta}) + 2\Delta h(\frac{r}{\Delta})}.
  \end{equation*}
\end{claim}

\begin{claim} [A bad matrix is likely not a codeword] \label{claim:bad-matrix-is-likely-not-codeword}
  Let $c \in \F2^{\Delta \times \Delta}$ be a matrix of rank $r$.
  Then 
  \begin{equation*}
    \cP(\Delta, r, k_a, k_b) \coloneqq \Pr[c \in \Sigma(C_A, C_B)] \le 512 (r+1) 2^{-\frac{3}{4} \frac{(\Delta-k_a) (\Delta-k_b)}{\Delta} r},
  \end{equation*}
  where $C_A, C_B$ are uniformly sampled from linear codes of length $\Delta$ and dimension $k_a, k_b$.
\end{claim}

\begin{proof} [Proof of \Cref{lem:large-distance-after-puncture}]
  
  To show the lemma, we apply the union bound over the bad events.
  First, there are $\binom{\Delta}{s}^2$ possible $s$-punctured code pairs $C_A', C_B'$ when puncturing both $C_A$ and $C_B$.
  Next, there are $\abs{\cM(\Delta-s, r, t)}$ many rank $r$ matrices where all rows and columns have weight $<t$.
  Finally, each of those rank $r$ matrix is a codeword with probability $\le \cP(\Delta-s, r, k_a, k_b)$ which we explain below.
  Recall that $C_A$ is a uniform distribution over dimension $k_a$,
    so $C_A'$ is a mixture of uniform distribution over dimensions $k'_a \le k_a$.
  Suppose $C_A'$($C_B'$) has dimension $k'_a$($k'_b$) with probability $p_{k'_a}$($q_{k'_b}$),
    then a rank $r$ matrix is a codeword with probability $\le \sum_{k'_a=0}^{k_a} \sum_{k'_b=0}^{k_b} p_{k'_a} q_{k'_b} \cP(\Delta-s, r, k'_a, k'_b)$.
  Because smaller dimension leads to smaller probability of becoming a codeword, $\cP(\Delta-s, r, k'_a, k'_b) \le \cP(\Delta-s, r, k_a, k_b)$,
    the probability for a rank $r$ matrix to be a codeword $\le \cP(\Delta-s, r, k_a, k_b)$ as claimed.
  
  Therefore, the probability where the main condition fails to hold (there are some $s$-punctured code pair $C'_A, C'_B$ with non-zero codeword in $\tensorcode(C_A', C_B')$
  where all its rows and columns has weight $< t$) is
  \begin{align*}
    &\le \binom{\Delta}{s}^2 \sum_{r=1}^{\Delta-s} \abs{\cM(\Delta-s, r, t)}\, \cP(\Delta-s, r, k_a, k_b) \\
    &\le 2^{2 \Delta h(\frac{s}{\Delta})} \cdot 
      \sum_{r=1}^{\Delta-s} 2^{2r(\Delta-s) h(\frac{t}{\Delta-s}) + 2(\Delta-s) h(\frac{r}{\Delta-s})} \cdot
      512 (r+1) 2^{-\frac{3}{4} \frac{(\Delta-s-k_a) (\Delta-s-k_b)}{\Delta-s} r}
  \end{align*}
  where the last inequality uses $\binom{\Delta}{s} \le 2^{\Delta h(\frac{s}{\Delta})}$ and the results from \Cref{claim:number-of-bad-matrix-is-small} and \Cref{claim:bad-matrix-is-likely-not-codeword}.

  Finally, we find the requirements for the value to approach $0$ as $\Delta$ goes to infinity.
  This boils down to two conditions.
  The first is 
    \begin{equation*}
      2 \frac{\Delta - s}{\Delta} h(\frac{t}{\Delta-s}) < \frac{3}{4} \frac{(\Delta - s - k_a)(\Delta - s - k_b)}{\Delta (\Delta - s)}
    \end{equation*}
    so that large $r$ has exponentially decaying contribution.
  The second is 
    \begin{equation*}
      2h(\frac{s}{\Delta}) + 2\frac{\Delta-s}{\Delta} h(\frac{t}{\Delta-s}) < \frac{3}{4} \frac{(\Delta - s - k_a)(\Delta - s - k_b)}{\Delta (\Delta - s)}
    \end{equation*}
    so that the contribution for $r=1$ approaches $0$.
  It is clear that we can reduce the two inequalities to just the second one.
  By rewriting the variables $s = \sigma \Delta, t = \tau \Delta, k_a = \rho_a \Delta, k_b = \rho_b \Delta$, we see that if
    \begin{equation*}
      2 h(\sigma) + 2 (1-\sigma)h(\frac{\tau}{1-\sigma}) < \frac{3}{4} \frac{(1 - \sigma - \rho_a)(1 - \sigma - \rho_b)}{1-\sigma}
    \end{equation*}
    as $\Delta$ goes to infinity, with probability tending to $1$,
    for any $s = \sigma \Delta$-punctured code $C_A', C_B'$,
    all non-zero codewords in $\tensorcode(C_A', C_B')$ have at least one row or one column with weight $\ge t = \tau \Delta$.
\end{proof}

\begin{proof} [Proof of \Cref{claim:number-of-bad-matrix-is-small}]
  The idea is that knowing the entries of specific $r$ rows and $r$ columns of a rank $r$ matrix $c$ is sufficient to determine rest of the entries.
  So to upper bound the number of rank $r$ matrices $c$ where each row and each column has weight $< t$,
    we suffice to upper bound the number of  configuration supported on $r$ rows and $r$ columns,
    $c_{I_a \times [n_b] \cup [n_a] \times I_b}$,
    such that each row and each column has weight $< t$,
    where $I_a$ and $I_b$ are the indices of some $r$ rows and $r$ columns.
  We discuss the details below.

  Given $c \in \cM(\Delta, r, t)$ of rank $r$ matrices where each row and each column has weight $< t$. 
  Because $c$ has rank $r$, there exists a $r \times r$ submatrix with rank $r$, i.e. full rank.
  Let the $r$ rows and $r$ columns be indexed by $I_a \subset [n_a]$ and $I_b \subset [n_b]$.
  We now claim that knowing the value of $c$ on $I_a \times [n_b] \cup [n_a] \times I_b$ is enough to uniquely determine $c$ through the following identity:
  \begin{equation*}
    c_{i_a, i_b} = c_{I_a \times i_b} c_{I_a \times I_b}^{-1} c_{i_a \times I_b} = \sum_{j_a \in I_a, j_b \in I_b} c_{j_a, i_b} (c_{I_a \times I_b})^{-1}_{j_a, j_b} c_{i_a, j_b}
  \end{equation*}
  where $c_{I_a \times I_b}^{-1}$ is the matrix inverse.

  Because $c$ has rank $r$, the vector $c_{(i_a \cup I_a) \times i_b}$ is a linear combination of the vectors $c_{(i_a \cup I_a) \times j_b}$ for $j_b \in I_b$.
    Say $c_{(i_a \cup I_a) \times i_b} = \sum_{j_b \in I_b} v_{j_b} c_{(i_a \cup I_a) \times j_b}$. 
    Using the coordinates of $I_a$, $c_{I_a \times i_b} = \sum_{j_b \in I_b} v_{j_b} c_{I_a \times j_b}$, one has $v = c_{I_a \times i_b} c_{I_a \times I_b}^{-1}$.
    When plug $v$ back to $c_{i_a, i_b} = \sum_{j_b \in I_b} v_{j_b} c_{i_a, j_b}$, we obtain the desired identity.

  \begin{figure}[H]
    \centering
    \begin{tikzpicture}
      \draw (0,0) rectangle (6,-6);
      \draw (2,0) rectangle (2,-6);
      \draw (3,0) rectangle (3,-6);
      \draw (0,-2) rectangle (6,-2);
      \draw (0,-3) rectangle (6,-3);

      \draw (2.5,-2.5)node{$c_{i_a, i_b}$};
      \draw (1,-2.5)node{$c_{i_a \times I_b}$};
      \draw (2.5,-1)node{$c_{I_a \times i_b}$};
      \draw (1,-1)node{$c_{I_a \times I_b}$};

      \path (0,0) --node[auto]{$I_b$} (2,0) --node[auto]{$i_b$} (3,0);
      \path (0,0) --node[auto,swap]{$I_a$} (0,-2) --node[auto,swap]{$i_a$} (0,-3);

      \draw (3,.7)node[rotate=90]{$\left.\rule{0pt}{\dimexpr 3cm}\right\}$};
      \draw (-.7,-3)node{$\left\{\rule{0pt}{\dimexpr 3cm}\right.$};
      \draw (3,1.3)node{$[n_b]$};
      \draw (-1.3,-3)node{$[n_a]$};

      \clip (0,0) -- (6,0) -- (6,-2) -- (2,-2) -- (2,-6) -- (0,-6) -- cycle;
      \foreach\x in{-6, -5.5, ..., 6}{
        \draw (\x, 0) -- +(6, -6);
      }
    \end{tikzpicture}
    \caption{Given $c$ of rank $r$. If $c_{I_a \times I_b}$ has rank $r$, then $c$ is uniquely determined by the value on the shaded region $c_{I_a \times [n_b] \cup [n_a] \times I_b}$.}
    \label{fig:rank-r-matrix}
  \end{figure}
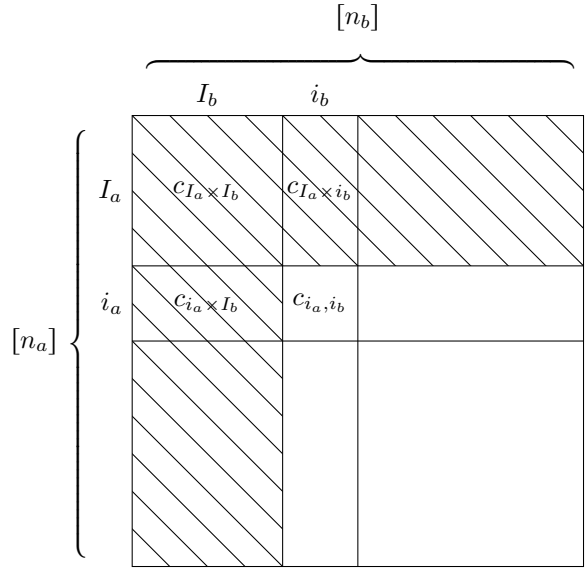

  So to bound the size of $\cM(\Delta, r, t)$, we suffice to bound the number of choices of $I_a$ and $I_b$ and the number of choices of $c_{I_a \times [n_b] \cup [n_a] \times I_b}$.

  We first bound the number of choices of $I_a$ and $I_b$. $I_a$ and $I_b$ each has size $r$ and is a subset of $[\Delta]$, so the number of choices 
  \begin{equation*}
    \le \binom{\Delta}{r}^2 \le 2^{2\Delta h(\frac{r}{\Delta})}
  \end{equation*}
  where we use the inequality $\binom{\Delta}{r} \le 2^{\Delta h(\frac{r}{\Delta})}$.

  Next, we bound the number of choices of $c_{I_a \times [n_b] \cup [n_a] \times I_b}$ given $I_a$ and $I_b$.
    It is enough to find $r$ columns and $r$ rows of weight $< t$,
    so the number of choices 
    \begin{equation*}
      \le \left(\sum_{i=0}^{t-1} \binom{\Delta}{i}\right)^{2r} \le 2^{2r\Delta h(\frac{t}{\Delta})}
    \end{equation*}
    where we use the inequality $\sum_{i=0}^{t-1} \binom{\Delta}{i} \le 2^{\Delta h(\frac{t}{\Delta})}$ when $t \le \Delta/2$. 
    
  Multiply the two, we have
  \begin{equation*}
    \abs{\cM(\Delta, r, t)} \le 2^{2r\Delta h(\frac{t}{\Delta}) + 2\Delta h(\frac{r}{\Delta})}.
  \end{equation*}
\end{proof}

\begin{proof} [Proof of \Cref{claim:bad-matrix-is-likely-not-codeword}]
  Let $H_A\colon \F2^\Delta \to \F2^{m_a}$ and $H_B\colon \F2^\Delta \to \F2^{m_b}$ be the parity check matrices of $C_A$ and $C_B$ where $m_a = \Delta - k_a$ and $m_b = \Delta - k_b$.
  Using the chain rule we have
  \begin{equation*}
    \Pr[c \in \tensorcode(C_A, C_B)] = \Pr[(H_A \otimes H_B) c = 0] = \sum_{r'=0}^{\min(r,m_a)} \Pr[\rk(H_A c) = r'] \cdot \Pr[H_B (H_A c) = 0 | \rk(H_A c) = r']
  \end{equation*}
  where $\rk$ is the rank and $H_B (H_A c)$ means we first apply $H_A$ to $c$ then apply $H_B$ to $H_A c$.
  So what remains is to bound $\Pr[\rk(H_A c) = r']$ and $\Pr[H_B (H_A c) = 0 | \rk(H_A c) = r']$.

  To do so, we introduce the Gaussian binomial coefficients.
  The Gaussian binomial coefficients is $\binom{x}{y}_2 \coloneqq \frac{[x]_2!}{[y]_2![x-y]_2!}$ where $[x]_2 \coloneqq 2^x - 1$ and $[x]_2! \coloneqq \prod_{i=1}^x [i]_2$.
  It is known that $\binom{x}{y}_2$ is the number of $y$-dimensional vector subspace of an $x$-dimensional vector space over $\F2$.
  This is the primary use of the Gaussian binomial coefficients.
  
  Using the knowledge of Gaussian binomial coefficient, we first bound $\Pr[\rk(H_A c) = r']$.
  The idea is to consider all the rank $r'$ subspaces of $\F2^{m_a}$ and compute the probability that $H_A c$ is contained in one of those subspaces.
  The number of the rank $r'$ subspace is simply $\binom{m_a}{r'}_2$.
  Let $V$ be the column space of $c$, i.e. the vector space spanned by the column vectors of $c$.
  Because $c$ has rank $r$, $\dim V = r$.
  Now, we consider a rank $r'$ subspace $V'$ of $\F2^{m_a}$ and ask what is the probability that $H_A V \subset V'$, i.e.
  $V \subset H_A^{-1} V' = \set{v \in \F2^\Delta: H_A v \in V'}$.
  Because $H_A$ is full rank, $\dim H_A^{-1} V' = r' + k_a$.
  So the probability is $\binom{r'+k_a}{r}_2 \big/ \binom{\Delta}{r}_2$, i.e. the number of rank $r$ subspace in $H_A^{-1} V'$ divided by the number of rank $r$ subspace in $\F2^\Delta$.
  Overall we have,
  \begin{equation*}
    \Pr[\rk(H_A c) = r'] \le \binom{m_a}{r'}_2 \binom{r'+k_a}{r}_2 \bigg/ \binom{\Delta}{r}_2.
  \end{equation*}
  (Note that this is an inequality because $H_A V$ could have rank less than $r'$.)
  
  Now we bound $\Pr[H_B (H_A c) = 0 | \rk(H_A c) = r']$.
  Let $W$ be the row space of $H_A c$ and $W'$ be the kernel $H_B^{-1}(0)$.
  We know $\dim W = r'$, because $\rk(H_A c) = r'$, and $H_B^{-1}(0) = k_b$, because $H_B$ has full rank.
  Therefore,
  \begin{equation*}
    \Pr[H_B (H_A c) = 0 | \rk(H_A c) = r'] = \binom{k_b}{r'}_2 \bigg/ \binom{\Delta}{r'}_2.
  \end{equation*}

  Combine the above results, we have
  \begin{equation*}
    \Pr[c \in \tensorcode(C_A, C_B)]
    \le \sum_{r'=0}^{\min(r, m_a)} \left( \binom{m_a}{r'}_2 \binom{r'+k_a}{r}_2 \middle/ \binom{\Delta}{r}_2 \right) \cdot \left( \binom{k_b}{r'}_2 \middle/ \binom{\Delta}{r'}_2 \right).
  \end{equation*}
  What is left to be done is to simplify the formula.

  We first use the following bound of the Gaussian binomial coefficient.
  \begin{claim} \label{claim:gaussian-binomial}
    \begin{equation*}
      2^{y(x-y)} \le \binom{x}{y}_2 \le 8 \cdot 2^{y(x-y)}
    \end{equation*}
  \end{claim}
  \begin{proof} [Proof of \Cref{claim:gaussian-binomial}]
    When $x=y$ or $y=0$, the inequality holds trivially.
    Otherwise, we will show when $z > 0$
    \begin{equation*}
      \frac{1}{4} \cdot 2^{z(z+1)/2} \le [z]_2! \le \frac{1}{2} \cdot 2^{z(z+1)/2}
    \end{equation*}
    which implies the desired result.
    By definition
    $[z]_2! = \prod_{i=1}^z (2^i-1) = 2^{z(z+1)/2} \prod_{i=1}^z (1-2^{-i})$.
    So we suffice to bound $\prod_{i=1}^z (1-2^{-i})$.
    For the upper bound, it is clear that $\prod_{i=1}^z (1-2^{-i}) \le 1/2$.
    For the lower bound, we have
    \begin{equation*}
      1/4 = (1-2^{-1})(1-2^{-1}) \le (1-2^{-1})(1-2^{-2})(1-2^{-2})
      \le ... \le (1-2^{-z}) \prod_{i=1}^z (1-2^{-i}) \le \prod_{i=1}^z (1-2^{-i})
    \end{equation*}
    where we apply $(1-2^{-i}) \le (1-2^{-i-1})^2$ iteratively.
  \end{proof}

  The bound on the Gaussian binomial coefficient implies
  \footnote{Note that $512$ can be improved to $8$ by bounding $\binom{r'+k_a}{r}_2 / \binom{\Delta}{r}_2 \le 2^{-r(\Delta - r' - k_a)}$ 
  and $\binom{k_b}{r'}_2 / \binom{\Delta}{r'}_2 \le 2^{-r'(\Delta - k_b)}$.}
  \begin{equation*}
    \Pr[c \in \tensorcode(C_A, C_B)]
    \le 512 \sum_{r'=0}^{\min(r, m_a)} 2^{-(m_a-r')(r-r')-m_b r'}.
  \end{equation*}

  Finally, we show $-(m_a-r')(r-r')-m_b r' \le -\frac{3}{4} \frac{m_a m_b}{\Delta} r$ for $0 \le r' \le m_a$.
  We first rewrite $-(m_a-r')(r-r')-m_b r' = (m_a-r')(m_b-(r-r')) - m_a m_b$.
  Then divide both sides by $m_a m_b$.
  If $m_b - (r-r') \ge 0$ we have
  \begin{align*}
    (1-\frac{r'}{m_a})(1-\frac{r-r'}{m_b}) - 1 
    &\le (1-\frac{r'}{\Delta})(1-\frac{r-r'}{\Delta}) - 1 \\
    &= -\frac{r}{\Delta} + \frac{r'(r-r')}{\Delta^2} \\
    &\le -\frac{r}{\Delta} + \frac{\frac{1}{4}r^2}{\Delta^2} \\
    &\le -\frac{3}{4} \frac{r}{\Delta}.
  \end{align*}
  Note the first inequality uses $m_a - r' \ge 0$ and $m_b - (r-r') \ge 0$.
  If $m_b - (r-r') < 0$ we have
  \begin{equation*}
    (1-\frac{r'}{m_a})(1-\frac{r-r'}{m_b}) - 1 
    \le - 1
    \le -\frac{3}{4} \frac{r}{\Delta}.
  \end{equation*}

  Therefore,
  \begin{equation*}
    \Pr[c \in \tensorcode(C_A, C_B)] \le 512 (r+1) 2^{-\frac{3}{4} \frac{m_a m_b}{\Delta} r}.
  \end{equation*}
\end{proof}

We now study the second lemma. We first recall the statement.

\RobustLemma*

\begin{proof} [Proof of \Cref{lem:small-distance-is-structured}]
  The intuition is that one can construct $c_a$ using $c_{([n_a] - I_a) \times I_b}$ by requiring each column to be a codeword in $C_A$.
  Because each column is only missing $|I_a| < d_1$ values,   
    which is less than the distance,
    one can uniquely recover the codeword.
  This intuition is correct, but requires more steps to make it rigorous as we will do below.


  Let $H_A' = H_A|_{I_a}\colon \F2^{I_a} \to \F2^{m_a}$ be the restriction of $H_A$. 
  Because $C_A$ has distance $d_1$ and $|I_a| < d_1$, $H_A'$ is injective.
  Since $H_A'$ is an injective linear map, there exists (not unique) a left inverse linear map $J_A'\colon \F2^{m_a} \to \F2^{I_a}$, such that
  \begin{equation*}
    J_A' H_A' = I_{I_a}
  \end{equation*}
  where $I_{I_a}$ is the identity for $\F2^{I_a}$.
  Similarly, we can define $H_B'$ and $J_B'$.

  Using $J_A'$ we can explicitly write down a recovery map for the erased codeword.
  \begin{claim} \label{claim:vector-recover}
    For any codeword $v \in C_A \subset \F2^{n_a}$,
    one can recover $v$ from $v_{[n_a]-I_a}$ through the following linear map
    \begin{equation*}
      v = (J_A' H_A|_{[n_a]-I_a} \Vert I_{[n_a] - I_a}) v_{[n_a]-I_a}
    \end{equation*}
    where $\Vert$ is concatenation.
  \end{claim}
  Notice that $J_A' H_A|_{[n_a]-I_a} \colon \F2^{[n_a] - I_a} \to \F2^{I_a}$ and $I_{[n_a] - I_a} \colon \F2^{[n_a] - I_a} \to \F2^{[n_a] - I_a}$, so \\ $J_A' H_A|_{[n_a]-I_a} \Vert I_{[n_a] - I_a} \colon \F2^{[n_a] - I_a} \to \F2^{[n_a]}$.
  \begin{proof} [Proof of \Cref{claim:vector-recover}]
    Because $v \in C_A$ is a codeword, we have $H_A v = 0$,
    so $H_A' v_{I_a} = H_A|_{[n_a]-I_a} v_{[n_a]-I_a}$.
    Multiply both side by $J_A'$, we obtain $J_A' H_A' v_{I_a} = v_{I_a} = J_A' H_A|_{[n_a]-I_a} v_{[n_a]-I_a}$ which is the desired result.
  \end{proof}

  This suggests one can recover $c_a$ from $c_{([n_a]-I_a) \times I_b}$
    by mapping each column supported on $[n_a]-I_a$ to the full codeword supported on $[n_a]$.
  However, we need to first show each $c_{([n_a]-I_a) \times i_b}$ is a truncation of a codeword in $C_A$ which we will do now.
  Because $H_A H_B c = 0$, each column of $H_B c$ is a codeword in $C_A$.
  Because $J_B' H_B' = I_{I_b}$, $c_{([n_a]-I_a) \times I_b} = J_B' H_B' c_{([n_a]-I_a) \times I_b} = (J_B' (H_B c))|_{([n_a]-I_a) \times [m_b]}$,
    indeed each column of $c_{([n_a]-I_a) \times I_b}$ is a truncation of codeword.

  This leads to the following choice of $c_a$ and $c_b$
  \begin{equation*}
    c_a = ((J_A' H_A|_{[n_a] - I_a} \Vert I_{[n_a] - I_a}) c_{([n_a]-I_a) \times I_b}) \big\Vert 0_{[n_a] \times ([n_b] - I_b)}
  \end{equation*}
  and
  \begin{equation*}
    c_b = ((J_B' H_B|_{[n_b] - I_b} \Vert I_{[n_b] - I_b}) c_{I_a \times ([n_b]-I_b)}) \big\Vert 0_{([n_a] - I_a) \times [n_b]}.
  \end{equation*}

  By construction $c_a$ is supported on $[n_a] \times I_b$ and $c_b$ is supported on $I_a \times [n_b]$.
  And from the discussion we know $c_a \in C_A \otimes \F2^{n_b}$ and $c_b \in \F2^{n_a} \otimes C_B$.
  What is left to be shown is $c = c_a + c_b$.

  By construction we know $c$ agrees with $c_a + c_b$ on everywhere outside of $I_a \otimes I_b$,
  that is $c' \coloneqq c + c_a + c_b$ is supported on $I_a \otimes I_b$.
  Now, $H_B c'$ is supported on $I_a \otimes [m_b]$ and each column is a codeword in $C_A$ (because $H_A (H_B c') = 0$).
  Because the distance of $C_A$ is $d_1$ and $|I_a|<d_1$, we have $H_B c' = 0$.
  Similarly, $c'$ is supported on $I_a \otimes I_b$ and each row is a codeword in $C_B$ (because $H_B c' = 0$).
  Because the distance of $C_B$ is $d_1$ and $|I_b|<d_1$, we have $c' = 0$.
  This concludes the first half of the lemma.

  The second half is simple. 
  Since each column of $c_a$ is a codeword, it is either $0$ or has weight $\ge d_1$.
  For those with weight $\ge d_1$, after removing the indices in $I_a$ it still have weight $> d_1/2$.
  That is 
  \begin{equation*}
    \abs{c_{([n_a] - I_a) \times I_b}} \ge \frac{d_1}{2} \norm{c_a}_{[n_b]}.
  \end{equation*}
  Similarly,
  \begin{equation*}
    \abs{c_{I_a \times ([n_b] - I_b)}} \ge \frac{d_1}{2} \norm{c_b}_{[n_a]}.
  \end{equation*}
  So 
  \begin{equation*}
    \abs{c} \ge \frac{d_1}{2} (\norm{c_a}_{[n_b]} + \norm{c_b}_{[n_a]}).
  \end{equation*}
\end{proof}



\bibliographystyle{unsrt}
\bibliography{references}


\end{document}